\documentclass[journal]{IEEEtran}

\usepackage{amsmath,amssymb,amsfonts,amsthm,mathtools}
\usepackage{bm}
\usepackage{mathrsfs}

\newcommand{\bTmat}{\mathbf{T}}
\newcommand{\bS}{\mathbf{S}}

\newcommand{\bW}{\mathbf{W}}
\newcommand{\bu}{\mathbf{u}}
\newcommand{\ba}{\mathbf{a}}
\newcommand{\bd}{\mathbf{d}}

\usepackage{graphicx}
\usepackage{float}
\usepackage{placeins}

\usepackage[caption=false,font=footnotesize]{subfig}

\usepackage{booktabs}
\usepackage{multirow}

\usepackage{algorithm}
\usepackage{algpseudocode}

\algrenewcommand\algorithmicindent{1.15em}

\usepackage{enumitem}
\usepackage{xcolor}
\usepackage{cite}
\usepackage{url}

\usepackage[
    colorlinks=true,
    linkcolor=blue,
    citecolor=blue,
    urlcolor=blue
]{hyperref}

\setlist[itemize]{
    leftmargin=1.25em,
    itemsep=1pt,
    topsep=2pt,
    parsep=0pt,
    partopsep=0pt
}

\setlist[enumerate]{
    leftmargin=1.45em,
    itemsep=1pt,
    topsep=2pt,
    parsep=0pt,
    partopsep=0pt
}

\allowdisplaybreaks

\newtheorem{assumption}{Assumption}
\newtheorem{definition}{Definition}
\newtheorem{lemma}{Lemma}
\newtheorem{proposition}{Proposition}
\newtheorem{theorem}{Theorem}

\newtheorem{remark}{Remark}
\newcommand{\bpsi}{\bm{\psi}}

\usepackage{algorithm}
\usepackage{algpseudocode}

\DeclareMathOperator{\diag}{diag}
\DeclareMathOperator{\col}{col}

\DeclareMathOperator{\sgn}{sgn}

\newcommand{\bx}{\bm{x}}

\newcommand{\bz}{\bm{z}}

\newcommand{\bphi}{\bm{\phi}}

\renewcommand{\arraystretch}{0.96}

\usepackage{orcidlink}

\usepackage{xcolor}
\usepackage{tikz}
\usepackage{scalerel}
\usepackage{hyperref}

\newcommand{\orcidicon}[1]{%
\href{https://orcid.org/#1}{%
\mbox{%
\scalerel*{%
\begin{tikzpicture}[yscale=-1,transform shape]
\pic{orcidlogo};
\end{tikzpicture}%
}{|}%
}%
}%
}

\begin{document}
\title{%
\resizebox{1.02\textwidth}{!}{%
\bfseries
\begin{tabular}{c}
Fixed-Time Resilient Integral Reinforcement Learning for Input-Constrained\\[-1.0pt]
Unknown Nonlinear Systems Under FDI Attacks and Disturbances:\\[-1.0pt]
A Data-Driven Admissible Warm Start
\end{tabular}%
}%
}
% \author{
% Tien Dat Vu~\orcidicon{0009-0004-0792-000X}\,\IEEEmembership{Student Member~IEEE,}
% Minh Doan~\orcidicon{0000-0002-9460-4256}\,\IEEEmembership{Member~IEEE}
% \thanks{T. D. Vu and M. Doan are with the Faculty of Mechanical Engineering,
% Ho Chi Minh City University of Technology (HCMUT), Vietnam National University
% Ho Chi Minh City (VNU-HCM), Ho Chi Minh City, Vietnam
% (e-mail: dat.vuv@hcmut.edu.vn; minh.doan@hcmut.edu.vn).}

% \thanks{Corresponding author: N. M. Doan (e-mail: minh.doan@hcmut.edu.vn).}
% }

\author{
Tien Dat Vu and
Minh Doan
\thanks{T. D. Vu and M. Doan are with the Faculty of Mechanical Engineering,
Ho Chi Minh City University of Technology (HCMUT), Vietnam National University
Ho Chi Minh City (VNU-HCM), Ho Chi Minh City, Vietnam
(e-mail: dat.vuv@hcmut.edu.vn; minh.doan@hcmut.edu.vn).}

\thanks{Corresponding author: N. M. Doan (e-mail: minh.doan@hcmut.edu.vn).}
}

% \markboth{IEEE Transactions on Control Systems Technology}%
% {How to Use the IEEEtran \LaTeX \ Templates}
\markboth{Preprint}%
{How to Use the IEEEtran \LaTeX \ Templates}
\maketitle

\begin{abstract}
This paper develops a resilient learning controller for unknown nonlinear systems operating under actuator limits, false-data-injection attacks, and external disturbances. The key idea is to learn a saturated secure policy directly from finite trajectory data while guaranteeing that both the learning error and the closed-loop state converge to compact neighborhoods within a uniform fixed time independent of initial conditions. An integral formulation removes the unknown drift from the implementable learning law, while stored informative data sustain learning after online excitation fades. To mitigate the closed-loop sensitivity to arbitrary critic initialization, pre-deployment data, which may also be reused from the replay stack, are lifted through a finite-dimensional Koopman representation to construct a stabilizing initial policy, whose inverse saturated-policy map provides a data-driven critic-weight warm start. The resulting controller preserves input constraints by construction and guarantees practical fixed-time robustness under persistent attacks and disturbances. The proposed learning and initialization architecture is further verified through a two-link robot stabilization example, where the results demonstrate rapid state recovery, bounded critic learning, reliable actuator-constraint satisfaction, and improved closed-loop behavior under informed critic initialization.
\end{abstract}

% ============================================================
% Index terms
% ============================================================
\begin{IEEEkeywords}
False-data-injection attacks, fixed-time
stability, Hamilton--Jacobi--Isaacs equation, input-constrained
nonlinear systems, integral reinforcement learning.
\end{IEEEkeywords}

\section{Introduction}
\label{sec:introduction}

Learning optimal control for unknown nonlinear systems becomes particularly difficult when actuator limits, malicious inputs, disturbances, and convergence-time requirements must be handled simultaneously. In safety-critical settings, false-data-injection (FDI) attacks and disturbances directly affect the plant, saturation invalidates unconstrained designs, and asymptotic stabilization is insufficient when recovery must occur within a predictable time. This paper develops a secure differential-game framework that embeds actuator constraints into the optimal-control problem, removes the unknown drift from the implementable learning law through an integral Bellman--Isaacs identity, retains informative transient data via finite experience replay, and shapes both critic and closed-loop dynamics toward practical fixed-time behavior. The resulting single-critic architecture preserves input admissibility by construction and avoids persistent excitation during regulation.

The framework is rooted in Hamilton--Jacobi--Isaacs (HJI) theory, where the controller acts against maximizing adversarial channels and robustness admits an \(H_\infty\)-type interpretation \cite{BasarBernhard1995,AbuKhalafLewis2005,AbuKhalafLewisHuang2006}. For unknown nonlinear drift, adaptive dynamic programming and reinforcement learning approximate optimal policies \cite{VrabieLewis2009,LewisVrabie2009,VamvoudakisLewis2010}, including robust actor--critic and identifier--critic schemes \cite{Bhasin2013,BianJiang2014,ModaresLewisJiang2015,LuoWuHuang2015}. Integral reinforcement learning avoids explicit drift evaluation through finite-interval Bellman identities \cite{VrabieLewis2009,Modares2014}, while experience replay preserves informative transient directions as excitation decays near the target \cite{Lin1992,Modares2014,Mnih2015}.

Convergence-time predictability adds another challenge. Finite-time stability allows settling time to depend on the initial condition, whereas fixed-time stability provides an initial-condition-independent upper bound \cite{BhatBernstein2000,Polyakov2012,MoulayPerruquetti2006}, useful when guaranteed recovery or completion times are required \cite{SanchezTorres2018PredefinedTime,Cao2022PrescribedTimeTracking} and motivating fixed-time control of uncertain nonlinear systems \cite{WangLai2020FixedTimeControl}. Meanwhile, FDI-resilient control has mainly focused on detection, estimation, and resilient-control methods \cite{Pasqualetti2013,Fawzi2014,Teixeira2015,Dibaji2019}. Secure learning, saturation, finite-data excitation, and fixed-time recovery therefore remain largely separated.

Critic initialization is also critical. Classical policy-iteration and HJB/HJI learning generally require an initial admissible or stabilizing policy to ensure well-posed policy evaluation and confinement within the value-function approximation region \cite{AbuKhalafLewis2005,VrabieLewis2009,VamvoudakisLewis2010}, motivating initialization schemes such as homotopy-based policy iteration \cite{ChenLewisXie2024}. Since \(\hat{\mathbf W}(0)\) immediately affects the value-function gradient and applied policy, an initial-condition-independent critic convergence bound does not guarantee an admissible transient under arbitrary initialization, although initial weights are often chosen from prior knowledge, trial-and-error, or auxiliary model-based designs.

When pre-deployment data are available, Koopman methods provide finite-dimensional lifted representations for approximately linear prediction and control synthesis \cite{WilliamsKevrekidisRowley2015,BruntonEtAl2016,KordaMezic2018}, with growing links to optimal and reinforcement learning \cite{KrolickiSutavaniVaidya2022}. However, their use as a certified bridge from collected trajectories to an admissible \emph{critic-space initialization} for an otherwise model-free integral HJI learner remains underdeveloped. Here, pre-deployment or compatible replay data construct a stabilizing lifted policy, which is mapped through the inverse saturated-policy relation to obtain a physically meaningful warm start \(\hat{\mathbf W}(0)\).

The remaining gap is thus to jointly ensure admissible initialization, unknown-drift-free learning, finite-data informativity, actuator-constrained secure control, and initial-condition-independent practical fixed-time regulation. Accordingly, this paper provides an end-to-end architecture linking pre-deployment data, admissible initialization, online optimal learning, and resilient fixed-time closed-loop regulation.

The contributions are as follows. First, a secure HJI formulation is developed for unknown nonlinear systems under matched FDI attacks, additive disturbances, and symmetric actuator constraints; a nonquadratic input penalty yields a smooth saturated policy, while the state cost induces the two-power structure required for fixed-time analysis. Second, a critic-only integral reinforcement-learning architecture recovers the control, attack, and disturbance policies from a single learned value function; the integral Bellman--Isaacs identity removes explicit drift dependence, while finite-data replay preserves informative directions without persistent excitation. Third, a two-power critic update guarantees practical fixed-time convergence of the critic-weight error under bounded approximation residuals. Fourth, a unified stability analysis propagates critic error through the learned saturated policy and proves practical fixed-time boundedness of the nonlinear closed loop under actuator constraints, FDI attacks, and external disturbances. Fifth, finite pre-deployment trajectories identify a local Koopman-lifted representation and a stabilizing bounded policy, whose inverse saturated-policy relation yields a data-driven critic warm start \(\hat{\mathbf W}(0)\) without replacing the original nonlinear integral Bellman--Isaacs learning problem.

% ============================================================
\section{Preliminaries and Problem Formulation}
\label{sec:prelim_problem}
% ============================================================

% ------------------------------------------------------------
\subsection{Preliminaries}
\label{subsec:preliminaries}

\textbf{Notation:}
Let \(\mathbb{R}\), \(\mathbb{R}_{\geq0}\), and \(\mathbb{R}_{>0}\) denote the real, nonnegative-real, and positive-real sets, respectively, and \(\mathbb{R}^{n\times m}\) the set of real \(n\times m\) matrices. For \(\boldsymbol{x}\in\mathbb{R}^n\) and \(\boldsymbol{A}\in\mathbb{R}^{n\times m}\), \(\|\boldsymbol{x}\|\) denotes the Euclidean norm, \(\|\boldsymbol{A}\|\) the induced \(2\)-norm, \(\boldsymbol{A}^{\top}\) the transpose, and \(\boldsymbol{I}_n\) the identity matrix. For symmetric \(\boldsymbol{A}\), \(\lambda_{\min}(\boldsymbol{A})\) and \(\lambda_{\max}(\boldsymbol{A})\) denote its minimum and maximum eigenvalues. The operators \(\operatorname{col}(\cdot)\) and \(\operatorname{diag}(\cdot)\) denote column stacking and diagonal construction. The symbol \(\mathbf{0}_n\) denotes the zero vector in \(\mathbb{R}^n\), while \(\mathbf{0}_{m\times n}\) denotes the \(m\times n\) zero matrix. For \(V:\mathbb{R}^n\to\mathbb{R}\), \(\nabla V(\boldsymbol{x})\) denotes its column gradient. The space \(C^k(\Omega)\) consists of functions with continuous derivatives up to order \(k\) on \(\Omega\), while \(L_2\) and \(L_\infty\) denote the square-integrable and essentially bounded signal spaces, respectively. The symbols \(\mathcal{K}\) and \(\mathcal{K}_\infty\) denote the standard classes of strictly increasing comparison functions, with \(\mathcal{K}_\infty\) additionally unbounded. The function \(\operatorname{sgn}(\cdot)\) denotes the sign function; for vector arguments, \(\tanh(\cdot)\), \(\tanh^{-1}(\cdot)\), and \(\operatorname{sgn}(\cdot)\) are applied elementwise.

\begin{theorem}[Fixed-time value-function condition \cite{Polyakov2012,Gong2025}]
\label{thm:fixed_time_value_condition}
Consider an absolutely continuous trajectory \(\bz(t)\) and a
continuously differentiable positive definite function
\(V:\Omega_z\to\mathbb R_{\geq0}\), with \(V(\bm 0)=0\). Suppose
that there exist constants
\begin{equation}
    a,b>0,
    \qquad
    0<\rho<1<\theta,
    \qquad
    \Delta_V\geq0,
\label{eq:value_condition_parameters}
\end{equation}
such that, along the trajectory,

\begin{equation}
    \dot V(\bz)
    \leq
    -aV^\rho(\bz)
    -bV^\theta(\bz)
    +\Delta_V .
\label{eq:fixed_time_value_condition}
\end{equation}
Then \(V(\bz(t))\) is practically fixed-time stable with respect
to
\begin{equation}
\begin{aligned}
    \Omega_V(\vartheta)
    =
    \bigg\{
        \bz\in\Omega_z:\;
        aV^\rho(\bz)+bV^\theta(\bz)
        \leq
        \frac{\Delta_V}{\vartheta}
    \bigg\},
\end{aligned}
\label{eq:value_residual_set}
\end{equation}
where \(\vartheta\in(0,1)\). The corresponding entering time
satisfies
\begin{equation}
    T_V
    \leq
    \frac{1}{(1-\vartheta)a(1-\rho)}
    +
    \frac{1}{(1-\vartheta)b(\theta-1)},
\label{eq:value_settling_time}
\end{equation}
which is independent of \(V(\bz(0))\). If \(\Delta_V=0\), then
the convergence is exact fixed-time convergence to the origin.
\end{theorem}

% \begin{proof}
% Let \(Y(t)=V(\bz(t))\). Since \(V\) is continuously
% differentiable and \(\bz(t)\) is absolutely continuous,
% \(Y(t)\) is absolutely continuous. Equation
% \eqref{eq:fixed_time_value_condition} satisfies the conditions
% of Lemma~\ref{lem:fixed_time}. The result therefore follows
% directly.
% \end{proof}

\subsection{Problem Formulation}
\label{subsec:problem_formulation}
% ------------------------------------------------------------

Consider the continuous-time nonlinear system
\begin{equation}
    \dot{\bx}
    =
    f(\bx)
    +
    \mathbf g(\bx)(\bu+\ba)
    +
    \bd,
\label{eq:plant}
\end{equation}
where \(\bx\in\mathbb R^n\), \(\bu\in\mathbb R^m\), \(\ba\in\mathbb R^m\), and \(\bd\in\mathbb R^n\) denote the state, intended control input, matched false-data-injection (FDI) signal, and external disturbance, respectively. The vector field \(f:\mathbb R^n\to\mathbb R^n\) represents the drift dynamics, while \(\mathbf g:\mathbb R^n\to\mathbb R^{n\times m}\) denotes the input effectiveness matrix. Since the FDI signal is matched with the control channel, the actuator receives \(\bu+\ba\) rather than the intended command \(\bu\), whereas \(\bd\) enters the system additively. Throughout this work, the control and learning problems are considered
on a prescribed compact operating region
\(\Omega\subset\mathbb R^n\) containing the origin. The compactness of
\(\Omega\) specifies a finite state domain over which the regularity,
admissibility, and function-approximation conditions invoked in the
subsequent HJI and learning analysis are required to hold. In
particular, continuous functions defined on \(\Omega\) attain finite
bounds, a property that will be used repeatedly in the subsequent
analysis.

\begin{remark}[FDI location and actuator constraint]
The constraint \(|u_i|<\bar u_i\) limits only the secure control
component generated by the controller. Two actuator-side FDI
architectures must be distinguished. If \(\ba\) is a matched malicious
input injected after the constrained action \(\bu\) has been generated,
then the plant receives \(\bu+\ba\) and the model above applies
directly. In this case, \(|a_i|\leq\bar a_i\) implies only
\(|u_i+a_i|\leq\bar u_i+\bar a_i\), not
\(|u_i+a_i|<\bar u_i\). Conversely, if \(\ba\) corrupts the command
before physical actuator saturation, the plant-input channel would
instead contain the componentwise saturated signal
\(\operatorname{sat}(\bu+\ba)\), and the differential game and saddle
policies developed below would generally require reformulation. The
present work adopts the former architecture.
\end{remark}

\begin{definition}[Admissible secure policy]
\label{def:admissible_secure_policy}
A continuous policy \(\bu_c:\Omega\rightarrow\mathcal U\) is admissible
on \(\Omega\), denoted by \(\bu_c\in\Psi(\Omega)\), if
\(\bu_c(\bm0)=\bm0\), the resulting closed-loop system is
asymptotically stable on \(\Omega\) in the absence of exogenous signals,
and the associated infinite-horizon cost is finite.
\end{definition}

\begin{remark}[Operating-region limitation and data-driven admissible initialization]
\label{rem:operating_domain}
\label{rem:data_driven_admissible_initialization}
A principal structural limitation of the present HJI--ADP formulation is that its learning and approximation guarantees are established on a prescribed compact operating region \(\Omega\), rather than globally on \(\mathbb R^n\). Such domain-dependent formulations are standard in nonlinear HJB/ADP, where approximation is performed over a stabilizing region or learning is initialized by an admissible policy \cite{AbuKhalaf2005SaturatingActuatorsHJB,Vamvoudakis2010OnlineActorCritic}. The present work does not separately design a nominal controller that first drives the state into \(\Omega\) and guarantees its invariance; instead, a feasible \(\Omega\) is assumed to be identified a priori, with the reinforcement intervals corresponding to trajectories evolving therein. A practical architecture may employ a nominal stabilizing controller before activating learning \cite{Heydari2018StabilizingInitialPolicy}, while constructing or enlarging certified admissible regions and preserving invariance during online learning remains an important open problem \cite{Lee2015InvariantExploration}.

% Importantly, the required admissible initialization need not be analytically available. In our recent fixed-time multi-agent IRL study \cite{VuEtAl2026FixedTimeMASIRL}, Appendix~A showed that finite previously collected trajectories can be reused to identify local Koopman lifted models, construct Riccati-based stabilizing controllers, and certify their admissibility on validated compact regions through explicit Koopman-residual bounds. This suggests a broader data-reuse principle: historical data may provide not only an admissible initial policy but also a set of critic weights whose induced policies inherit a certified stabilizing margin. Accordingly, Appendix~A of the present work develops an offline Koopman-based bridge from finite trajectory data to an admissible critic-weight initialization region, while the subsequent online IRL remains governed entirely by the finite-window Bellman--Isaacs formulation and is model-free with respect to the unknown nonlinear drift. The subsequent results are therefore understood conditionally on the learning trajectory remaining in \(\Omega\); this scope is stated explicitly to avoid any circular argument in which membership in \(\Omega\) is invoked to establish its own invariance.

Importantly, the required admissible initialization need not be obtained
analytically. In our recent fixed-time multi-agent IRL study
\cite{VuEtAl2026FixedTimeMASIRL}, finite previously collected trajectories
were shown to be reusable for constructing data-driven Koopman lifted
representations and corresponding stabilizing controllers. Motivated by
this idea, Section~\ref{khoitao} of the present work establishes a direct
bridge between such a Koopman-based pre-deployment design and the online
Bellman--Isaacs learning architecture. Specifically, finite trajectory
data are first lifted through a Koopman representation to obtain a
finite-dimensional linear surrogate, from which a stabilizing control
signal is constructed. Rather than using this surrogate model during the
subsequent learning phase, the resulting control signal is mapped back
through the saturated policy parameterization in
\eqref{eq:learned_control_policy} to infer a compatible and informed
initial critic weight \(\hat{\bW}(0)\) (See~\eqref{eq:learned_control_policy}). The Koopman model therefore serves
only as an offline data-driven warm-start mechanism, whereas the ensuing
online IRL remains governed entirely by the finite-window
Bellman--Isaacs residual and is model-free with respect to the unknown
nonlinear drift. The subsequent results are understood conditionally on
the learning trajectory remaining in \(\Omega\); this scope is stated
explicitly to avoid any circular argument in which membership in
\(\Omega\) is invoked to establish its own invariance.
\end{remark}

\begin{assumption}[Adversarial-signal regularity]
\label{ass:signal_regular}
The FDI signal and external disturbance are measurable and satisfy
\begin{equation}
\begin{aligned}
\ba &\in L_2([0,\infty);\mathbb R^m)
      \cap L_\infty([0,\infty);\mathbb R^m),\\
\bd &\in L_2([0,\infty);\mathbb R^n)
      \cap L_\infty([0,\infty);\mathbb R^n).
\end{aligned}
\end{equation}
\end{assumption}

The control input is subject to the symmetric componentwise
constraint
\begin{equation}
    |u_j(t)|<\lambda,
    \qquad
    j=1,\ldots,m,
\label{eq:symmetric_constraint}
\end{equation}
where \(\lambda>0\) is known. Hence,
\begin{equation}
    \mathcal U
    =
    \left\{
        \bu\in\mathbb R^m:
        |u_j|<\lambda,\;
        j=1,\ldots,m
    \right\}.
\label{eq:U_set}
\end{equation}

% \begin{assumption}
% \label{ass:system_regular}
% The origin is an equilibrium of the unforced nominal system, that
% is, \(f(\bm 0)=\bm 0\). On a compact set
% \(\Omega\subset\mathbb R^n\) containing the origin, \(f(\bx)\)
% and \(\mathbf g(\bx)\) are locally Lipschitz and bounded. The
% input effectiveness matrix \(\mathbf g(\bx)\) is known, whereas
% the drift \(f(\bx)\) is not required to be known by the learning
% algorithm. Moreover, there exists a constant \(\bar g>0\) such
% that
% \begin{equation}
%     \|\mathbf g(\bx)\|
%     \leq
%     \bar g,
%     \qquad
%     \forall\bx\in\Omega.
% \label{eq:g_bound}
% \end{equation}
% The signals \(\ba(t)\) and \(\bd(t)\) are measurable and locally
% square integrable, and the corresponding solution of
% \eqref{eq:plant} exists uniquely on every finite time interval on
% which it remains in \(\Omega\).
% \end{assumption}

\begin{assumption}
\label{ass:system_regular}
The origin is an equilibrium of the nominal system, i.e., \(f(\bm0)=\bm0\). On the compact set \(\Omega\subset\mathbb R^n\), the vector fields \(f(\bx)\) and \(\mathbf g(\bx)\) are locally Lipschitz, \(\mathbf g(\bx)\) is known, \(f(\bx)\) is unknown, and \(\|\mathbf g(\bx)\|\le\bar g\) for some constant \(\bar g>0\).
\end{assumption}

% \begin{assumption}
% \label{ass:cost_regular}
% The state penalty
% \(Q:\Omega\to\mathbb R_{\geq0}\) is continuous and positive
% definite. There exist class-\(\mathcal K_\infty\) functions
% \(\underline q(\cdot)\) and \(\bar q(\cdot)\) such that
% \begin{equation}
%     \underline q(\|\bx\|)
%     \leq
%     Q(\bx)
%     \leq
%     \bar q(\|\bx\|),
%     \qquad
%     \forall\bx\in\Omega .
% \label{eq:Q_bounds}
% \end{equation}
% The weighting matrices satisfy
% \begin{equation}
% \begin{aligned}
%     \mathbf R
%     &=
%     \diag(r_1,\ldots,r_m)>0,\\
%     \mathbf T
%     &=
%     \mathbf T^\top>0,\\
%     \mathbf S
%     &=
%     \mathbf S^\top>0,
% \end{aligned}
% \label{eq:game_weight_conditions}
% \end{equation}
% and are known.
% \end{assumption}

% \begin{remark}
% \label{rem:general_state_penalty}
% The state penalty \(Q(\bx)\) is intentionally kept general in
% Assumption~\ref{ass:cost_regular}. This separates the basic secure
% optimal-control formulation from the subsequent fixed-time
% shaping design. When fixed-time closed-loop convergence is
% established, \(Q(\bx)\) will be specialized to contain one state
% power below two and another state power above two. That
% specialization is a design choice within the HJI running cost and
% is not imposed as a structural assumption on the original plant.
% \end{remark}

To embed the input constraint directly into the optimality
equation, define the nonquadratic control penalty
\begin{equation}
    U(\bu)
    =
    2\sum_{j=1}^{m}
    \int_{0}^{u_j}
    \lambda r_j
    \tanh^{-1}
    \left(
        \frac{v}{\lambda}
    \right)
    dv,
    \qquad
    \bu\in\mathcal U .
\label{eq:nonquadratic_U}
\end{equation}
The inverse hyperbolic tangent is evaluated componentwise whenever
it is applied to a vector.

The secure infinite-horizon performance index is
% \begin{equation}
% \begin{aligned}
%     J(\bx_0,\bu,\ba,\bd)
%     =
%     \int_{0}^{\infty}
%     \big[
%         Q(\bx)
%         +U(\bu)
%         -\gamma_a^2\ba^\top\mathbf T\ba
%         -\gamma_d^2\bd^\top\mathbf S\bd
%     \big]dt,
% \end{aligned}
% \label{eq:hji_cost}
% \end{equation}

\begin{equation}
\begin{aligned}
J(\bx_0,\bu,\ba,\bd)
&=
\int_{0}^{\infty}
\Big[
Q(\bx)+U(\bu)\\
&\qquad
-\gamma_a^2\ba^\top\mathbf T\ba
-\gamma_d^2\bd^\top\mathbf S\bd
\Big]dt .
\end{aligned}
\label{eq:hji_cost}
\end{equation}

where \(\gamma_a>0\) and \(\gamma_d>0\) are prescribed attenuation
levels. The controller minimizes \eqref{eq:hji_cost}, whereas the
FDI and disturbance channels maximize it.

% The optimal value function is defined as
% \[
% V^*(\bx(0))
% :=
% \min_{\bu\in\mathcal U}
%     \max_{\ba\in\mathbb R^m,\;\bd\in\mathbb R^n}
% J(\bx(0),\bu,\ba,\bd).
% \]

The optimal value function is defined as
\[
V^*(\bx(0))
:=
\min_{\bu\in\mathcal U}
\max_{\ba,\bd}
J(\bx(0),\bu,\ba,\bd).
\]

For a continuously differentiable candidate value function
\(V(\bx)\), define the Hamiltonian
\begin{equation}
\begin{aligned}
    H(\bx,\bu,\ba,\bd,\nabla V)
    ={}&
    Q(\bx)
    +U(\bu)
    -\gamma_a^2\ba^\top\mathbf T\ba
    -\gamma_d^2\bd^\top\mathbf S\bd\\
    &+
    (\nabla V(\bx))^\top
    \left[
        f(\bx)
        +
        \mathbf g(\bx)(\bu+\ba)
        +
        \bd
    \right].
\end{aligned}
\label{eq:H_def}
\end{equation}

If the value function is continuously differentiable, it satisfies
the Hamilton--Jacobi--Isaacs equation
\begin{equation}
    0
    =
    \min_{\bu\in\mathcal U}
    \max_{\ba,\;\bd}
    H(\bx,\bu,\ba,\bd,\nabla V^*).
\label{eq:HJI_eq}
\end{equation}

% \begin{assumption}
% \label{ass:HJI_solution}
% There exists a continuously differentiable positive definite
% solution
% \[
%     V^*:\Omega\to\mathbb R_{\geq0},
%     \qquad
%     V^*(\bm 0)=0,
% \]
% of \eqref{eq:HJI_eq}. The policies generated by \(V^*\) are
% admissible on \(\Omega\). In addition, there exists a constant
% \(\bar c_V>0\) such that
% \begin{equation}
%     V^*(\bx)
%     \leq
%     \bar c_V\|\bx\|^2,
%     \qquad
%     \forall\bx\in\Omega .
% \label{eq:value_quadratic_upper_bound}
% \end{equation}
% \end{assumption}

% \begin{remark}
% \label{rem:value_upper_bound}
% Condition \eqref{eq:value_quadratic_upper_bound} is used only to
% translate measurable state powers into powers of the ideal value
% function. It is satisfied locally, for example, when
% \(\nabla V^*(\bm 0)=\bm 0\) and \(\nabla V^*\) is locally
% Lipschitz on a neighborhood of the origin. It does not require
% the value function itself to be quadratic.
% \end{remark}

The stationary conditions associated with
\eqref{eq:HJI_eq} are
\begin{equation}
    \nabla_{\bu}H=\bm 0,
    \qquad
    \nabla_{\ba}H=\bm 0,
    \qquad
    \nabla_{\bd}H=\bm 0.
\label{eq:H_stationarity}
\end{equation}
Using \eqref{eq:H_stationarity}, the minimizing control policy
satisfies
\begin{equation}
    2\lambda\mathbf R
    \tanh^{-1}
    \left(
        \frac{\bu^*}{\lambda}
    \right)
    +
    \mathbf g^\top(\bx)\nabla V^*(\bx)
    =
    \bm 0,
\label{eq:stationary_u}
\end{equation}
and is therefore
\begin{equation}
    \bu^*(\bx)
    =
    -\lambda
    \tanh
    \left(
        \frac{1}{2\lambda}
        \mathbf R^{-1}
        \mathbf g^\top(\bx)
        \nabla V^*(\bx)
    \right).
\label{eq:u_star}
\end{equation}
Consequently,
\begin{equation}
    |u_j^*(\bx)|<\lambda,
    \qquad
    j=1,\ldots,m.
\label{eq:optimal_input_constraint}
\end{equation}

The maximizing FDI and disturbance policies are
\begin{equation}
    \ba^*(\bx)
    =
    \frac{1}{2\gamma_a^2}
    \mathbf T^{-1}
    \mathbf g^\top(\bx)
    \nabla V^*(\bx),
\label{eq:a_star}
\end{equation}
and
\begin{equation}
    \bd^*(\bx)
    =
    \frac{1}{2\gamma_d^2}
    \mathbf S^{-1}
    \nabla V^*(\bx).
\label{eq:d_star}
\end{equation}

% \begin{remark}
% \label{rem:virtual_adversarial_policies}
% The policies \(\ba^*(\bx)\) and \(\bd^*(\bx)\) are the
% worst-case maximizing policies associated with the HJI game.
% They are not control commands applied by the secure controller.
% Their role is to characterize the worst-case attack and
% disturbance directions against which the minimizing policy
% \(\bu^*(\bx)\) is designed.
% \end{remark}

Substituting \eqref{eq:u_star}--\eqref{eq:d_star} into
\eqref{eq:H_def} gives the closed HJI identity
\begin{equation}
\begin{aligned}
    0
    ={}&
    Q(\bx)
    +U(\bu^*)
    +
    (\nabla V^*(\bx))^\top f(\bx)\\
    &+
    (\nabla V^*(\bx))^\top
    \mathbf g(\bx)\bu^*\\
    &+
    \frac{1}{4\gamma_a^2}
    (\nabla V^*(\bx))^\top
    \mathbf g(\bx)\mathbf T^{-1}
    \mathbf g^\top(\bx)
    \nabla V^*(\bx)\\
    &+
    \frac{1}{4\gamma_d^2}
    (\nabla V^*(\bx))^\top
    \mathbf S^{-1}
    \nabla V^*(\bx).
\end{aligned}
\label{eq:HJI_closed}
\end{equation}

\begin{proposition}
\label{prop:saddle}
Under Assumptions~\ref{ass:system_regular}, the triple
\((\bu^*,\ba^*,\bd^*)\), defined by
\eqref{eq:u_star}--\eqref{eq:d_star}, is the pointwise saddle
policy triple of the Hamiltonian. In particular,
\begin{equation}
\begin{aligned}
    H(\bx,\bu^*,\ba,\bd,\nabla V^*)
    &\leq
    H(\bx,\bu^*,\ba^*,\bd^*,\nabla V^*)\\
    &\leq
    H(\bx,\bu,\ba^*,\bd^*,\nabla V^*)
\end{aligned}
\label{eq:saddle_property}
\end{equation}
% for every \(\bu\in\mathcal U,
%     \quad
%     \ba\in\mathbb R^m,
%     \quad
%     \bd\in\mathbb R^n.\)

\end{proposition}

The origin is retained as the nominal regulation target and is not
excluded from the operating region \(\Omega\). For the subsequent
fixed-time analysis, choose \(r_->0\) such that
\(B_{r_-}(0)\subset\Omega\), consistently with the prescribed
practical terminal region, and define the nonterminal comparison set
\(\Omega^{r}:=\{x\in\Omega:\|x\|\ge r_-\}\).
If \(x=0\), the regulation objective has already been achieved,
whereas if \(0<\|x\|<r_-\), the trajectory already lies inside
the prescribed terminal neighborhood and no fixed-time entrance
estimate is required. Hence, the nontrivial convergence analysis is
needed only for \(\|x\|\ge r_-\), i.e., on
\(\Omega^{r}\). Since \(\Omega\) is compact and
\(\{x:\|x\|\ge r_-\}\) is closed,
\(\Omega^{r}\) is compact, with \(0\notin\Omega^{r}\) and
\(\inf_{x\in\Omega^{r}}\|x\|\ge r_->0\).
Consequently, the ratios
\(V^\ast(x)/\|x\|^2\) and
\(\|\nabla_xV^\ast(x)\|/\|x\|\) are well defined on
\(\Omega^{r}\), which is precisely the region required for the
fixed-time comparison analysis. If disturbances or FDI signals drive
the trajectory outside the terminal neighborhood, the same comparison
argument becomes applicable again once \(\|x\|\ge r_-\).
Thus, excluding the origin from \(\Omega^{r}\) is purely an analytical
device and does not remove the origin from either the HJI operating
domain or the admissible-policy formulation.

\begin{lemma}[Local bounds of the ideal HJI value function]
\label{lemm1}
Suppose that \(V^\ast\) is positive definite on \(\Omega\), satisfies \(V^\ast(0)=0\), and \(V^\ast\in C^1(\Omega\setminus\{0\})\). Then there exist a function \(\underline{\alpha}\in\mathcal K_\infty\) and constants \(c_V>0\) and \(c_{\nabla}>0\) such that \(\underline{\alpha}(\|x\|)\leq V^\ast(x)\leq c_V\|x\|^2\) and \(\|\nabla_xV^\ast(x)\|\leq c_{\nabla}\|x\|\) for all \(x\in\Omega^{r}\).
\label{lem:local_value_bounds}
\end{lemma}

\begin{proof}
By construction, \(\Omega^{r}\) is compact, \(0\notin\Omega^{r}\), and \(\inf_{x\in\Omega^{r}}\|x\|\geq r_->0\). Since \(V^\ast\in C^1(\Omega\setminus\{0\})\), the functions \(V^\ast(x)/\|x\|^2\) and \(\|\nabla_xV^\ast(x)\|/\|x\|\) are continuous on \(\Omega^{r}\). Hence, by the Weierstrass theorem, the finite constants \(c_V:=\max_{x\in\Omega^{r}}V^\ast(x)/\|x\|^2\) and \(c_{\nabla}:=\max_{x\in\Omega^{r}}\|\nabla_xV^\ast(x)\|/\|x\|\) exist, which directly yield \(V^\ast(x)\leq c_V\|x\|^2\) and \(\|\nabla_xV^\ast(x)\|\leq c_{\nabla}\|x\|\) on \(\Omega^{r}\). Moreover, since \(V^\ast(x)>0\) for every \(x\in\Omega^{r}\), continuity and compactness imply that \(\underline c_V:=\min_{x\in\Omega^{r}}V^\ast(x)/\|x\|^2>0\). Defining \(\underline{\alpha}(s):=\underline c_Vs^2\) gives \(\underline{\alpha}\in\mathcal K_\infty\) and \(\underline{\alpha}(\|x\|)\leq V^\ast(x)\) on \(\Omega^{r}\), which completes the proof.
\end{proof}

\begin{remark}[Cost-induced fixed-time structure]
\label{rem:cost_induced_fixed_time}
For the fixed-time design developed later, the general penalty
\(Q(\bx)\) is specialized as
\begin{equation}
    Q(\bx)
    =
    \bx^\top\mathbf Q_x\bx
    +
    \kappa_1\|\bx\|^{2\alpha}
    +
    \kappa_2\|\bx\|^{2\beta},
\label{eq:fixed_time_state_penalty}
\end{equation}
where
\begin{equation}
    \mathbf Q_x
    =
    \mathbf Q_x^\top>0,
    \qquad
    \kappa_1,\kappa_2>0,
    \qquad
    0<\alpha<1<\beta.
\label{eq:fixed_time_cost_parameters}
\end{equation}
This construction depends only on the measurable state and does
not introduce the unknown value function into the running cost.

\begin{remark}
\label{remark3}
Lemma~\ref{lem:local_value_bounds} is introduced solely to provide the local value-function and gradient bounds required for the subsequent fixed-time stability analysis.
\end{remark}

From Lemma \ref{lem:local_value_bounds},
\begin{equation}
    \|\bx\|^{2\alpha}
    \geq
    \bar c_V^{-\alpha}
    \big(V^*(\bx)\big)^\alpha,
    \qquad
    \|\bx\|^{2\beta}
    \geq
    \bar c_V^{-\beta}
    \big(V^*(\bx)\big)^\beta.
\label{eq:state_power_to_value_power}
\end{equation}
Along the ideal saddle-point trajectory,
\begin{equation}
\begin{aligned}
    \dot V^*(\bx)
    ={}&
    -Q(\bx)
    -U(\bu^*)\\
    &+
    \gamma_a^2
    (\ba^*(\bx))^\top
    \mathbf T\ba^*(\bx)\\
    &+
    \gamma_d^2
    (\bd^*(\bx))^\top
    \mathbf S\bd^*(\bx).
\end{aligned}
\label{eq:ideal_value_derivative}
\end{equation}
Since \(\Omega^r\) is compact and \(V^*\) is continuously
differentiable, the last two terms are bounded on \(\Omega^r\).
Hence there exists \(\Delta_H\geq0\) such that
\begin{equation}
\begin{aligned}
    &\gamma_a^2
    (\ba^*(\bx))^\top
    \mathbf T\ba^*(\bx)\\
    &\quad+
    \gamma_d^2
    (\bd^*(\bx))^\top
    \mathbf S\bd^*(\bx)
    \leq
    \Delta_H,
    \qquad
    \forall\bx\in\Omega^r.
\end{aligned}
\label{eq:ideal_adversarial_bound}
\end{equation}
Using \(U(\bu^*)\geq0\) and
\eqref{eq:state_power_to_value_power} in
\eqref{eq:ideal_value_derivative} gives
\begin{equation}
    \dot V^*(\bx)
    \leq
    -c_1\big(V^*(\bx)\big)^\alpha
    -c_2\big(V^*(\bx)\big)^\beta
    +\Delta_H,
\label{eq:ideal_fixed_time_value_inequality}
\end{equation}
where
\begin{equation}
    c_1
    =
    \kappa_1\bar c_V^{-\alpha},
    \qquad
    c_2
    =
    \kappa_2\bar c_V^{-\beta}.
\label{eq:ideal_fixed_time_coefficients}
\end{equation}
Thus, the two fixed-time powers arise directly from the
HJI-shaped state penalty rather than from an external assumption
on the unknown value function. In the learned closed loop,
critic-weight and neural-approximation errors produce additional
residual terms, which will be handled explicitly in
Section~\ref{sec:stability}.
\end{remark}

\begin{remark}[Computation of \(\bar c\)]
\label{rem:computable_cbar}
The constant \(\bar c\) in Lemma~\ref{lemm1} need not be evaluated from the unknown \(V^*\). All signals entering the relevant cost and finite-window quantities admit uniform bounds on \(\Omega^r\). For the primitive formulation in~\eqref{eq:plant}, Assumption~\ref{ass:signal_regular} gives \(\|\ba(t)\|\leq\bar a\) and \(\|\bd(t)\|\leq\bar d\). If the signals are generated by the ideal saddle policies, Lemma~~\ref{lemm1} , Assumption~\ref{ass:system_regular}, and compactness of \(\Omega^r\) give \(\bar g:=\max_{\bx\in\Omega^r}\|\bm g(\bx)\|<\infty\), \(\bar k:=\max_{\bx\in\Omega^r}\|\bm k(\bx)\|<\infty\), and \(\bar x_r:=\max_{\bx\in\Omega^r}\|\bx\|<\infty\), so that, from (14b)--(14c), \(\|\ba^*(\bx)\|\leq\|\mathbf T^{-1}\|\bar g c_{\nabla}\bar x_r/(2\gamma_a^2)=:\bar a^*<\infty\) and \(\|\bd^*(\bx)\|\leq\|\mathbf S^{-1}\|\bar k c_{\nabla}\bar x_r/(2\gamma_d^2)=:\bar d^*<\infty\). Likewise, since \(\|\nabla\bm\phi(\bx)\|\leq\bar\phi_g\) on \(\Omega^r\) and Lemma~\ref{lemm3} gives \(\|\hat{\bW}(t)\|\leq\bar W_{\hat{}}\), the learned saddle signals satisfy \(\|\hat{\ba}(t)\|\leq\|\mathbf T^{-1}\|\bar g\bar\phi_g\bar W_{\hat{}}/(2\gamma_a^2)=:\bar{\hat a}<\infty\) and \(\|\hat{\bd}(t)\|\leq\|\mathbf S^{-1}\|\bar k\bar\phi_g\bar W_{\hat{}}/(2\gamma_d^2)=:\bar{\hat d}<\infty\), while the saturation parameterization gives \(|\hat u_i(t)|<\bar u_i\).

Accordingly, let \(u_b\in\Psi(\Omega)\) be any certified admissible secure policy and let \(\bar J_b<\infty\) be a computable worst-case performance certificate satisfying \(\bar J_b\geq\sup_{\bx_0\in\Omega^r}\sup_{\ba,\bd}J(\bx_0;u_b,\ba,\bd)\), where the admissible \(\ba,\bd\) obey Assumptions~1 and~2. Since \(V^*(\bx)=\min_u\max_{\ba,\bd}J(\bx;u,\ba,\bd)\), one has \(V^*(\bx)\leq\sup_{\ba,\bd}J(\bx;u_b,\ba,\bd)\leq\bar J_b\) for every \(\bx\in\Omega^r\). Moreover, \(\|\bx\|\geq r_-\) on \(\Omega^r\). The operating domain \(\Omega\), and hence \(\Omega^r=\{\bx\in\Omega:\|\bx\|\geq r_-\}\), is prescribed a priori from the design specifications, physical operating limits, safety requirements, or the intended data-collection region, rather than inferred from the unknown drift dynamics. Hence
\[
\bar c
=
\max_{\bx\in\Omega^r}
\frac{V^*(\bx)}{\|\bx\|^2}
\leq
\frac{\bar J_b}{r_-^2}
=:
\bar c_b .
\]
Thus, the explicit certified number \(\bar c_b=\bar J_b/r_-^2\) may be used everywhere in place of the unknown exact \(\bar c\), and consequently \(\|\bx\|^{2\gamma_1}\geq\bar c_b^{-\gamma_1}(V^*(\bx))^{\gamma_1}\) and \(\|\bx\|^{2\nu}\geq\bar c_b^{-\nu}(V^*(\bx))^\nu\). The price is conservatism: a small prescribed radius \(r_-\) or a worst-case \(\bar J_b\) computed over a large operating domain can make \(\bar c_b\) large, thereby increasing the penalty or learning gains required by the predefined-time design and potentially producing more aggressive near-saturation transients, greater numerical stiffness and dynamic range, and tighter sampling and computational requirements in digital hardware. This is the deliberate tradeoff for obtaining an a-priori computable convergence certificate without knowing \(V^*\).
\end{remark}

Along the saddle-point trajectory,
\begin{equation}
\dot V^*(\bx)
+
\ell(\bx,\bu^*,\ba^*,\bd^*)
=
0.
\label{eq:differential_BI_star}
\end{equation}
Integrating \eqref{eq:differential_BI_star} over the reinforcement
interval $[t-\Delta T,t]$, where $\Delta T>0$, yields
\begin{equation}
\begin{aligned}
&V^*(\bx(t))
-
V^*(\bx(t-\Delta T))\\
&\quad+
\int_{t-\Delta T}^{t}
\ell\big(
\bx(\tau),
\bu^*(\tau),
\ba^*(\tau),
\bd^*(\tau)
\big)d\tau
=
0.
\end{aligned}
\label{eq:integral_BI_star}
\end{equation}

\begin{remark}
\label{rem:integral_BI_bridge}
Equation \eqref{eq:integral_BI_star} is the bridge from the secure
HJI equation to integral reinforcement learning. The differential
HJI identity contains the unknown drift \(f(\bx)\), whereas
\eqref{eq:integral_BI_star} only requires state increments,
implemented policies, and an integral of the measurable running
cost. Consequently, the critic can be trained without explicitly
evaluating \(f(\bx)\).
\end{remark}

% \begin{remark}
% \label{rem:secure_game_interpretation}
% The proposed formulation differs fundamentally from first
% designing a nominal HJB controller and subsequently checking its
% robustness. The FDI and disturbance channels are introduced as
% maximizing players from the beginning, while the control
% constraint is embedded in \(U(\bu)\). Therefore, the value
% function learned in the sequel is the value of a saturated secure
% HJI game rather than that of an unconstrained nominal control
% problem.
% \end{remark}

% ============================================================
\section{Critic-Only Fixed-Time Integral Reinforcement Learning}
\label{sec:fixed_time_irl}
% ============================================================

\begingroup
\setlength{\abovedisplayskip}{3pt}
\setlength{\belowdisplayskip}{3pt}
\setlength{\abovedisplayshortskip}{2pt}
\setlength{\belowdisplayshortskip}{2pt}

This section develops a critic-only fixed-time integral
reinforcement learning method for solving the saturated secure
HJI problem formulated in Section~\ref{sec:prelim_problem}.

% ------------------------------------------------------------
\subsection{Critic Approximation and Induced Policies}
\label{subsec:critic_approximation}
% ------------------------------------------------------------

% Let \(L\) denote the number of critic basis functions. For a
% sufficiently rich continuously differentiable basis
% \(\bphi:\Omega^r\to\mathbb R^L\), the ideal value function and its
% gradient can be represented on the compact set \(\Omega^r\) as
% \begin{equation}
% \begin{aligned}
%     V^*(\bx)
%     &=
%     \mathbf W^\top\bphi(\bx)
%     +
%     \varepsilon(\bx),\\
%     \nabla V^*(\bx)
%     &=
%     \big(\nabla\bphi(\bx)\big)^\top\mathbf W
%     +
%     \nabla\varepsilon(\bx),
% \end{aligned}
% \label{eq:value_NN_approximation}
% \end{equation}
% where

% \[
%     \mathbf W\in\mathbb R^L
% \]

% is the constant ideal critic-weight vector,
% \(\bphi(\bx)\in\mathbb R^L\) is the activation vector, and
% \(\varepsilon(\bx)\) is the value-function approximation error.
% The activation vector is selected such that

% \[
%     \bphi(\bm 0)=\bm 0.
% \]

% The critic estimate is defined by
% \begin{equation}
% \begin{aligned}
%     \hat V(\bx)
%     &=
%     \hat{\mathbf W}^\top\bphi(\bx),\\
%     \nabla\hat V(\bx)
%     &=
%     \big(\nabla\bphi(\bx)\big)^\top
%     \hat{\mathbf W},
% \end{aligned}
% \label{eq:value_critic_estimate}
% \end{equation}
% where \(\hat{\mathbf W}(t)\in\mathbb R^L\) is the adjustable
% critic-weight vector. Define the critic-weight estimation error as
% \begin{equation}
%     \widetilde{\mathbf W}
%     =
%     \hat{\mathbf W}
%     -
%     \mathbf W .
% \label{eq:critic_weight_error}
% \end{equation}

Let \(L\) denote the number of critic basis functions. For a sufficiently rich continuously differentiable basis \(\bphi:\Omega^r\to\mathbb R^L\), the ideal value function and its gradient can be represented on \(\Omega^r\) as \(V^*(\bx)=\mathbf W^\top\bphi(\bx)+\varepsilon(\bx)\) and \(\nabla V^*(\bx)=\big(\nabla\bphi(\bx)\big)^\top\mathbf W+\nabla\varepsilon(\bx)\), where \(\mathbf W\in\mathbb R^L\) is the constant ideal critic-weight vector, \(\bphi(\bx)\in\mathbb R^L\) is the activation vector, and \(\varepsilon(\bx)\) is the value-function approximation error. The activation vector is selected such that \(\bphi(\bm0)=\bm0\). The critic estimate and its gradient are defined as \(\hat V(\bx)=\hat{\mathbf W}^\top\bphi(\bx)\) and \(\nabla\hat V(\bx)=\big(\nabla\bphi(\bx)\big)^\top\hat{\mathbf W}\), where \(\hat{\mathbf W}(t)\in\mathbb R^L\) is the adjustable critic-weight vector. Define the critic-weight estimation error as \(\widetilde{\mathbf W}:=\hat{\mathbf W}-\mathbf W\).

% \begin{assumption}[Critic approximation bounds]
% \label{ass:critic_approximation}
% The ideal weight vector is finite. Moreover, there exist positive
% constants
% \[
%     W_M,\quad
%     \phi_M,\quad
%     \phi_{gM},\quad
%     \bar\varepsilon,\quad
%     \bar\varepsilon_g
% \]
% such that, for every \(\bx\in\Omega^r\),
% \begin{equation}
% \begin{aligned}
%     \|\mathbf W\|
%     &\leq W_M,\\
%     \|\bphi(\bx)\|
%     &\leq \phi_M,\\
%     \|\nabla\bphi(\bx)\|
%     &\leq \phi_{gM},
% \end{aligned}
% \label{eq:critic_basis_bounds}
% \end{equation}
% and
% \begin{equation}
% \begin{aligned}
%     |\varepsilon(\bx)|
%     &\leq \bar\varepsilon,\\
%     \|\nabla\varepsilon(\bx)\|
%     &\leq \bar\varepsilon_g.
% \end{aligned}
% \label{eq:critic_error_bounds}
% \end{equation}
% \end{assumption}

% \begin{assumption}[Critic approximation bounds]
% \label{ass:critic_approximation}
% The ideal weight vector is finite. Moreover, there exist positive constants
% \(\phi_M\), \(\phi_{gM}\), \(\bar\varepsilon\), and \(\bar\varepsilon_g\) such that, for every \(\bx\in\Omega^r\),
% \(\|\bphi(\bx)\|\le\phi_M\),
% \(\|\nabla\bphi(\bx)\|\le\phi_{gM}\),
% \(|\varepsilon(\bx)|\le\bar\varepsilon\), and
% \(\|\nabla\varepsilon(\bx)\|\le\bar\varepsilon_g\).
% \end{assumption}

\begin{assumption}[Critic approximation error]
\label{ass:critic_approximation}
There exist positive constants \(\bar\varepsilon\) and \(\bar\varepsilon_g\) such that, for every \(\bx\in\Omega^r\),
\(|\varepsilon(\bx)|\le\bar\varepsilon\) and
\(\|\nabla\varepsilon(\bx)\|\le\bar\varepsilon_g\).
\end{assumption}

\begin{lemma}[Bounded basis functions]
\label{lem:basis_bound}
Suppose the basis function vector \(\bphi(\bx)\) is continuously differentiable on the compact set \(\Omega^r\). Then there exist positive constants \(\phi_M\) and \(\phi_{gM}\) such that, for every \(\bx\in\Omega^r\),
\(\|\bphi(\bx)\|\le\phi_M\) and
\(\|\nabla\bphi(\bx)\|\le\phi_{gM}\).
\end{lemma}

\begin{proof}
Since \(\bphi\in C^{1}(\Omega^r)\), both \(\bphi(\bx)\) and \(\nabla\bphi(\bx)\) are continuous on the compact set \(\Omega^r\). Hence, by the extreme value theorem,
\(\phi_M:=\max_{\bx\in\Omega^r}\|\bphi(\bx)\|<\infty\) and
\(\phi_{gM}:=\max_{\bx\in\Omega^r}\|\nabla\bphi(\bx)\|<\infty\).
Therefore, for every \(\bx\in\Omega^r\),
\(\|\bphi(\bx)\|\le\max_{\bz\in\Omega^r}\|\bphi(\bz)\|=\phi_M\) and
\(\|\nabla\bphi(\bx)\|\le\max_{\bz\in\Omega^r}\|\nabla\bphi(\bz)\|=\phi_{gM}\).
\end{proof}

% \begin{remark}
% Typical choices of the basis function vector include polynomial bases
% \(\bphi(\bx)=[1,\bx^\top,(\bx\otimes\bx)^\top,\ldots]^\top\),
% sigmoidal bases
% \(\phi_i(\bx)=\left(1+e^{-a_i^\top\bx-b_i}\right)^{-1}\),
% hyperbolic tangent bases
% \(\phi_i(\bx)=\tanh(a_i^\top\bx+b_i)\),
% radial basis functions
% \(\phi_i(\bx)=\exp\!\left(-\|\bx-c_i\|^2/(2\sigma_i^2)\right)\),
% and Fourier bases
% \(\phi_i(\bx)=\cos(a_i^\top\bx+b_i)\) or
% \(\phi_i(\bx)=\sin(a_i^\top\bx+b_i)\).
% \end{remark}

\begin{remark}
Typical choices of the basis function vector include polynomial bases
\(\bphi(\bx)=[1,\bx^\top,(\bx\otimes\bx)^\top,\ldots]^\top\),
sigmoidal bases
\(\phi_i(\bx)=(1+e^{-a_i^\top\bx-b_i})^{-1}\),
hyperbolic tangent bases
\(\phi_i(\bx)=\tanh(a_i^\top\bx+b_i)\),
radial basis functions
\(\phi_i(\bx)=\exp(-\|\bx-c_i\|^2/(2\sigma_i^2))\),
and Fourier bases
\(\phi_i(\bx)\in\{\sin(a_i^\top\bx+b_i),\,\cos(a_i^\top\bx+b_i)\}\).
Since these basis functions are continuous (and continuously differentiable) on the compact set \(\Omega^r\), Lemma~\ref{lem:basis_bound} guarantees the existence of finite constants \(\phi_M\) and \(\phi_{gM}\) by the Weierstrass extreme value theorem. In particular, sigmoid, hyperbolic tangent, radial basis, and Fourier functions are intrinsically bounded, whereas polynomial bases become bounded when restricted to the compact domain \(\Omega^r\).
\end{remark}

\eqref{eq:u_star}--\eqref{eq:d_star} gives
\begin{equation}
\begin{aligned}
    \bu^*(\bx)
    =
    -\lambda\tanh
    \bigg[
        \frac{1}{2\lambda}
        \mathbf R^{-1}\mathbf g^\top(\bx)
        \Big(
            \big(\nabla\bphi(\bx)\big)^\top\mathbf W
            +
            \nabla\varepsilon(\bx)
        \Big)
    \bigg],
\end{aligned}
\label{eq:ideal_control_NN}
\end{equation}
\begin{equation}
\begin{aligned}
    \ba^*(\bx)
    =
    \frac{1}{2\gamma_a^2}
    \mathbf T^{-1}\mathbf g^\top(\bx)
    \Big(
        \big(\nabla\bphi(\bx)\big)^\top\mathbf W
        +
        \nabla\varepsilon(\bx)
    \Big),
\end{aligned}
\label{eq:ideal_attack_NN}
\end{equation}
and
\begin{equation}
\begin{aligned}
    \bd^*(\bx)
    =
    \frac{1}{2\gamma_d^2}
    \mathbf S^{-1}
    \Big(
        \big(\nabla\bphi(\bx)\big)^\top\mathbf W
        +
        \nabla\varepsilon(\bx)
    \Big).
\end{aligned}
\label{eq:ideal_disturbance_NN}
\end{equation}

Because \(\mathbf W\) and \(\varepsilon(\bx)\) are unknown, the
implementable policies are generated from the estimated value
gradient 
\begin{equation}
\begin{aligned}
    \hat{\bu}(\bx)
    =
    -\lambda\tanh
    \left[
        \frac{1}{2\lambda}
        \mathbf R^{-1}\mathbf g^\top(\bx)
        \big(\nabla\bphi(\bx)\big)^\top
        \hat{\mathbf W}
    \right],
\end{aligned}
\label{eq:learned_control_policy}
\end{equation}
\begin{equation}
\begin{aligned}
    \hat{\ba}(\bx)
    =
    \frac{1}{2\gamma_a^2}
    \mathbf T^{-1}\mathbf g^\top(\bx)
    \big(\nabla\bphi(\bx)\big)^\top
    \hat{\mathbf W},
\end{aligned}
\label{eq:learned_attack_policy}
\end{equation}
and
\begin{equation}
\begin{aligned}
    \hat{\bd}(\bx)
    =
    \frac{1}{2\gamma_d^2}
    \mathbf S^{-1}
    \big(\nabla\bphi(\bx)\big)^\top
    \hat{\mathbf W}.
\end{aligned}
\label{eq:learned_disturbance_policy}
\end{equation}

Since the hyperbolic tangent in
\eqref{eq:learned_control_policy} is applied componentwise,
\begin{equation}
    |\hat u_j(\bx)|
    <
    \lambda,
    \qquad
    j=1,\ldots,m,
\label{eq:learned_input_bound}
\end{equation}
for every finite \(\hat{\mathbf W}\) and every
\(\bx\in\Omega^r\).

\subsection{Integral Bellman--Isaacs Residual}
\label{subsec:integral_BI_residual}
% ------------------------------------------------------------

Let \(\Delta T>0\) denote the reinforcement-interval length.
Define the ideal finite-window running-cost integral
\begin{equation}
\begin{aligned}
\mathcal R^*(t)
=
\int_{t-\Delta T}^{t}
\ell\big(
\bx(\tau),
\bu^*(\tau),
\ba^*(\tau),
\bd^*(\tau)
\big)d\tau .
\end{aligned}
\label{eq:ideal_window_cost}
\end{equation}

The integral Bellman--Isaacs identity
\eqref{eq:integral_BI_star} can then be written as
\begin{equation}
V^*(\bx(t))
-
V^*(\bx(t-\Delta T))
+
\mathcal R^*(t)
=
0.
\label{eq:exact_integral_BI_identity}
\end{equation}

\begin{remark}[Measured-trajectory mismatch]
\label{rem:measured_trajectory_mismatch}
Equation~\eqref{eq:exact_integral_BI_identity} is exact when the state transition over \([t-\Delta T,t]\) is generated by the ideal saddle vector field. In the learning implementation, however, the available transition is the measured one generated by the current policies. Let \(\dot{\bx}^{\,*}(\tau)\) denote the ideal saddle vector field evaluated at \(\bx(\tau)\) and define \(\eta_{\mathrm{tr}}(t):=\int_{t-\Delta T}^{t}\nabla V^{*\top}(\bx(\tau))[\dot{\bx}(\tau)-\dot{\bx}^{\,*}(\tau)]d\tau\). Hence, along the measured transition, \(V^*(\bx(t))-V^*(\bx(t-\Delta T))+\mathcal R^*(t)-\eta_{\mathrm{tr}}(t)=0\). In particular, \(\eta_{\mathrm{tr}}(t)=0\) whenever the measured and ideal saddle evolutions coincide over the entire reinforcement interval.
\end{remark}

Define the current feature increment
\begin{equation}
\Delta\bphi(t)
=
\bphi(\bx(t))
-
\bphi(\bx(t-\Delta T)).
\label{eq:current_feature_increment}
\end{equation}

The running-cost datum generated by the current policy triple is
\begin{equation}
\begin{aligned}
\widehat{\mathcal R}(t)
=
\int_{t-\Delta T}^{t}
\ell\big(
\bx(\tau),
\hat{\bu}(\tau),
\hat{\ba}(\tau),
\hat{\bd}(\tau)
\big)d\tau .
\end{aligned}
\label{eq:implemented_window_cost}
\end{equation}

The current integral Bellman--Isaacs residual is selected as
\begin{equation}
\xi(t)
=
\hat{\mathbf W}^\top(t)
\Delta\bphi(t)
+
\widehat{\mathcal R}(t).
\label{eq:current_integral_BI_residual}
\end{equation}

Unlike a pointwise Hamiltonian residual,
\eqref{eq:current_integral_BI_residual} does not contain
\(\dot{\bx}\), \(f(\bx)\), or
\((\nabla\hat V)^\top f(\bx)\). Its evaluation only requires the
state values at the two endpoints of the reinforcement interval
and the integral of the running game cost over that interval.

For later analysis, define
\begin{equation}
\Delta\varepsilon(t)
=
\varepsilon(\bx(t))
-
\varepsilon(\bx(t-\Delta T)).
\label{eq:approximation_error_increment}
\end{equation}

Along the measured trajectory, Remark~\ref{rem:measured_trajectory_mismatch} gives
\begin{equation}
\mathbf W^\top\Delta\bphi(t)
+
\Delta\varepsilon(t)
+
\mathcal R^*(t)
-
\eta_{\mathrm{tr}}(t)
=
0.
\label{eq:ideal_weight_integral_identity}
\end{equation}

Since \(\hat{\mathbf W}=\mathbf W+\widetilde{\mathbf W}\),
substitution of \eqref{eq:ideal_weight_integral_identity} into
\eqref{eq:current_integral_BI_residual} yields
\begin{equation}
\xi(t)
=
\widetilde{\mathbf W}^\top(t)
\Delta\bphi(t)
+
\varepsilon_B(t),
\label{eq:current_residual_decomposition}
\end{equation}
where
\begin{equation}
\begin{aligned}
\varepsilon_B(t)
=
\widehat{\mathcal R}(t)
-
\mathcal R^*(t)
-
\Delta\varepsilon(t)
+
\eta_{\mathrm{tr}}(t).
\end{aligned}
\label{eq:current_BI_approximation_error}
\end{equation}

The term \(\varepsilon_B(t)\) collects the value-function
approximation error over the reinforcement interval together with
the difference between the implemented and ideal policy costs and
the measured-to-saddle trajectory mismatch.

\begin{lemma}[Local continuity of the critic-weight trajectory]
\label{lemm3}
Let \(\hat{\mathbf W}(\cdot)\) be a maximal Carath\'eodory solution of the critic
update law in \eqref{eq:fixed_time_critic_update}. Then
\(\hat{\mathbf W}\in AC_{\mathrm{loc}}\), and hence
\(\hat{\mathbf W}\in C^0\), on its maximal interval of existence. In particular,
for every compact interval \([t_1,t_2]\) contained in that interval,
there exists a finite constant \(\bar W_{[t_1,t_2]}>0\) such that
\(\|\hat{\mathbf W}(\tau)\|\leq\bar W_{[t_1,t_2]}\) for all
\(\tau\in[t_1,t_2]\).
\label{lem:critic_weight_continuity}
\end{lemma}

\begin{proof}
Denote the right-hand side of \eqref{eq:fixed_time_critic_update} by
\(\mathcal G_W\), so that \(\dot{\hat{\mathbf W}}=\mathcal G_W\).
The signed-power mappings
\(|s|^q\operatorname{sgn}(s)\) and
\(|s|^r\operatorname{sgn}(s)\), with \(0<q<1<r\), are continuous,
including at \(s=0\). Together with the continuity of the normalized
regressors and residuals entering \eqref{eq:fixed_time_critic_update}, this makes
\(\mathcal G_W\) locally integrable along every maximal
Carath\'eodory solution. Hence
\(\hat{\mathbf W}(t)=\hat{\mathbf W}(t_0)+\int_{t_0}^{t}\mathcal G_W(\tau)\,d\tau\),
which implies \(\hat{\mathbf W}\in AC_{\mathrm{loc}}\subset C^0\).
Therefore, on every compact interval \([t_1,t_2]\), the Weierstrass
theorem gives
\(\bar W_{[t_1,t_2]}:=
\max_{\tau\in[t_1,t_2]}\|\hat{\mathbf W}(\tau)\|<\infty\).
\end{proof}

\begin{lemma}[Bounded finite-window implemented running cost]
\label{lemma4}
Suppose that the learning trajectory remains in the prescribed compact
operating region \(\Omega^r\), and let \(\hat{\mathbf W}\) evolve according to
\eqref{eq:fixed_time_critic_update}. Then, for every fixed
\(t\geq\Delta T\) such that \([t-\Delta T,t]\) is contained in the
maximal interval of existence, there exist finite constants
\(\bar\ell_t>0\), \(\bar\ell_t^*>0\), and
\(\bar\eta_{\mathrm{tr},t}>0\) such that
\(|\ell(x(\tau),\hat u(\tau),\hat a(\tau),\hat d(\tau))|
\leq\bar\ell_t\),
\(|\ell(x(\tau),u^*(\tau),a^*(\tau),d^*(\tau))|
\leq\bar\ell_t^*\), and
\(|\eta_{\mathrm{tr}}(t)|\leq\bar\eta_{\mathrm{tr},t}\).
Consequently,
\(|\hat{\mathcal R}(t)|\leq\Delta T\,\bar\ell_t\) and
\(|\mathcal R^*(t)|\leq\Delta T\,\bar\ell_t^*\).
For the finite replay stack, the corresponding bounds over all recorded
windows can be dominated by finite common constants
\(\bar\ell\), \(\bar\ell^*\), and \(\bar\eta_{\mathrm{tr}}\).
\label{lem:finite_window_running_cost}
\end{lemma}

\begin{proof}
By Lemma~\ref{lem:critic_weight_continuity},
\(\hat{\mathbf W}\) is continuous and therefore bounded on the compact interval
\([t-\Delta T,t]\). Since \(x(\tau)\in\Omega^r\), with \(\Omega^r\)
compact, and \(g(x)\) and \(\nabla\phi(x)\) are continuous on
\(\Omega^r\), they are bounded on the same interval. The hyperbolic-tangent structure gives
\(|\hat u_j(\tau)|<\lambda\), \(j=1,\ldots,m\), while
\eqref{eq:learned_attack_policy}--\eqref{eq:learned_disturbance_policy} and the finite-window boundedness of
\(\hat{\mathbf W}\) imply that \(\hat a(\tau)\) and \(\hat d(\tau)\) are also
bounded on \([t-\Delta T,t]\). Since the running cost
\(\ell(x,\hat u,\hat a,\hat d)\) is continuous in its arguments,
\(\tau\mapsto\ell(x(\tau),\hat u(\tau),\hat a(\tau),\hat d(\tau))\)
is continuous on the compact interval \([t-\Delta T,t]\).
Therefore, the Weierstrass theorem yields
\(\bar\ell_t:=
\max_{\tau\in[t-\Delta T,t]}
|\ell(x(\tau),\hat u(\tau),\hat a(\tau),\hat d(\tau))|<\infty\).

Likewise, the ideal saddle policies are continuous on the compact
operating region and hence bounded. Therefore,
\(\ell(x,u^*,a^*,d^*)\) is continuous on the same finite window and admits
\(\bar\ell_t^*:=
\max_{\tau\in[t-\Delta T,t]}
|\ell(x(\tau),u^*(\tau),a^*(\tau),d^*(\tau))|<\infty\).
Moreover, compactness of \(\Omega^r\), boundedness of
\(\nabla V^*\), and boundedness of the measured and ideal saddle vector
fields imply that
\(\tau\mapsto\nabla V^{*\top}(\bx(\tau))
[\dot{\bx}(\tau)-\dot{\bx}^{\,*}(\tau)]\)
is bounded on \([t-\Delta T,t]\), and thus
\(|\eta_{\mathrm{tr}}(t)|<\infty\).

It follows directly that
\(|\hat{\mathcal R}(t)|
\leq\int_{t-\Delta T}^{t}
|\ell(x(\tau),\hat u(\tau),\hat a(\tau),\hat d(\tau))|\,d\tau
\leq\Delta T\,\bar\ell_t\)
and
\(|\mathcal R^*(t)|
\leq\Delta T\,\bar\ell_t^*\).
Since the replay stack contains only finitely many stored windows, the
corresponding finite maxima admit finite common upper bounds
\(\bar\ell\), \(\bar\ell^*\), and \(\bar\eta_{\mathrm{tr}}\).
\end{proof}

\begin{lemma}[Bounded integral Bellman--Isaacs approximation
residual]
\label{lem:bounded_integral_residual}
Under Assumptions~\ref{ass:critic_approximation}, and
Lemma~\ref{lemma4}, the approximation residual
\(\varepsilon_B(t)\) in
\eqref{eq:current_BI_approximation_error} is uniformly bounded.
In particular,
\begin{equation}
|\varepsilon_B(t)|
\leq
\bar\varepsilon_B,
\label{eq:current_BI_error_bound}
\end{equation}
where one admissible bound is
\begin{equation}
\bar\varepsilon_B
=
2\bar\varepsilon
+
\Delta T(\bar\ell+\bar\ell^*)
+
\bar\eta_{\mathrm{tr}}.
\label{eq:explicit_BI_error_bound}
\end{equation}
\end{lemma}

\begin{proof}
From Assumption~\ref{ass:critic_approximation},
\begin{equation}
\begin{aligned}
|\Delta\varepsilon(t)|
&=
\left|
\varepsilon(\bx(t))
-
\varepsilon(\bx(t-\Delta T))
\right|\\
&\leq
|\varepsilon(\bx(t))|
+
|\varepsilon(\bx(t-\Delta T))|\\
&\leq
2\bar\varepsilon.
\end{aligned}
\label{eq:epsilon_increment_bound}
\end{equation}
Moreover, Lemma~\ref{lemma4} gives
\begin{equation}
\begin{aligned}
\left|
\widehat{\mathcal R}(t)
-
\mathcal R^*(t)
\right|
&\leq
\int_{t-\Delta T}^{t}
\left|
\ell(\bx,\hat{\bu},\hat{\ba},\hat{\bd})
\right|d\tau\\
&\quad+
\int_{t-\Delta T}^{t}
\left|
\ell(\bx,\bu^*,\ba^*,\bd^*)
\right|d\tau\\
&\leq
\Delta T(\bar\ell+\bar\ell^*).
\end{aligned}
\label{eq:window_cost_difference_bound}
\end{equation}
Together with
\(|\eta_{\mathrm{tr}}(t)|\leq\bar\eta_{\mathrm{tr}}\),
combining \eqref{eq:current_BI_approximation_error},
\eqref{eq:epsilon_increment_bound}, and
\eqref{eq:window_cost_difference_bound} proves
\eqref{eq:current_BI_error_bound}.
\end{proof}
\subsection{Experience Replay and Finite Data Informativity}
\label{subsec:experience_replay}
% ------------------------------------------------------------

To preserve informative learning directions, collect a finite
history stack
\begin{equation}
    \mathcal H
    =
    \left\{
        \big(
            \Delta\bphi_k,
            \widehat{\mathcal R}_k
        \big)
    \right\}_{k=1}^{N},
\label{eq:history_stack}
\end{equation}
where \(N\) is the number of stored samples and
\begin{equation}
    \Delta\bphi_k
    =
    \bphi(\bx(t_k))
    -
    \bphi(\bx(t_k-\Delta T)),
\label{eq:stored_feature_increment}
\end{equation}
\begin{equation}
\begin{aligned}
    \widehat{\mathcal R}_k
    =
    \int_{t_k-\Delta T}^{t_k}
    r\big(
        \bx(\tau),
        \hat{\bu}(\tau),
        \hat{\ba}(\tau),
        \hat{\bd}(\tau)
    \big)d\tau .
\end{aligned}
\label{eq:stored_window_cost}
\end{equation}
Once inserted into the history stack,
\(\Delta\bphi_k\) and \(\widehat{\mathcal R}_k\) are held fixed.

The replayed integral Bellman--Isaacs residual associated with
the \(k\)th stored datum is evaluated using the current critic
weight:
\begin{equation}
    \xi_k(t)
    =
    \hat{\mathbf W}^\top(t)\Delta\bphi_k
    +
    \widehat{\mathcal R}_k,
    \qquad
    k=1,\ldots,N.
\label{eq:replayed_integral_residual}
\end{equation}
Define
\begin{equation}
    \Delta\varepsilon_k
    =
    \varepsilon(\bx(t_k))
    -
    \varepsilon(\bx(t_k-\Delta T))
\label{eq:stored_approximation_increment}
\end{equation}
and
\begin{equation}
    \varepsilon_{B,k}
    =
    \widehat{\mathcal R}_k
    -
    \mathcal R_k^*
    -
    \Delta\varepsilon_k,
\label{eq:stored_BI_approximation_error}
\end{equation}
where
\begin{equation}
\begin{aligned}
    \mathcal R_k^*
    =
    \int_{t_k-\Delta T}^{t_k}
    r\big(
        \bx(\tau),
        \bu^*(\tau),
        \ba^*(\tau),
        \bd^*(\tau)
    \big)d\tau .
\end{aligned}
\label{eq:stored_ideal_window_cost}
\end{equation}
Then
\begin{equation}
    \xi_k(t)
    =
    \widetilde{\mathbf W}^\top(t)\Delta\bphi_k
    +
    \varepsilon_{B,k}.
\label{eq:stored_residual_decomposition}
\end{equation}
By the same argument as in
Lemma~\ref{lem:bounded_integral_residual},
\begin{equation}
    |\varepsilon_{B,k}|
    \leq
    \bar\varepsilon_B,
    \qquad
    k=1,\ldots,N.
\label{eq:stored_BI_error_bound}
\end{equation}

Define the normalization terms
\begin{equation}
\begin{aligned}
    m(t)
    &=
    1+\Delta\bphi^\top(t)\Delta\bphi(t),\\
    m_k
    &=
    1+\Delta\bphi_k^\top\Delta\bphi_k,
    \qquad
    k=1,\ldots,N.
\end{aligned}
\label{eq:normalization_terms}
\end{equation}
The normalized online regressor and residual are
\begin{equation}
    \bpsi(t)
    =
    \frac{\Delta\bphi(t)}{m(t)},
    \qquad
    s(t)
    =
    \frac{\xi(t)}{m(t)},
\label{eq:normalized_online_data}
\end{equation}
whereas the normalized replay regressors and residuals are
\begin{equation}
    \bpsi_k
    =
    \frac{\Delta\bphi_k}{m_k},
    \qquad
    s_k(t)
    =
    \frac{\xi_k(t)}{m_k}.
\label{eq:normalized_replay_data}
\end{equation}

Using \eqref{eq:current_residual_decomposition} and
\eqref{eq:stored_residual_decomposition}, the normalized
residuals satisfy
\begin{equation}
\begin{aligned}
    s(t)
    &=
    \widetilde{\mathbf W}^\top(t)\bpsi(t)
    +
    \bar\varepsilon_B(t),\\
    s_k(t)
    &=
    \widetilde{\mathbf W}^\top(t)\bpsi_k
    +
    \bar\varepsilon_{B,k},
\end{aligned}
\label{eq:normalized_residual_decomposition}
\end{equation}
where
\begin{equation}
    \bar\varepsilon_B(t)
    =
    \frac{\varepsilon_B(t)}{m(t)},
    \qquad
    \bar\varepsilon_{B,k}
    =
    \frac{\varepsilon_{B,k}}{m_k}.
\label{eq:normalized_BI_errors}
\end{equation}
Since \(m(t)\geq1\) and \(m_k\geq1\),
\begin{equation}
    |\bar\varepsilon_B(t)|
    \leq
    \bar\varepsilon_B,
    \qquad
    |\bar\varepsilon_{B,k}|
    \leq
    \bar\varepsilon_B.
\label{eq:normalized_BI_error_bounds}
\end{equation}

% \begin{assumption}[Finite excitation of the replay stack]
% \label{ass:finite_excitation}
% The history stack contains at least \(L\) informative samples and
% satisfies
% \begin{equation}
% \begin{aligned}
%     \mathbf\Psi
%     &:=
%     \sum_{k=1}^{N}
%     \bpsi_k\bpsi_k^\top\\
%     &=
%     \sum_{k=1}^{N}
%     \frac{
%         \Delta\bphi_k\Delta\bphi_k^\top
%     }{
%         \left(
%             1+\Delta\bphi_k^\top\Delta\bphi_k
%         \right)^2
%     },
% \end{aligned}
% \label{eq:replay_Gramian}
% \end{equation}
% with
% \begin{equation}
%     \lambda_{\min}(\mathbf\Psi)
%     \geq
%     \underline\lambda
%     >
%     0.
% \label{eq:finite_excitation_condition}
% \end{equation}
% \end{assumption}

\begin{assumption}[Finite excitation of the replay stack]
\label{ass:finite_excitation}
The history stack satisfies
\(\mathbf\Psi:=\sum_{k=1}^{N}\bpsi_k\bpsi_k^\top=\sum_{k=1}^{N}\dfrac{\Delta\bphi_k\Delta\bphi_k^\top}{\left(1+\Delta\bphi_k^\top\Delta\bphi_k\right)^2}\),
with
\(\lambda_{\min}(\mathbf\Psi)\ge\underline\lambda>0\).
\end{assumption}

\begin{remark}[Finite excitation versus persistent excitation]
\label{rem:finite_vs_persistent_excitation}
Assumption~\ref{ass:finite_excitation} is a finite-data
informativity condition. It does not require the online regressor
to remain persistently exciting for all future time. Once a
finite set of linearly independent integral feature increments has
been collected, those directions are retained in
\(\mathcal H\) and repeatedly injected into the critic update.
The closed-loop trajectory may therefore lose excitation as the
state approaches the target set without eliminating the
rank information already stored in the replay matrix
\(\mathbf\Psi\).
\end{remark}

% ------------------------------------------------------------
\subsection{Fixed-Time Integral Critic Update}
\label{subsec:fixed_time_critic_update}
% ------------------------------------------------------------

Select two residual exponents satisfying
\begin{equation}
    0<q<1<r.
\label{eq:critic_residual_exponents}
\end{equation}
Define the fixed-time integral residual loss
\begin{equation}
\begin{aligned}
    \mathcal E(\hat{\mathbf W})
    ={}&
    \frac{1}{q+1}|s(t)|^{q+1}
    +
    \frac{1}{r+1}|s(t)|^{r+1}\\
    &+
    \sum_{k=1}^{N}
    \left[
        \frac{1}{q+1}|s_k(t)|^{q+1}
        +
        \frac{1}{r+1}|s_k(t)|^{r+1}
    \right].
\end{aligned}
\label{eq:fixed_time_residual_loss}
\end{equation}

\begin{remark}
For completeness, the affine dependence of the integral residual on
the current critic weight follows directly from the saddle-point
stationarity conditions. Let the approximate Hamiltonian associated
with the learned policies be
\(\hat{\mathcal H}
=\ell(\bx,\hat{\bu},\hat{\ba},\hat{\bd})+
\hat{\bW}^{\top}\nabla\bm\phi(\bx)
[ f(\bx)+\bm g(\bx)(\hat{\bu}+\hat{\ba})
+\bm k(\bx)\hat{\bd}]\).
Although
\(\hat{\bu}=\hat{\bu}(\hat{\bW})\),
\(\hat{\ba}=\hat{\ba}(\hat{\bW})\), and
\(\hat{\bd}=\hat{\bd}(\hat{\bW})\), differentiation with respect to
\(\hat{\bW}\) gives the direct term
\(\nabla\bm\phi(\bx)
[ f+\bm g(\hat{\bu}+\hat{\ba})+\bm k\hat{\bd}]\)
plus terms proportional to
\(\partial\hat{\bu}/\partial\hat{\bW}\),
\(\partial\hat{\ba}/\partial\hat{\bW}\), and
\(\partial\hat{\bd}/\partial\hat{\bW}\).
The latter vanish because
\(\partial U(\hat{\bu})/\partial\hat{\bu}
+\bm g^{\top}\nabla\bm\phi^{\top}\hat{\bW}=\bm0\),
\(-2\gamma_a^2\bTmat\hat{\ba}
+\bm g^{\top}\nabla\bm\phi^{\top}\hat{\bW}=\bm0\), and
\(-2\gamma_d^2\bS\hat{\bd}
+\bm k^{\top}\nabla\bm\phi^{\top}\hat{\bW}=\bm0\).
Hence
\(\partial\hat{\mathcal H}/\partial\hat{\bW}
=\nabla\bm\phi(\bx)
[ f+\bm g(\hat{\bu}+\hat{\ba})+\bm k\hat{\bd}]
=\nabla\bm\phi(\bx)\dot{\bx}\).
Over the fixed integration window,
\(\int_{t-\Delta T}^{t}
\nabla\bm\phi(\bx(\tau))\dot{\bx}(\tau)\,d\tau
=\bm\phi(\bx(t))-\bm\phi(\bx(t-\Delta T))\).
Therefore, with the normalization already introduced in the integral
residual, the current regressor is fixed during instantaneous
weight differentiation and
\(\nabla_{\hat{\bW}}s(t)=\bm\psi(t)\).
Likewise, each stored tuple
\((\Delta\bm\phi_k,\hat{\mathcal R}_k)\) is frozen after insertion into the replay
stack, so
\(\nabla_{\hat{\bW}}s_k(t)=\bm\psi_k\).

\end{remark}

Consequently,
\begin{equation}
    \frac{\partial s(t)}
    {\partial\hat{\mathbf W}}
    =
    \frac{\Delta\bphi(t)}{m(t)}
    =
    \bpsi(t),
\label{eq:online_residual_gradient}
\end{equation}
and
\begin{equation}
    \frac{\partial s_k(t)}
    {\partial\hat{\mathbf W}}
    =
    \frac{\Delta\bphi_k}{m_k}
    =
    \bpsi_k.
\label{eq:replay_residual_gradient}
\end{equation}

Applying the chain rule to
\eqref{eq:fixed_time_residual_loss} gives
\begin{equation}
\begin{aligned}
    \frac{\partial\mathcal E}
    {\partial\hat{\mathbf W}}
    ={}&
    \bpsi(t)
    \left[
        |s(t)|^q\sgn(s(t))
        +
        |s(t)|^r\sgn(s(t))
    \right]\\
    &+
    \sum_{k=1}^{N}
    \bpsi_k
    \left[
        |s_k(t)|^q\sgn(s_k(t))
        +
        |s_k(t)|^r\sgn(s_k(t))
    \right].
\end{aligned}
\label{eq:fixed_time_loss_gradient}
\end{equation}

Let \(\mathbf\Gamma=\mathbf\Gamma^\top>0\) be a constant
learning-gain matrix and let \(\sigma\geq0\) be a leakage gain.
The critic-weight update law is selected as
\begin{equation}
\begin{aligned}
    \dot{\hat{\mathbf W}}
    ={}&
    -\mathbf\Gamma\bpsi(t)
    \left[
        |s(t)|^q\sgn(s(t))
        +
        |s(t)|^r\sgn(s(t))
    \right]\\
    &-
    \mathbf\Gamma
    \sum_{k=1}^{N}
    \bpsi_k
    \left[
        |s_k(t)|^q\sgn(s_k(t))
        +
        |s_k(t)|^r\sgn(s_k(t))
    \right]\\
    &-
    \sigma\mathbf\Gamma\hat{\mathbf W}.
\end{aligned}
\label{eq:fixed_time_critic_update}
\end{equation}

\begin{remark}[Role and necessity of leakage]
\label{rem:leakage}
The leakage term \(-\sigma\mathbf{\Gamma}\hat{\mathbf W}\) is included in the present formulation mainly to enhance generality and robustness, following the classical robust adaptive-control philosophy of suppressing parameter drift in the presence of persistent approximation, measurement, or numerical-integration errors. Its use, however, introduces an intrinsic bias because it drives \(\hat{\mathbf W}\) toward the origin rather than exactly toward \(\bW^*\), thereby contributing an additional residual term, typically proportional to \(\sigma\|\bW^*\|^2\), and potentially enlarging the critic and closed-loop ultimate bounds. Importantly, leakage is not the fundamental mechanism responsible for critic-weight convergence here: once the replay stack satisfies the required richness condition, the stored-data term itself provides coercive dissipation in the weight-error direction, in the same spirit as composite learning and integral concurrent learning, where historical or integrated data yield parameter convergence under relaxed finite-data excitation conditions~\cite{PanYu2016CompositeLearning,ParikhKamalapurkarDixon2019ICL}. Therefore, when the replay condition is reliably satisfied and approximation and numerical errors are sufficiently controlled, one may simply choose \(\sigma=0\) in implementation without removing the data-driven convergence mechanism. Conversely, \(\sigma>0\) provides an additional robustness safeguard when such perturbations are non-negligible. Thus, the leakage term is retained here as an optional robustification device rather than an essential requirement of the proposed replay-based learning architecture.
\end{remark}

\begin{remark}[Quantitative role of normalization]
\label{rem:normalization_role}
For every vector \(\Delta\bphi\),
\begin{equation}
\begin{aligned}
    \left\|
        \frac{\Delta\bphi}
        {1+\|\Delta\bphi\|^2}
    \right\|
    =
    \frac{\|\Delta\bphi\|}
    {1+\|\Delta\bphi\|^2}
    \leq
    \frac{1}{2}.
\end{aligned}
\label{eq:normalized_regressor_bound}
\end{equation}
Hence
\begin{equation}
    \|\bpsi(t)\|
    \leq
    \frac{1}{2},
    \qquad
    \|\bpsi_k\|
    \leq
    \frac{1}{2}.
\label{eq:all_normalized_regressor_bounds}
\end{equation}
Normalization therefore prevents large finite-window feature
increments from directly producing excessively large gradient
steps. The update magnitude is governed mainly by the normalized
residual powers rather than by unbounded feature increments.
\end{remark}

\endgroup
% ============================================================

% ============================================================
\section{Stability Analysis}
\label{sec:stability}
% ============================================================

\begingroup
\setlength{\abovedisplayskip}{3pt}
\setlength{\belowdisplayskip}{3pt}
\setlength{\abovedisplayshortskip}{2pt}
\setlength{\belowdisplayshortskip}{2pt}

This section establishes the stability of the proposed critic-only fixed-time integral reinforcement learning scheme. Practical fixed-time boundedness of the critic-weight estimation error is first proved from the normalized integral Bellman--Isaacs residuals and the finite experience-replay stack. This result is then combined with the cost-induced fixed-time structure of the saturated secure HJI problem to establish practical fixed-time stability of the learned nonlinear closed-loop system. The proof preserves the integral learning structure by evaluating the state Lyapunov component over a finite reinforcement window, while keeping the critic component pointwise since the critic weights follow an ordinary differential equation.

% ------------------------------------------------------------
\subsection{Practical Fixed-Time Boundedness of the Critic Error}
\label{subsec:critic_weight_stability}
% ------------------------------------------------------------

\begin{remark}
Fixed-time convergence of the critic weights is not required as an
independent physical objective. It is imposed here because the learned
control policy depends directly on the critic through
\(\nabla\hat V(\bx)\). Hence, the critic error enters the closed-loop
stability analysis as a policy-approximation error. Driving
\(\widetilde{\mathbf W}\) into a bounded residual set within an
initial-condition-independent time prevents an arbitrarily long learning
transient from propagating into the closed-loop convergence bound and
thereby supports the subsequent practical fixed-time stability result.
\end{remark}

Recall from \eqref{eq:normalized_residual_decomposition} that the
normalized current and replayed integral Bellman--Isaacs residuals
satisfy
\begin{equation}
\begin{aligned}
    s(t)
    &=
    \widetilde{\mathbf W}^{\top}(t)\bpsi(t)
    +
    \bar\varepsilon_B(t),\\
    s_k(t)
    &=
    \widetilde{\mathbf W}^{\top}(t)\bpsi_k
    +
    \bar\varepsilon_{B,k},
    \qquad
    k=1,\ldots,N,
\end{aligned}
\label{eq:stab_residual_decomposition}
\end{equation}
% where
% \begin{equation}
%     \widetilde{\mathbf W}
%     =
%     \hat{\mathbf W}-\mathbf W.
% \label{eq:stab_weight_error_recall}
% \end{equation}
% By \eqref{eq:normalized_BI_error_bounds}, there exists a constant
% \(\varepsilon_{BM}>0\) such that
% \begin{equation}
%     |\bar\varepsilon_B(t)|
%     \leq
%     \varepsilon_{BM},
%     \qquad
%     |\bar\varepsilon_{B,k}|
%     \leq
%     \varepsilon_{BM},
%     \qquad
%     k=1,\ldots,N.
% \label{eq:stab_normalized_residual_bound}
% \end{equation}

The replay matrix satisfies
\begin{equation}
    \mathbf\Psi
    =
    \sum_{k=1}^{N}
    \bpsi_k\bpsi_k^\top,
\label{eq:stab_replay_matrix}
\end{equation}
and Assumption~\ref{ass:finite_excitation} guarantees that
\begin{equation}
    \lambda_{\min}(\mathbf\Psi)
    \geq
    \underline\lambda
    >
    0.
\label{eq:stab_replay_rank}
\end{equation}

For convenience, define
\begin{equation}
    \mu
    =
    \frac{q+1}{2},
    \qquad
    \nu
    =
    \frac{r+1}{2},
\label{eq:critic_power_definitions}
\end{equation}
where \(0<q<1<r\). Hence,
\begin{equation}
    0<\mu<1<\nu.
\label{eq:critic_power_ranges}
\end{equation}

\begin{lemma}[Residual perturbation inequality]
\label{lem:residual_perturbation}
For any \(p>0\), there exist constants
\(\kappa_p>0\) and \(c_p>0\) such that, for all
\(z,e\in\mathbb R\),
\begin{equation}
\begin{aligned}
    z
    |z+e|^p
    \sgn(z+e)
    \geq
    \kappa_p|z|^{p+1}
    -
    c_p|e|^{p+1}.
\end{aligned}
\label{eq:residual_perturbation_inequality}
\end{equation}
One admissible selection is
\begin{equation}
    \kappa_p
    =
    2^{-(p+1)},
    \qquad
    c_p
    =
    2\cdot 3^p+1.
\label{eq:residual_perturbation_constants}
\end{equation}
\end{lemma}

% \begin{proof}
% Consider first the case
% \[
%     |z|\geq2|e|.
% \]
% Then \(z+e\) has the same sign as \(z\), and
% \begin{equation}
%     |z+e|
%     \geq
%     |z|-|e|
%     \geq
%     \frac{|z|}{2}.
% \label{eq:perturbation_case_one_magnitude}
% \end{equation}
% Consequently,
% \begin{equation}
% \begin{aligned}
%     z|z+e|^p\sgn(z+e)
%     &=
%     |z||z+e|^p\\
%     &\geq
%     2^{-p}|z|^{p+1}\\
%     &\geq
%     \kappa_p|z|^{p+1}.
% \end{aligned}
% \label{eq:perturbation_case_one}
% \end{equation}

% Consider next the case
% \[
%     |z|<2|e|.
% \]
% By the triangle inequality,
% \begin{equation}
%     |z+e|
%     \leq
%     |z|+|e|
%     <
%     3|e|.
% \label{eq:perturbation_case_two_magnitude}
% \end{equation}
% Therefore,
% \begin{equation}
% \begin{aligned}
%     z|z+e|^p\sgn(z+e)
%     &\geq
%     -|z||z+e|^p\\
%     &>
%     -2|e|(3|e|)^p\\
%     &=
%     -2\cdot3^p|e|^{p+1}.
% \end{aligned}
% \label{eq:perturbation_case_two_first}
% \end{equation}
% Moreover,
% \begin{equation}
%     \kappa_p|z|^{p+1}
%     <
%     \kappa_p2^{p+1}|e|^{p+1}
%     =
%     |e|^{p+1}.
% \label{eq:perturbation_case_two_second}
% \end{equation}
% Combining \eqref{eq:perturbation_case_two_first} and
% \eqref{eq:perturbation_case_two_second} yields
% \begin{equation}
% \begin{aligned}
%     z|z+e|^p\sgn(z+e)
%     \geq
%     \kappa_p|z|^{p+1}
%     -
%     (2\cdot3^p+1)|e|^{p+1}.
% \end{aligned}
% \label{eq:perturbation_case_two_final}
% \end{equation}
% Thus, \eqref{eq:residual_perturbation_inequality} holds in both
% cases.
% \end{proof}

\begin{proof}
Consider first the case \( |z|\ge2|e| \). Then \(z+e\) has the same sign as \(z\), and \( |z+e|\ge|z|-|e|\ge|z|/2 \). Hence,
\( z|z+e|^p\sgn(z+e)=|z||z+e|^p\ge2^{-p}|z|^{p+1}\ge\kappa_p|z|^{p+1} \).
If instead \( |z|<2|e| \), then \( |z+e|\le|z|+|e|<3|e| \), implying
\( z|z+e|^p\sgn(z+e)\ge-|z||z+e|^p>-2|e|(3|e|)^p=-2\cdot3^p|e|^{p+1} \).
Moreover,
\( \kappa_p|z|^{p+1}<\kappa_p2^{p+1}|e|^{p+1}=|e|^{p+1} \).
Therefore,
\( z|z+e|^p\sgn(z+e)\ge\kappa_p|z|^{p+1}-(2\cdot3^p+1)|e|^{p+1} \),
which proves \eqref{eq:residual_perturbation_inequality}.
\end{proof}

% \begin{remark}
% \label{rem:residual_perturbation_role}
% Lemma~\ref{lem:residual_perturbation} is required because the
% integral Bellman--Isaacs residual is not exactly equal to the
% critic-weight projection
% \(\widetilde{\mathbf W}^{\top}\bpsi\). The additional term
% collects value-function approximation error, finite-window
% integration error, and policy-induced residual error. The lemma
% preserves a negative power of the critic-weight projection in the
% Lyapunov derivative while transferring all residual
% imperfections into a bounded additive term.
% \end{remark}

% The next result converts reduction of the replayed scalar
% residuals into reduction of the complete critic-weight vector.

\begin{lemma}[Finite-replay projection bound]
\label{lem:finite_replay_projection}
Suppose \eqref{eq:stab_replay_rank} holds. Then, for every
\(\widetilde{\mathbf W}\in\mathbb R^L\),
\begin{equation}
\begin{aligned}
    \sum_{k=1}^{N}
    \left|
        \widetilde{\mathbf W}^{\top}\bpsi_k
    \right|^{q+1}
    \geq
    \eta_q
    \|\widetilde{\mathbf W}\|^{q+1},
\end{aligned}
\label{eq:replay_projection_low}
\end{equation}
where
\begin{equation}
    \eta_q
    =
    \underline\lambda^{\frac{q+1}{2}},
\label{eq:eta_q_definition}
\end{equation}
and
\begin{equation}
\begin{aligned}
    \sum_{k=1}^{N}
    \left|
        \widetilde{\mathbf W}^{\top}\bpsi_k
    \right|^{r+1}
    \geq
    \eta_r
    \|\widetilde{\mathbf W}\|^{r+1},
\end{aligned}
\label{eq:replay_projection_high}
\end{equation}
where
\begin{equation}
    \eta_r
    =
    N^{1-\frac{r+1}{2}}
    \underline\lambda^{\frac{r+1}{2}}.
\label{eq:eta_r_definition}
\end{equation}
\end{lemma}

\begin{proof}
Define the stacked replay projection
\begin{equation}
    \bz_H
    =
    \begin{bmatrix}
        \widetilde{\mathbf W}^{\top}\bpsi_1&
        \widetilde{\mathbf W}^{\top}\bpsi_2&
        \cdots&
        \widetilde{\mathbf W}^{\top}\bpsi_N
    \end{bmatrix}^{\top}.
\label{eq:stacked_replay_projection}
\end{equation}
Its Euclidean norm satisfies
\begin{equation}
\begin{aligned}
    \|\bz_H\|_2^2
    &=
    \sum_{k=1}^{N}
    \left|
        \widetilde{\mathbf W}^{\top}\bpsi_k
    \right|^2\\
    &=
    \widetilde{\mathbf W}^{\top}
    \left(
        \sum_{k=1}^{N}
        \bpsi_k\bpsi_k^\top
    \right)
    \widetilde{\mathbf W}\\
    &=
    \widetilde{\mathbf W}^{\top}
    \mathbf\Psi
    \widetilde{\mathbf W}\\
    &\geq
    \underline\lambda
    \|\widetilde{\mathbf W}\|^2.
\end{aligned}
\label{eq:stacked_replay_rank_bound}
\end{equation}

Since \(q+1\in(1,2)\), monotonicity of finite-dimensional
\(p\)-norms gives
\begin{equation}
    \|\bz_H\|_{q+1}
    \geq
    \|\bz_H\|_2.
\label{eq:low_norm_monotonicity}
\end{equation}
Consequently,
\begin{equation}
\begin{aligned}
    \sum_{k=1}^{N}
    \left|
        \widetilde{\mathbf W}^{\top}\bpsi_k
    \right|^{q+1}
    &=
    \|\bz_H\|_{q+1}^{q+1}\\
    &\geq
    \|\bz_H\|_2^{q+1}\\
    &\geq
    \underline\lambda^{\frac{q+1}{2}}
    \|\widetilde{\mathbf W}\|^{q+1}.
\end{aligned}
\label{eq:low_projection_proof}
\end{equation}
This proves \eqref{eq:replay_projection_low}.

Since \(r+1>2\), the finite-dimensional norm relation gives
\begin{equation}
    \|\bz_H\|_{r+1}
    \geq
    N^{\frac{1}{r+1}-\frac12}
    \|\bz_H\|_2.
\label{eq:high_norm_relation}
\end{equation}
Raising \eqref{eq:high_norm_relation} to the power \(r+1\) and
using \eqref{eq:stacked_replay_rank_bound} yield
\begin{equation}
\begin{aligned}
    \sum_{k=1}^{N}
    \left|
        \widetilde{\mathbf W}^{\top}\bpsi_k
    \right|^{r+1}
    &=
    \|\bz_H\|_{r+1}^{r+1}\\
    &\geq
    N^{1-\frac{r+1}{2}}
    \|\bz_H\|_2^{r+1}\\
    &\geq
    N^{1-\frac{r+1}{2}}
    \underline\lambda^{\frac{r+1}{2}}
    \|\widetilde{\mathbf W}\|^{r+1}.
\end{aligned}
\label{eq:high_projection_proof}
\end{equation}
This proves \eqref{eq:replay_projection_high}.
\end{proof}

\begin{theorem}[Practical fixed-time convergence of the critic]
\label{thm:critic_fixed_time}
Suppose Assumptions~\ref{ass:critic_approximation},
\ref{ass:finite_excitation} hold. Let the critic weights be
updated according to \eqref{eq:fixed_time_critic_update}. Then
the critic-weight estimation error
\(\widetilde{\mathbf W}\) is practically fixed-time uniformly
ultimately bounded.

More precisely, consider
\begin{equation}
    V_W
    =
    \frac{1}{2}
    \widetilde{\mathbf W}^{\top}
    \mathbf\Gamma^{-1}
    \widetilde{\mathbf W}.
\label{eq:critic_weight_Lyapunov}
\end{equation}
There exist constants \(h_q>0\), \(h_r>0\), and
\(\Delta_W\geq0\) such that
\begin{equation}
    \dot V_W
    \leq
    -h_qV_W^\mu
    -h_rV_W^\nu
    +
    \Delta_W.
\label{eq:critic_final_differential_inequality}
\end{equation}
The constants may be selected as
\begin{equation}
\begin{aligned}
    h_q
    &=
    \kappa_q\eta_q
    (2\underline\gamma)^\mu,\\
    h_r
    &=
    \kappa_r\eta_r
    (2\underline\gamma)^\nu.
\end{aligned}
\label{eq:critic_decay_coefficients}
\end{equation}
Using the constants \(\kappa_q\) and \(\kappa_r\) introduced in
the residual bounds, and defining \(\|\mathbf W\|=W_M\), we obtain
\begin{equation}
\begin{aligned}
    \Delta_W
    ={}&
    (N+1)
    \left(
        c_q\varepsilon_{BM}^{q+1}
        +
        c_r\varepsilon_{BM}^{r+1}
    \right)\\
    &+
    \frac{\sigma}{2}W_M^2.
\end{aligned}
\label{eq:critic_residual_constant}
\end{equation}

Consequently, for every \(\vartheta_W\in(0,1)\),
\(V_W(t)\) enters the residual set
\begin{equation}
\begin{aligned}
    \Omega_W(\vartheta_W)
    =
    \bigg\{
        V_W\geq0:\;
        h_qV_W^\mu+h_rV_W^\nu
        \leq
        \frac{\Delta_W}{\vartheta_W}
    \bigg\}
\end{aligned}
\label{eq:critic_residual_set}
\end{equation}
within the fixed time
\begin{equation}
\begin{aligned}
    T_W
    \leq{}&
    \frac{1}
    {(1-\vartheta_W)h_q(1-\mu)}
    +
    \frac{1}
    {(1-\vartheta_W)h_r(\nu-1)}.
\end{aligned}
\label{eq:critic_settling_time}
\end{equation}
The bound in \eqref{eq:critic_settling_time} is independent of
\(\widetilde{\mathbf W}(0)\). If
\(\varepsilon_{BM}=0\) and \(\sigma=0\), then
\(\Delta_W=0\), and the critic-weight error converges exactly to
the origin in fixed time.
\end{theorem}

\begin{proof}
Since the ideal critic weight \(\mathbf W\) is constant,
\begin{equation}
    \dot{\widetilde{\mathbf W}}
    =
    \dot{\hat{\mathbf W}}.
\label{eq:weight_error_derivative}
\end{equation}
Differentiating \eqref{eq:critic_weight_Lyapunov} gives
\begin{equation}
\begin{aligned}
    \dot V_W
    &=
    \widetilde{\mathbf W}^{\top}
    \mathbf\Gamma^{-1}
    \dot{\widetilde{\mathbf W}}\\
    &=
    \widetilde{\mathbf W}^{\top}
    \mathbf\Gamma^{-1}
    \dot{\hat{\mathbf W}}.
\end{aligned}
\label{eq:critic_Lyapunov_derivative_start}
\end{equation}
Substituting the update law
\eqref{eq:fixed_time_critic_update} into
\eqref{eq:critic_Lyapunov_derivative_start} yields
\begin{equation}
\begin{aligned}
    \dot V_W
    ={}&
    -\widetilde{\mathbf W}^{\top}\bpsi(t)
    \left[
        |s(t)|^q\sgn(s(t))
        +
        |s(t)|^r\sgn(s(t))
    \right]\\
    &-
    \sum_{k=1}^{N}
    \widetilde{\mathbf W}^{\top}\bpsi_k
    \left[
        |s_k(t)|^q\sgn(s_k(t))
        +
        |s_k(t)|^r\sgn(s_k(t))
    \right]\\
    &-
    \sigma
    \widetilde{\mathbf W}^{\top}
    \hat{\mathbf W}.
\end{aligned}
\label{eq:critic_Lyapunov_expanded}
\end{equation}

Define the scalar projections
\begin{equation}
    z(t)
    =
    \widetilde{\mathbf W}^{\top}(t)\bpsi(t),
    \qquad
    z_k(t)
    =
    \widetilde{\mathbf W}^{\top}(t)\bpsi_k.
\label{eq:critic_scalar_projections}
\end{equation}
Then \eqref{eq:stab_residual_decomposition} becomes
\begin{equation}
    s(t)
    =
    z(t)+\bar\varepsilon_B(t),
    \qquad
    s_k(t)
    =
    z_k(t)+\bar\varepsilon_{B,k}.
\label{eq:critic_residual_projection_form}
\end{equation}

Applying Lemma~\ref{lem:residual_perturbation} with \(p=q\)
gives
\begin{equation}
\begin{aligned}
    z(t)|s(t)|^q\sgn(s(t))
    \geq
    \kappa_q|z(t)|^{q+1}
    -
    c_q|\bar\varepsilon_B(t)|^{q+1},
\end{aligned}
\label{eq:online_low_residual_bound}
\end{equation}
and applying the same lemma with \(p=r\) gives
\begin{equation}
\begin{aligned}
    z(t)|s(t)|^r\sgn(s(t))
    \geq
    \kappa_r|z(t)|^{r+1}
    -
    c_r|\bar\varepsilon_B(t)|^{r+1}.
\end{aligned}
\label{eq:online_high_residual_bound}
\end{equation}
Similarly, for every replay datum,
\begin{equation}
\begin{aligned}
    z_k(t)|s_k(t)|^q\sgn(s_k(t))
    \geq
    \kappa_q|z_k(t)|^{q+1}
    -
    c_q|\bar\varepsilon_{B,k}|^{q+1},
\end{aligned}
\label{eq:replay_low_residual_bound}
\end{equation}
and
\begin{equation}
\begin{aligned}
    z_k(t)|s_k(t)|^r\sgn(s_k(t))
    \geq
    \kappa_r|z_k(t)|^{r+1}
    -
    c_r|\bar\varepsilon_{B,k}|^{r+1}.
\end{aligned}
\label{eq:replay_high_residual_bound}
\end{equation}

For the leakage term, use
\begin{equation}
    \hat{\mathbf W}
    =
    \widetilde{\mathbf W}
    +
    \mathbf W.
\label{eq:weight_estimate_decomposition}
\end{equation}
It follows that
\begin{equation}
\begin{aligned}
    -\sigma
    \widetilde{\mathbf W}^{\top}
    \hat{\mathbf W}
    &=
    -\sigma
    \|\widetilde{\mathbf W}\|^2
    -
    \sigma
    \widetilde{\mathbf W}^{\top}\mathbf W\\
    &\leq
    -\frac{\sigma}{2}
    \|\widetilde{\mathbf W}\|^2
    +
    \frac{\sigma}{2}
    \|\mathbf W\|^2\\
    &\leq
    -\frac{\sigma}{2}
    \|\widetilde{\mathbf W}\|^2
    +
    \frac{\sigma}{2}W_M^2.
\end{aligned}
\label{eq:leakage_term_bound}
\end{equation}

Substituting \eqref{eq:online_low_residual_bound}--
\eqref{eq:replay_high_residual_bound} and
\eqref{eq:leakage_term_bound} into
\eqref{eq:critic_Lyapunov_expanded}, and then using
\eqref{eq:online_high_residual_bound}, give
\begin{equation}
\begin{aligned}
    \dot V_W
    \leq{}&
    -\kappa_q|z(t)|^{q+1}
    -\kappa_r|z(t)|^{r+1}\\
    &-
    \kappa_q
    \sum_{k=1}^{N}|z_k(t)|^{q+1}
    -
    \kappa_r
    \sum_{k=1}^{N}|z_k(t)|^{r+1}\\
    &-
    \frac{\sigma}{2}
    \|\widetilde{\mathbf W}\|^2
    +
    \Delta_W.
\end{aligned}
\label{eq:critic_derivative_before_replay_bound}
\end{equation}

The current-data terms and the negative leakage contribution are
nonpositive. They may therefore be dropped to obtain the
conservative estimate
\begin{equation}
\begin{aligned}
    \dot V_W
    \leq{}&
    -\kappa_q
    \sum_{k=1}^{N}
    \left|
        \widetilde{\mathbf W}^{\top}\bpsi_k
    \right|^{q+1}\\
    &-
    \kappa_r
    \sum_{k=1}^{N}
    \left|
        \widetilde{\mathbf W}^{\top}\bpsi_k
    \right|^{r+1}
    +
    \Delta_W.
\end{aligned}
\label{eq:critic_derivative_replay_only}
\end{equation}
Applying Lemma~\ref{lem:finite_replay_projection} gives
\begin{equation}
\begin{aligned}
    \dot V_W
    \leq{}&
    -\kappa_q\eta_q
    \|\widetilde{\mathbf W}\|^{q+1}\\
    &-
    \kappa_r\eta_r
    \|\widetilde{\mathbf W}\|^{r+1}
    +
    \Delta_W.
\end{aligned}
\label{eq:critic_derivative_weight_norm}
\end{equation}

% From Assumption~\ref{ass:learning_gain_matrix},

Since \(\mathbf\Gamma=\mathbf\Gamma^\top>0\) is constant, letting
\(\underline\gamma:=\lambda_{\min}(\mathbf\Gamma)\) and
\(\bar\gamma:=\lambda_{\max}(\mathbf\Gamma)\) gives
\begin{equation}
    \underline\gamma\mathbf I_L
    \leq
    \mathbf\Gamma
    \leq
    \bar\gamma\mathbf I_L,
\label{eq:learning_gain_bounds}
\end{equation}
where \(0<\underline\gamma\leq\bar\gamma\).

\begin{equation}
\begin{aligned}
    \frac{1}{2\bar\gamma}
    \|\widetilde{\mathbf W}\|^2
    \leq
    V_W
    \leq
    \frac{1}{2\underline\gamma}
    \|\widetilde{\mathbf W}\|^2.
\end{aligned}
\label{eq:critic_Lyapunov_eigenvalue_bounds}
\end{equation}
The upper inequality in
\eqref{eq:critic_Lyapunov_eigenvalue_bounds} implies
\begin{equation}
    \|\widetilde{\mathbf W}\|^2
    \geq
    2\underline\gamma V_W.
\label{eq:weight_norm_lower_from_Lyapunov}
\end{equation}
Therefore,
\begin{equation}
\begin{aligned}
    \|\widetilde{\mathbf W}\|^{q+1}
    &\geq
    (2\underline\gamma)^\mu
    V_W^\mu,\\
    \|\widetilde{\mathbf W}\|^{r+1}
    &\geq
    (2\underline\gamma)^\nu
    V_W^\nu.
\end{aligned}
\label{eq:critic_weight_power_relations}
\end{equation}
Substitution of \eqref{eq:critic_weight_power_relations} into
\eqref{eq:critic_derivative_weight_norm} yields
\begin{equation}
    \dot V_W
    \leq
    -h_qV_W^\mu
    -h_rV_W^\nu
    +
    \Delta_W.
\label{eq:critic_fixed_time_comparison_form}
\end{equation}
This is precisely the differential inequality stated in
\eqref{eq:critic_final_differential_inequality}.

Deriving \eqref{eq:critic_residual_set}, \eqref{eq:critic_settling_time} from expression \eqref{eq:critic_final_differential_inequality} and Theorem \ref{thm:fixed_time_value_condition} completes the proof.

\end{proof}

\subsection{Practical Fixed-Time Stability of the Learned System}
\label{subsec:closed_loop_stability}
% ------------------------------------------------------------

We next establish the state-side stability result. Unlike the
critic proof, the physical closed-loop analysis uses the learned
control policy \(\hat{\bu}\) and the actual FDI and disturbance
signals \(\ba\) and \(\bd\). The learned policies
\(\hat{\ba}\) and \(\hat{\bd}\) introduced in
Section~\ref{sec:fixed_time_irl} are virtual maximizing policies
used in construction of the secure Bellman--Isaacs residual; they
are not additional physical control inputs.

% \begin{assumption}[Closed-loop compactness and adversarial bound]
% \label{ass:closed_loop_compactness}
% Under the learned policy \(\hat{\bu}\), the closed-loop solution
% exists uniquely and remains in the compact set \(\Omega^r\). The
% actual FDI and disturbance signals satisfy
% \begin{equation}
%     \gamma_a^2
%     \ba^\top(t)\mathbf T\ba(t)
%     +
%     \gamma_d^2
%     \bd^\top(t)\mathbf S\bd(t)
%     \leq
%     \bar d,
% \label{eq:adversarial_energy_rate_bound}
% \end{equation}
% for some constant \(\bar d\geq0\) and all times under
% consideration.
% \end{assumption}

% \begin{lemma}[Adversarial-signal bound]
% \label{lem:adversarial_bound}
% Under Assumption~\ref{ass:signal_regular}, there exists a finite
% constant \(\bar\delta_{ad}\geq0\) such that
% \begin{equation}
% \gamma_a^2\ba^\top(t)\mathbf T\ba(t)
% +
% \gamma_d^2\bd^\top(t)\mathbf S\bd(t)
% \leq
% \bar\delta_{ad}
% \label{eq:adversarial_energy_rate_bound}
% \end{equation}
% for almost all \(t\geq0\).
% \end{lemma}

% \begin{proof}
% Since \(\ba\in L_\infty\) and \(\bd\in L_\infty\), there exist finite
% constants \(\bar a,\bar d\geq0\) such that
% \(\|\ba(t)\|\leq\bar a\) and \(\|\bd(t)\|\leq\bar d\) a.e. Hence,
% \[
% \bar\delta_{ad}
% =
% \gamma_a^2\lambda_{\max}(\mathbf T)\bar a^2
% +
% \gamma_d^2\lambda_{\max}(\mathbf S)\bar d^2
% \]
% satisfies \eqref{eq:adversarial_energy_rate_bound}.
% \end{proof}

\begin{lemma}[Adversarial-signal bound]
\label{lem:adversarial_bound}
Under Assumption~\ref{ass:signal_regular}, there exists a finite
constant \(\bar d\geq0\) such that
\begin{equation}
\gamma_a^2
\ba^\top(t)\mathbf T\ba(t)
+
\gamma_d^2
\bd^\top(t)\mathbf S\bd(t)
\leq
\bar d
\label{eq:adversarial_energy_rate_bound}
\end{equation}
for almost all \(t\geq0\).
\end{lemma}

\begin{proof}
Since \(\ba\in L_\infty([0,\infty);\mathbb R^m)\) and
\(\bd\in L_\infty([0,\infty);\mathbb R^n)\), there exist finite
constants \(\bar a,\bar d_d\geq0\) such that
\(\|\ba(t)\|\leq\bar a\) and \(\|\bd(t)\|\leq\bar d_d\) almost
everywhere. Hence, \eqref{eq:adversarial_energy_rate_bound} holds with
\[
\bar d
=
\gamma_a^2\lambda_{\max}(\mathbf T)\bar a^2
+
\gamma_d^2\lambda_{\max}(\mathbf S)\bar d_d^2.
\]
\end{proof}

The state penalty is now selected according to the fixed-time
shaping introduced in
\eqref{eq:fixed_time_state_penalty}:
\refstepcounter{equation}\label{eq:stability_fixed_time_state_cost}
\(Q(\bx)
=
\bx^\top\mathbf Q_x\bx
+
\kappa_1\|\bx\|^{2\alpha}
+
\kappa_2\|\bx\|^{2\beta}\).
Here,
\refstepcounter{equation}\label{eq:stability_state_cost_parameters}
\(\mathbf Q_x
=
\mathbf Q_x^\top>0\),
\(\kappa_1,\kappa_2>0\), and
\(0<\alpha<1<\beta\).

To align the critic fixed-time powers with the state-side powers,
select
\refstepcounter{equation}\label{eq:matched_critic_state_powers}
\(q=\alpha\) and \(r=2\beta-1\).
Then
\refstepcounter{equation}\label{eq:matched_low_power}
\(\rho
:=
\frac{q+1}{2}
=
\frac{\alpha+1}{2}
\in(0,1)\),
and
\refstepcounter{equation}\label{eq:matched_high_power}
\(\theta
:=
\frac{r+1}{2}
=
\beta
>
1\).

% \begin{remark}[Matched state and critic powers]
% \label{rem:matched_state_critic_powers}
% The selection \eqref{eq:matched_critic_state_powers} is not
% required for critic convergence alone. It is introduced to make
% the critic Lyapunov powers identical to the powers obtained from
% the finite-window state analysis. The state cost produces the
% powers
% \((\alpha+1)/2\) and \(\beta\), while the critic update produces
% \((q+1)/2\) and \((r+1)/2\). The matching condition allows both
% parts to be combined into one fixed-time Lyapunov inequality.
% \end{remark}

We have
\refstepcounter{equation}\label{eq:value_gradient_mismatch}
\(\nabla\hat V(\bx)-\nabla V^*(\bx)
=
\big(\nabla\bphi(\bx)\big)^\top
\widetilde{\mathbf W}
-
\nabla\varepsilon(\bx)\).
Using
\eqref{eq:u_star} and \eqref{eq:learned_control_policy},
\refstepcounter{equation}\label{eq:learned_control_mismatch_bound}
\(\|\hat{\bu}-\bu^*\|
\leq
\frac{1}{2}
\|\mathbf R^{-1}\|
\|\mathbf g(\bx)\|
\|
\nabla\hat V(\bx)-\nabla V^*(\bx)
\|
\leq
m_W\|\widetilde{\mathbf W}\|
+
\varepsilon_u\),
where
\refstepcounter{equation}\label{eq:control_weight_mismatch_coefficient}
\(m_W
=
\frac{1}{2}
\|\mathbf R^{-1}\|
\bar g\phi_{gM}\),
and
\refstepcounter{equation}\label{eq:control_approximation_mismatch}
\(\varepsilon_u
=
\frac{1}{2}
\|\mathbf R^{-1}\|
\bar g\bar\varepsilon_g\).

% \begin{remark}
% \label{rem:no_actor_error}
% The control-policy error in
% \eqref{eq:learned_control_mismatch_bound} contains no independent
% actor-weight error. It is determined entirely by the
% critic-weight estimation error and the value-function gradient
% approximation error. This is why the subsequent composite
% Lyapunov functional requires only one adaptive-weight component.
% \end{remark}

Let
\refstepcounter{equation}\label{eq:state_quadratic_minimum}
\(q_x
=
\lambda_{\min}(\mathbf Q_x)>0\).
For an arbitrary constant \(\epsilon_x\in(0,1)\), Young's
inequality gives
\refstepcounter{equation}\label{eq:state_weight_cross_bound}
\(c_{Vg}\bar g\,m_W
\|\bx\|
\|\widetilde{\mathbf W}\|
\leq
\frac{\epsilon_x}{2}
\bx^\top\mathbf Q_x\bx
+
\frac{c_{Vg}^2\bar g^2m_W^2}
{2\epsilon_xq_x}
\|\widetilde{\mathbf W}\|^2\),
and
\refstepcounter{equation}\label{eq:state_approximation_cross_bound}
\(c_{Vg}\bar g\,\varepsilon_u\|\bx\|
\leq
\frac{\epsilon_x}{2}
\bx^\top\mathbf Q_x\bx
+
\frac{c_{Vg}^2\bar g^2\varepsilon_u^2}
{2\epsilon_xq_x}\).
Define
\refstepcounter{equation}\label{eq:state_weight_coupling_coefficient}
\(\ell_W
=
\frac{c_{Vg}^2\bar g^2m_W^2}
{2\epsilon_xq_x}\),
and
\refstepcounter{equation}\label{eq:state_control_residual}
\(\Delta_u
=
\frac{c_{Vg}^2\bar g^2\varepsilon_u^2}
{2\epsilon_xq_x}\).

\begin{lemma}[Derivative-limited fractional-power window bound]
\label{lem:fractional_window_bound}
Let \(\xi:[t-\Delta T,t]\to\mathbb R_{\geq0}\) be continuously
differentiable and satisfy
\begin{equation}
    |\dot\xi(\tau)|
    \leq
    L_\xi,
    \qquad
    0\leq
    \xi(\tau)
    \leq
    \xi_{\max},
\label{eq:window_signal_bounds}
\end{equation}
for all \(\tau\in[t-\Delta T,t]\), where \(\Delta T>0\). Define
\begin{equation}
    I_\xi(t)
    =
    \int_{t-\Delta T}^{t}\xi(\tau)d\tau.
\label{eq:window_signal_integral}
\end{equation}
Then, for every \(0<\alpha<1\),
\begin{equation}
    \int_{t-\Delta T}^{t}
    \xi^\alpha(\tau)d\tau
    \geq
    \omega_\alpha
    I_\xi^{\frac{\alpha+1}{2}}(t),
\label{eq:fractional_window_lower_bound}
\end{equation}
where
\begin{equation}
    \omega_\alpha
    =
    C_\xi^{\alpha-1},
\label{eq:window_low_power_coefficient}
\end{equation}
and
\begin{equation}
    C_\xi
    =
    \max
    \left\{
        2\sqrt{L_\xi},
        \;
        4\sqrt{\frac{\xi_{\max}}{\Delta T}}
    \right\}.
\label{eq:window_signal_constant}
\end{equation}
\end{lemma}

\begin{proof}
Define
\begin{equation}
    M_\xi(t)
    =
    \max_{\tau\in[t-\Delta T,t]}\xi(\tau).
\label{eq:window_maximum}
\end{equation}
If \(M_\xi(t)=0\), then
\(\xi(\tau)=0\) throughout the window and
\eqref{eq:fractional_window_lower_bound} holds trivially.
Suppose \(M_\xi(t)>0\), and let
\(\tau_M\in[t-\Delta T,t]\) satisfy
\begin{equation}
    \xi(\tau_M)
    =
    M_\xi(t).
\label{eq:window_maximum_location}
\end{equation}

At least one of the two intervals adjacent to \(\tau_M\) and
contained in \([t-\Delta T,t]\) has length no smaller than \(\Delta T/2\).
Along that interval, the derivative bound in
\eqref{eq:window_signal_bounds} implies that
\begin{equation}
    \xi(\tau)
    \geq
    \frac{M_\xi(t)}{2}
\label{eq:half_maximum_bound}
\end{equation}
over a subinterval of length at least
\begin{equation}
    \delta_\xi(t)
    =
    \min
    \left\{
        \frac{M_\xi(t)}{2L_\xi},
        \frac{\Delta T}{2}
    \right\}.
\label{eq:half_maximum_interval}
\end{equation}
Consequently,
\begin{equation}
    I_\xi(t)
    \geq
    \frac{M_\xi(t)}{2}
    \delta_\xi(t).
\label{eq:window_integral_from_maximum}
\end{equation}

If \(M_\xi(t)\leq L_\xi\Delta T\), then
\[
    \delta_\xi(t)
    =
    \frac{M_\xi(t)}{2L_\xi},
\]
and \eqref{eq:window_integral_from_maximum} gives
\begin{equation}
    I_\xi(t)
    \geq
    \frac{M_\xi^2(t)}
    {4L_\xi}.
\label{eq:window_maximum_case_one}
\end{equation}
Thus,
\begin{equation}
    M_\xi(t)
    \leq
    2\sqrt{L_\xi I_\xi(t)}.
\label{eq:window_maximum_case_one_bound}
\end{equation}

If \(M_\xi(t)>L_\xi\Delta T\), then
\[
    \delta_\xi(t)
    =
    \frac{\Delta T}{2},
\]
and \eqref{eq:window_integral_from_maximum} gives
\begin{equation}
    I_\xi(t)
    \geq
    \frac{\Delta T}{4}M_\xi(t).
\label{eq:window_maximum_case_two}
\end{equation}
Hence,
\begin{equation}
    M_\xi(t)
    \leq
    \frac{4}{\Delta T}I_\xi(t).
\label{eq:window_maximum_case_two_first}
\end{equation}
Since
\[
    I_\xi(t)
    \leq
    \Delta T\xi_{\max},
\]
one has
\begin{equation}
\begin{aligned}
    \frac{4}{\Delta T}I_\xi(t)
    &=
    4I_\xi^{1/2}(t)
    \frac{I_\xi^{1/2}(t)}{\Delta T}\\
    &\leq
    4\sqrt{\frac{\xi_{\max}}{\Delta T}}
    I_\xi^{1/2}(t).
\end{aligned}
\label{eq:window_maximum_case_two_bound}
\end{equation}
Combining both cases yields
\begin{equation}
    M_\xi(t)
    \leq
    C_\xi I_\xi^{1/2}(t).
\label{eq:window_maximum_unified_bound}
\end{equation}

Since \(0\leq\xi(\tau)\leq M_\xi(t)\) and
\(\alpha-1<0\),
\begin{equation}
    \xi^{\alpha-1}(\tau)
    \geq
    M_\xi^{\alpha-1}(t)
\label{eq:fractional_negative_power_bound}
\end{equation}
whenever \(\xi(\tau)>0\). Therefore,
\begin{equation}
\begin{aligned}
    \int_{t-\Delta T}^{t}\xi^\alpha(\tau)d\tau
    &=
    \int_{t-\Delta T}^{t}
    \xi^{\alpha-1}(\tau)\xi(\tau)d\tau\\
    &\geq
    M_\xi^{\alpha-1}(t)
    I_\xi(t).
\end{aligned}
\label{eq:fractional_integral_maximum_bound}
\end{equation}
Because \(\alpha-1<0\), inequality
\eqref{eq:window_maximum_unified_bound} implies
\begin{equation}
    M_\xi^{\alpha-1}(t)
    \geq
    C_\xi^{\alpha-1}
    I_\xi^{\frac{\alpha-1}{2}}(t).
\label{eq:maximum_negative_power_conversion}
\end{equation}
Substitution into
\eqref{eq:fractional_integral_maximum_bound} gives
\begin{equation}
\begin{aligned}
    \int_{t-\Delta T}^{t}\xi^\alpha(\tau)d\tau
    \geq
    C_\xi^{\alpha-1}
    I_\xi^{\frac{\alpha+1}{2}}(t),
\end{aligned}
\label{eq:fractional_window_proof_final}
\end{equation}
which proves \eqref{eq:fractional_window_lower_bound}.
\end{proof}

\begin{remark}
\label{rem:fractional_window_interpretation}
For \(0<\alpha<1\), the mapping \(s\mapsto s^\alpha\) is
concave. Jensen's inequality therefore provides an upper bound on
the integral of \(s^\alpha\), whereas the Lyapunov analysis
requires a lower bound. Lemma~\ref{lem:fractional_window_bound}
uses the bounded rate of variation of the value function to rule
out arbitrarily narrow energy spikes. This removes the need to
postulate an independent moving-window coercivity condition.
\end{remark}

\begin{lemma}[Applied Jensen’s inequality]
\label{lem:convex_window_bound}
Let \(\xi:[t-\Delta T,t]\to\mathbb R_{\geq0}\) be integrable and let
\(\beta>1\). Then
\begin{equation}
\begin{aligned}
    \int_{t-\Delta T}^{t}
    \xi^\beta(\tau)d\tau
    \geq
    (\Delta T)^{1-\beta}
    \left(
        \int_{t-\Delta T}^{t}\xi(\tau)d\tau
    \right)^\beta.
\end{aligned}
\label{eq:convex_window_lower_bound}
\end{equation}
\end{lemma}

\begin{proof}
The function \(s\mapsto s^\beta\) is convex on
\(\mathbb R_{\geq0}\). Jensen's inequality gives
\begin{equation}
\begin{aligned}
    \frac{1}{\Delta T}
    \int_{t-\Delta T}^{t}\xi^\beta(\tau)d\tau
    \geq
    \left[
        \frac{1}{\Delta T}
        \int_{t-\Delta T}^{t}\xi(\tau)d\tau
    \right]^\beta.
\end{aligned}
\label{eq:convex_window_Jensen}
\end{equation}
Multiplication by \(\Delta T\) proves
\eqref{eq:convex_window_lower_bound}.
\end{proof}

\begin{lemma}[Mixed-power domination]
\label{lem:mixed_power_domination}
Let \(s\geq0\) and \(0<\rho<1<\theta\). Then
\begin{equation}
    s
    \leq
    s^\rho+s^\theta.
\label{eq:linear_mixed_power_bound}
\end{equation}
Moreover, for all \(x,y\geq0\),
\begin{equation}
    x^\rho+y^\rho
    \geq
    (x+y)^\rho,
\label{eq:fractional_subadditivity}
\end{equation}
and
\begin{equation}
    x^\theta+y^\theta
    \geq
    2^{1-\theta}(x+y)^\theta.
\label{eq:convex_power_sum}
\end{equation}
\end{lemma}

% \begin{proof}
% If \(0\leq s\leq1\), then \(s^\rho\geq s\). If \(s\geq1\),
% then \(s^\theta\geq s\). Hence
% \eqref{eq:linear_mixed_power_bound} follows.

% Inequality \eqref{eq:fractional_subadditivity} follows from the
% subadditivity of fractional powers for \(0<\rho<1\). For
% \(\theta>1\), convexity gives
% \[
%     \left(\frac{x+y}{2}\right)^\theta
%     \leq
%     \frac{x^\theta+y^\theta}{2},
% \]
% which is equivalent to
% \eqref{eq:convex_power_sum}.
% \end{proof}

\begin{proof} For the first inequality, consider separately the cases \(0\le s\le 1\) and \(s\ge 1\). Since \(0<\rho<1\), for \(0\le s\le 1\) one has \(s^\rho\ge s\), and hence \(s\le s^\rho\le s^\rho+s^\theta\). On the other hand, since \(\theta>1\), for \(s\ge 1\) one has \(s^\theta\ge s\), which similarly yields \(s\le s^\theta\le s^\rho+s^\theta\). Therefore, \(s\le s^\rho+s^\theta\) for every \(s\ge0\). To prove the second inequality, let \(x,y\ge0\). If \(x+y=0\), then \(x=y=0\) and the result is immediate. Otherwise, define \(a=x/(x+y)\) and \(b=y/(x+y)\), so that \(a,b\in[0,1]\) and \(a+b=1\). Since \(0<\rho<1\), one has \(a^\rho\ge a\) and \(b^\rho\ge b\), and therefore \(a^\rho+b^\rho\ge a+b=1\). Multiplying both sides by \((x+y)^\rho\) gives \(x^\rho+y^\rho=(x+y)^\rho(a^\rho+b^\rho)\ge(x+y)^\rho\). Finally, because \(\theta>1\), the function \(r\mapsto r^\theta\) is convex on \([0,\infty)\). Jensen's inequality therefore gives \(((x+y)/2)^\theta\le(x^\theta+y^\theta)/2\), or equivalently \(x^\theta+y^\theta\ge2^{1-\theta}(x+y)^\theta\). This establishes all three inequalities. \end{proof}

\begin{theorem}[Practical fixed-time stability of the learned nonlinear system]
\label{thm:closed_loop_fixed_time}
Suppose Assumptions~\ref{ass:signal_regular},~\ref{ass:system_regular}, \ref{ass:critic_approximation}, and \ref{ass:finite_excitation} hold. Let the state penalty be selected as in \eqref{eq:stability_fixed_time_state_cost}, and let the critic powers satisfy \eqref{eq:matched_critic_state_powers}. Define the moving-window value energy \(V_I(t)=\int_{t-\Delta T}^{t}V^*(\bx(\tau))d\tau\), where \(\Delta T>0\), and consider the mixed Lyapunov functional
\begin{equation}
    \mathcal J(t)=V_I(t)+\lambda_WV_W(t),
\label{eq:mixed_Lyapunov_functional}
\end{equation}
where \(V_W\) is defined in \eqref{eq:critic_weight_Lyapunov}. Choose \(\lambda_W>0\) such that
\begin{equation}
    \lambda_W>
    \max\left\{
        \frac{e_W\Delta T}{h_q},
        \frac{e_W\Delta T}{h_r}
    \right\},
\label{eq:lambda_W_condition}
\end{equation}
where \(e_W>0\) is defined in the proof. Then there exist constants \(c_1>0\), \(c_2>0\), and \(\Delta_c\geq0\) such that
\begin{equation}
    \dot{\mathcal J}
    \leq
    -c_1\mathcal J^\rho
    -c_2\mathcal J^\theta
    +\Delta_c.
\label{eq:closed_loop_fixed_time_inequality}
\end{equation}
Consequently, for every \(\vartheta_c\in(0,1)\), \(\mathcal J(t)\) enters \(\Omega_c(\vartheta_c)=\{\mathcal J\geq0:\;c_1\mathcal J^\rho+c_2\mathcal J^\theta\leq\Delta_c/\vartheta_c\}\) within the fixed time \(T_c\leq[(1-\vartheta_c)c_1(1-\rho)]^{-1}+[(1-\vartheta_c)c_2(\theta-1)]^{-1}\). The bound is independent of \(\mathcal J(0)\). Consequently, the augmented signal \(\operatorname{col}\{\bx(t),\widetilde{\mathbf W}(t)\}\) is practically fixed-time uniformly ultimately bounded.
\end{theorem}

\begin{proof}
Along the physical learned closed-loop system~\eqref{eq:plant}, add and subtract the ideal minimizing policy \(\bu^*\) to obtain \(\dot V^*(\bx)=(\nabla V^*(\bx))^\top[f(\bx)+\mathbf g(\bx)(\bu^*+\ba)+\bd]+(\nabla V^*(\bx))^\top\mathbf g(\bx)(\hat{\bu}-\bu^*)\). By the saddle property in Proposition~\ref{prop:saddle},
\begin{equation}
    H(\bx,\bu^*,\ba,\bd,\nabla V^*)
    \leq
    H(\bx,\bu^*,\ba^*,\bd^*,\nabla V^*)
    =0.
\label{eq:saddle_Hamiltonian_bound}
\end{equation}
Expanding \eqref{eq:saddle_Hamiltonian_bound} gives \((\nabla V^*)^\top[f(\bx)+\mathbf g(\bx)(\bu^*+\ba)+\bd]\leq-Q(\bx)-U(\bu^*)+\gamma_a^2\ba^\top\mathbf T\ba+\gamma_d^2\bd^\top\mathbf S\bd\). Using Lemma~\ref{lem:adversarial_bound} and \(U(\bu^*)\geq0\), one obtains
\begin{equation}
    \dot V^*(\bx)
    \leq
    -Q(\bx)
    +(\nabla V^*(\bx))^\top
    \mathbf g(\bx)(\hat{\bu}-\bu^*)
    +\bar d.
\label{eq:value_derivative_before_policy_bound}
\end{equation}
From \eqref{eq:learned_control_mismatch_bound}, \(\left|(\nabla V^*(\bx))^\top\mathbf g(\bx)(\hat{\bu}-\bu^*)\right|\leq c_{Vg}\bar g\|\bx\|\left(m_W\|\widetilde{\mathbf W}\|+\varepsilon_u\right)\). Applying \eqref{eq:state_weight_cross_bound} and \eqref{eq:state_approximation_cross_bound} yields
\begin{equation}
    \left|
        (\nabla V^*(\bx))^\top
        \mathbf g(\bx)(\hat{\bu}-\bu^*)
    \right|
    \leq
    \epsilon_x\bx^\top\mathbf Q_x\bx
    +\ell_W\|\widetilde{\mathbf W}\|^2
    +\Delta_u.
\label{eq:policy_mismatch_final_bound}
\end{equation}
Substituting the state penalty \eqref{eq:stability_fixed_time_state_cost} and \eqref{eq:policy_mismatch_final_bound} into \eqref{eq:value_derivative_before_policy_bound} gives
\begin{equation}
\begin{aligned}
    \dot V^*(\bx)
    \leq{}&
    -(1-\epsilon_x)\bx^\top\mathbf Q_x\bx
    -\kappa_1\|\bx\|^{2\alpha}
    -\kappa_2\|\bx\|^{2\beta}\\
    &+\ell_W\|\widetilde{\mathbf W}\|^2
    +\Delta_x,
\end{aligned}
\label{eq:value_derivative_state_powers}
\end{equation}
where \(\Delta_x=\bar d+\Delta_u\). Using the value-function upper bound \eqref{eq:state_power_to_value_power}, \(\|\bx\|^{2\alpha}\geq\bar c_V^{-\alpha}(V^*(\bx))^\alpha\) and \(\|\bx\|^{2\beta}\geq\bar c_V^{-\beta}(V^*(\bx))^\beta\). Define \(a_x=\kappa_1\bar c_V^{-\alpha}\) and \(b_x=\kappa_2\bar c_V^{-\beta}\). Dropping the nonpositive quadratic term in \eqref{eq:value_derivative_state_powers} gives
\begin{equation}
    \dot V^*(\bx)
    \leq
    -a_x\big(V^*(\bx)\big)^\alpha
    -b_x\big(V^*(\bx)\big)^\beta
    +\ell_W\|\widetilde{\mathbf W}\|^2
    +\Delta_x.
\label{eq:value_derivative_value_powers}
\end{equation}

By the Leibniz rule, \(\dot V_I(t)=V^*(\bx(t))-V^*(\bx(t-\Delta T))=\int_{t-\Delta T}^{t}\dot V^*(\bx(\tau))d\tau\). Integrating \eqref{eq:value_derivative_value_powers} over \([t-\Delta T,t]\) gives
\begin{equation}
\begin{aligned}
    \dot V_I
    \leq{}&
    -a_x\int_{t-\Delta T}^{t}
    \big(V^*(\bx(\tau))\big)^\alpha d\tau
    -b_x\int_{t-\Delta T}^{t}
    \big(V^*(\bx(\tau))\big)^\beta d\tau\\
    &+\ell_W\int_{t-\Delta T}^{t}
    \|\widetilde{\mathbf W}(\tau)\|^2d\tau
    +\Delta T\Delta_x.
\end{aligned}
\label{eq:window_value_derivative_before_bounds}
\end{equation}
From the lower bound in \eqref{eq:critic_Lyapunov_eigenvalue_bounds}, \(\|\widetilde{\mathbf W}\|^2\leq2\bar\gamma V_W\). Define \(e_W=2\bar\gamma\ell_W\). Then
\begin{equation}
    \ell_W
    \int_{t-\Delta T}^{t}
    \|\widetilde{\mathbf W}(\tau)\|^2d\tau
    \leq
    e_W
    \int_{t-\Delta T}^{t}
    V_W(\tau)d\tau.
\label{eq:weight_window_conversion}
\end{equation}
Theorem~\ref{thm:critic_fixed_time} guarantees that \(V_W(t)\) is bounded. Since its differential equation is continuous on the resulting compact set, there exists \(M_W>0\) such that \(|\dot V_W(t)|\leq M_W\). For every \(\tau\in[t-\Delta T,t]\),
\begin{equation}
    V_W(\tau)
    \leq
    V_W(t)+M_W(t-\tau).
\label{eq:critic_window_pointwise_bound}
\end{equation}
Integrating \eqref{eq:critic_window_pointwise_bound} yields
\begin{equation}
    \int_{t-\Delta T}^{t}V_W(\tau)d\tau
    \leq
    \Delta T V_W(t)
    +\frac{1}{2}M_W(\Delta T)^2.
\label{eq:critic_window_integral_bound}
\end{equation}

Because the learned closed-loop trajectory remains in \(\Omega^r\), there exist constants \(L_V>0\) and \(V_{\max}>0\) such that \(|\dot V^*(\bx(t))|\leq L_V\) and \(0\leq V^*(\bx(t))\leq V_{\max}\). Applying Lemma~\ref{lem:fractional_window_bound} to \(\xi(t)=V^*(\bx(t))\) gives
\begin{equation}
    \int_{t-\Delta T}^{t}
    \big(V^*(\bx(\tau))\big)^\alpha d\tau
    \geq
    \omega_\alpha V_I^\rho(t),
\label{eq:state_low_power_window_bound}
\end{equation}
where \(\omega_\alpha=C_V^{\alpha-1}\) and \(C_V=\max\{2\sqrt{L_V},\,4\sqrt{V_{\max}/\Delta T}\}\). Applying Lemma~\ref{lem:convex_window_bound} gives
\begin{equation}
    \int_{t-\Delta T}^{t}
    \big(V^*(\bx(\tau))\big)^\beta d\tau
    \geq
    (\Delta T)^{1-\beta}V_I^\beta(t).
\label{eq:state_high_power_window_bound}
\end{equation}
Substituting \eqref{eq:weight_window_conversion}, \eqref{eq:critic_window_integral_bound}, \eqref{eq:state_low_power_window_bound}, and \eqref{eq:state_high_power_window_bound} into \eqref{eq:window_value_derivative_before_bounds} yields
\begin{equation}
    \dot V_I
    \leq
    -A_IV_I^\rho
    -B_IV_I^\theta
    +e_W\Delta T V_W
    +\Delta_I,
\label{eq:window_value_final_bound}
\end{equation}
where \(A_I=a_x\omega_\alpha\), \(B_I=b_x(\Delta T)^{1-\beta}\), and \(\Delta_I=\Delta T\Delta_x+\frac{1}{2}e_WM_W(\Delta T)^2\).

By Theorem~\ref{thm:critic_fixed_time} and the matched-power selection \eqref{eq:matched_critic_state_powers},
\begin{equation}
    \dot V_W
    \leq
    -h_qV_W^\rho
    -h_rV_W^\theta
    +\Delta_W.
\label{eq:matched_critic_derivative}
\end{equation}
Differentiating the mixed Lyapunov functional \eqref{eq:mixed_Lyapunov_functional} and using \eqref{eq:window_value_final_bound} and \eqref{eq:matched_critic_derivative} give \(\dot{\mathcal J}\leq-A_IV_I^\rho-B_IV_I^\theta-\lambda_Wh_qV_W^\rho-\lambda_Wh_rV_W^\theta+e_W\Delta T V_W+\Delta_I+\lambda_W\Delta_W\). By Lemma~\ref{lem:mixed_power_domination}, \(e_W\Delta T V_W\leq e_W\Delta T V_W^\rho+e_W\Delta T V_W^\theta\). Therefore,
\begin{equation}
    \dot{\mathcal J}
    \leq
    -A_IV_I^\rho
    -B_IV_I^\theta
    -C_WV_W^\rho
    -D_WV_W^\theta
    +\Delta_c,
\label{eq:mixed_Lyapunov_after_absorption}
\end{equation}
where \(C_W=\lambda_Wh_q-e_W\Delta T\), \(D_W=\lambda_Wh_r-e_W\Delta T\), and \(\Delta_c=\Delta_I+\lambda_W\Delta_W\). Condition \eqref{eq:lambda_W_condition} ensures \(C_W>0\) and \(D_W>0\).

Rewrite the critic powers as \(C_WV_W^\rho=C_W\lambda_W^{-\rho}(\lambda_WV_W)^\rho\) and \(D_WV_W^\theta=D_W\lambda_W^{-\theta}(\lambda_WV_W)^\theta\). Define \(\bar A=\min\{A_I,C_W\lambda_W^{-\rho}\}\) and \(\bar B=\min\{B_I,D_W\lambda_W^{-\theta}\}\). Then \eqref{eq:mixed_Lyapunov_after_absorption} implies \(\dot{\mathcal J}\leq-\bar A[V_I^\rho+(\lambda_WV_W)^\rho]-\bar B[V_I^\theta+(\lambda_WV_W)^\theta]+\Delta_c\). Using \eqref{eq:fractional_subadditivity}, \(V_I^\rho+(\lambda_WV_W)^\rho\geq(V_I+\lambda_WV_W)^\rho=\mathcal J^\rho\). Using \eqref{eq:convex_power_sum}, \(V_I^\theta+(\lambda_WV_W)^\theta\geq2^{1-\theta}(V_I+\lambda_WV_W)^\theta=2^{1-\theta}\mathcal J^\theta\). Consequently,
\begin{equation}
    \dot{\mathcal J}
    \leq
    -c_1\mathcal J^\rho
    -c_2\mathcal J^\theta
    +\Delta_c,
\label{eq:composite_fixed_time_final}
\end{equation}
where \(c_1=\bar A\) and \(c_2=2^{1-\theta}\bar B\). This proves \eqref{eq:closed_loop_fixed_time_inequality}. From \eqref{eq:composite_fixed_time_final}, by Theorem~\ref{thm:fixed_time_value_condition}, we deduce the result to be proved.
\end{proof}

\endgroup

\begin{remark}
\label{rem:inner_radius_tradeoff}
The inner radius \(r_->0\) represents a deliberate certification
tradeoff rather than a physical limitation of the closed-loop system.
Indeed, the fixed-time comparison analysis is carried out on the
nonterminal compact set
\(\Omega^r=\{x\in\Omega:\|x\|\ge r_-\}\); hence, the resulting
certificate guarantees convergence to the prescribed terminal
neighborhood \(B_{r_-}(0)\), rather than claiming exact convergence to
the origin. Decreasing \(r_-\) tightens this terminal neighborhood, but
simultaneously enlarges the range of the comparison quantities
\(V^*(x)/\|x\|^2\) and
\(\|\nabla_x V^*(x)\|/\|x\|\) that must be bounded over
\(\Omega^r\). Consequently, a smaller \(r_-\) generally leads to more
conservative comparison constants and may require larger control,
state-penalty, or learning gains to preserve the same fixed-time
certificate. Thus, approaching the origin more closely is obtained at
the cost of increased conservatism, control effort, and numerical
sensitivity. If the actual trajectory enters \(B_{r_-}(0)\), the
outer-region comparison is simply no longer invoked, since the
prescribed practical regulation objective has already been achieved.
If exact convergence to the origin is additionally required, a
two-region design may instead be adopted, in which a local stabilizing
controller or a conventional ADP/IRL law is activated inside
\(B_{r_-}(0)\), while the proposed fixed-time comparison mechanism is
retained on \(\Omega^r\).
\end{remark}

\section{Offline Data-Driven Construction of an Admissible Critic-Weight Initialization Region}
\label{khoitao}

The initialization \(\hat{\mathbf W}(0)\) is not merely a numerical choice, since it directly determines the initially implemented control policy. Although the saturated parametrization guarantees \(|\hat u_{\ell}|<\lambda\), \(\ell=1,\ldots,m\), for any finite \(\hat{\mathbf W}(0)\), this constraint alone does not ensure admissibility in the sense of Definition~1. To make this dependence quantitative, let \(\mathbf u_a\in\Psi(\Omega)\) denote an admissible reference policy associated with a critic weight vector \(\mathbf W_a\), and suppose that the implemented policy is locally Lipschitz with respect to the critic weights such that \(\|\hat{\mathbf u}(\mathbf x,\hat{\mathbf W}(0))-\mathbf u_a(\mathbf x)\|\leq L_u\|\hat{\mathbf W}(0)-\mathbf W_a\|\) for all \(\mathbf x\in\Omega\), where \(L_u>0\) denotes the corresponding local Lipschitz constant. If the nominal closed loop under \(\mathbf u_a\) admits a Lyapunov function satisfying \(\dot V_a\leq-\alpha_a\|\mathbf x\|^2\), while the perturbation induced by the initialization error contributes at most \(c_u\|\mathbf x\|\|\hat{\mathbf u}-\mathbf u_a\|\) to its derivative, then a sufficient initialization condition is \(\|\hat{\mathbf W}(0)-\mathbf W_a\|<\alpha_a/(c_uL_u)\), which preserves the stabilizing control direction over \(\Omega\). Moreover, a smaller initial critic-weight mismatch generally yields a smaller initial policy error and Bellman--Isaacs residual, thereby reducing transient control effort, critic correction magnitude, and sensitivity to approximation errors during the early learning stage. Consequently, the choice \(\hat{\mathbf W}(0)=\mathbf 0_{20}\)  should be interpreted as a simulation initialization rather than a generic admissibility guarantee; for practical implementation, \(\hat{\mathbf W}(0)\) should preferably be selected from the data-driven admissible initialization region \(\mathcal W_0^{\mathrm{adm}}\) constructed by follow.

The construction follows the sequence
\(\{\mathcal D_j^{\mathrm{off}}\}
\rightarrow
\{A_{K,j},B_{K,j}\}
\rightarrow
\{\mathbf u_{K,j}\}
\rightarrow
\{\hat{\mathbf W}_{j}^{K}\}
\rightarrow
\mathcal W_{0}^{\mathrm{adm}}\),
where \(\mathcal D_j^{\mathrm{off}}\) denotes the \(j\)th finite offline
state--input data batch, \(A_{K,j}\) and \(B_{K,j}\) are the corresponding
identified finite-dimensional Koopman lifted-system matrices,
\(\mathbf u_{K,j}\) is the stabilizing control policy constructed from the
identified lifted model, \(\hat{\mathbf W}_{j}^{K}\) is the critic-weight
initialization inferred from \(\mathbf u_{K,j}\) through the inverse
saturated-policy relation, and \(\mathcal W_{0}^{\mathrm{adm}}\) denotes
the resulting common admissible critic-weight initialization region.
The Koopman model is used only as an offline bridge for constructing a
stabilizing warm start; the subsequent online learning continues to use
the finite-window Bellman--Isaacs residual of Section~\ref{sec:fixed_time_irl}.

Consider \(M\) finite offline data batches \(\mathcal D_j^{\mathrm{off}}\), \(j=1,\ldots,M\), containing measured state--input trajectories over compact operating regions \(\Omega_j\), where \(j\) denotes the offline-data-batch index. Introduce a continuously differentiable lifting \(\boldsymbol{\eta}_j=\mathcal K_j(\mathbf x)\in\mathbb R^{n_j^{K}}\), where \(n_j^{K}\) denotes the dimension of the lifted state associated with the \(j\)th offline data batch, with \(\mathcal K_j(\mathbf 0)=\mathbf 0\), and suppose that \(c_{\eta,j}^{-}\|\mathbf x\|\leq\|\boldsymbol{\eta}_j(\mathbf x)\|\leq c_{\eta,j}^{+}\|\mathbf x\|\) on \(\Omega_j\), where \(c_{\eta,j}^{-}>0\) and \(c_{\eta,j}^{+}>0\) denote finite lifting bounds.

For each sampling interval \([t_{j,i}^{-},t_{j,i}^{+}]\), \(i=1,\ldots,N_j^{\mathrm{id}}\), define \(\Delta\boldsymbol{\eta}_{j,i}:=\boldsymbol{\eta}_j(\mathbf x(t_{j,i}^{+}))-\boldsymbol{\eta}_j(\mathbf x(t_{j,i}^{-}))\), \(\mathbf H_{j,i}:=\int_{t_{j,i}^{-}}^{t_{j,i}^{+}}\boldsymbol{\eta}_j(\mathbf x(\tau))\,d\tau\), and \(\mathbf U_{j,i}:=\int_{t_{j,i}^{-}}^{t_{j,i}^{+}}\mathbf u_{\mathrm{off}}(\tau)\,d\tau\), where \(i\) denotes the sampling-interval index and \(N_j^{\mathrm{id}}\) denotes the number of Koopman-identification intervals in the \(j\)th offline data batch. The finite-dimensional lifted relation is \(\Delta\boldsymbol{\eta}_{j,i}=A_{K,j}\mathbf H_{j,i}+B_{K,j}\mathbf U_{j,i}+\mathbf e_{K,j,i}\), where \(\mathbf e_{K,j,i}\) denotes the corresponding finite-interval lifting residual.

Let \(\mathbf Y_j:=[\Delta\boldsymbol{\eta}_{j,1},\ldots,\Delta\boldsymbol{\eta}_{j,N_j^{\mathrm{id}}}]\) and \(\mathbf Z_j:=\operatorname{col}\{[\mathbf H_{j,1},\ldots,\mathbf H_{j,N_j^{\mathrm{id}}}],[\mathbf U_{j,1},\ldots,\mathbf U_{j,N_j^{\mathrm{id}}}]\}\). If \(\operatorname{rank}(\mathbf Z_j)=n_j^{K}+m\), one may use \([A_{K,j}\;\;B_{K,j}]=\mathbf Y_j\mathbf Z_j^{\dagger}\), where \((\cdot)^{\dagger}\) denotes the Moore--Penrose pseudoinverse. The corresponding local lifted dynamics are written as \(\dot{\boldsymbol{\eta}}_j=A_{K,j}\boldsymbol{\eta}_j+B_{K,j}\mathbf u+\mathbf d_{K,j}\), where \(\mathbf d_{K,j}\) denotes the local Koopman-model mismatch and satisfies \(\|\mathbf d_{K,j}\|\leq\bar d_{K,j}\|\boldsymbol{\eta}_j\|\) on \(\Omega_j\).

Choose \(Q_{K,j}=Q_{K,j}^{\top}>0\) and \(R_{K,j}=R_{K,j}^{\top}>0\), and let \(P_j=P_j^{\top}>0\) solve \(A_{K,j}^{\top}P_j+P_jA_{K,j}-P_jB_{K,j}R_{K,j}^{-1}B_{K,j}^{\top}P_j+Q_{K,j}=0\). Define \(K_j:=R_{K,j}^{-1}B_{K,j}^{\top}P_j\), \(\mathbf u_{K,j}(\mathbf x):=-K_j\boldsymbol{\eta}_j(\mathbf x)\), and \(Q_{c,j}:=Q_{K,j}+K_j^{\top}R_{K,j}K_j\). Let \(c_{K,j}:=\lambda_{\min}(Q_{c,j})-2\|P_j\|\bar d_{K,j}\), and assume \(c_{K,j}>0\). With \(V_{K,j}:=\boldsymbol{\eta}_j^{\top}P_j\boldsymbol{\eta}_j\), choose \(\Omega_j^{0}:=\{\mathbf x\in\Omega_j:V_{K,j}(\mathbf x)\leq\rho_j\}\) such that \(\sqrt{\rho_j}\|K_jP_j^{-1/2}\|_{\infty,2}\leq(1-\delta_{u,j})\lambda\) for some \(\delta_{u,j}\in(0,1)\), where \(\|\cdot\|_{\infty,2}\) denotes the induced matrix norm from the Euclidean norm to the infinity norm. Then \(\|\mathbf u_{K,j}(\mathbf x)\|_{\infty}<\lambda\) on \(\Omega_j^{0}\), while \(\dot V_{K,j}\leq-c_{K,j}\|\boldsymbol{\eta}_j\|^2<0\). Hence \(\mathbf u_{K,j}\) is an admissible local stabilizing policy on \(\Omega_j^{0}\).

The next step maps this certified policy into the critic-weight space. Recall that the critic-induced saturated policy is \(\hat{\mathbf u}(\mathbf x;\hat{\mathbf W})=-\lambda\tanh\left(\frac{1}{2\lambda}R^{-1}\mathbf g^{\top}(\mathbf x)(\nabla\boldsymbol{\phi}(\mathbf x))^{\top}\hat{\mathbf W}\right)\). Define \(\mathcal G(\mathbf x):=\mathbf g^{\top}(\mathbf x)(\nabla\boldsymbol{\phi}(\mathbf x))^{\top}\). Since \(\|\mathbf u_{K,j}(\mathbf x)\|_{\infty}<\lambda\), the exact matching condition \(\hat{\mathbf u}(\mathbf x;\hat{\mathbf W})=\mathbf u_{K,j}(\mathbf x)\) is equivalent to \(\mathcal G(\mathbf x)\hat{\mathbf W}=\mathbf y_j(\mathbf x)\), where \(\mathbf y_j(\mathbf x):=-2\lambda R\operatorname{artanh}(\lambda^{-1}\mathbf u_{K,j}(\mathbf x))\) denotes the inverse-policy matching target and \(\operatorname{artanh}(\cdot)\) is applied componentwise.

Select samples \(\{\mathbf x_{j,i}\}_{i=1}^{N_j^{\mathrm{fit}}}\subset\Omega_j^{0}\) from the same offline data and define \(\mathbf G_j:=\operatorname{col}\{\mathcal G(\mathbf x_{j,1}),\ldots,\mathcal G(\mathbf x_{j,N_j^{\mathrm{fit}}})\}\) and \(\mathbf y_j^{\mathrm{fit}}:=\operatorname{col}\{\mathbf y_j(\mathbf x_{j,1}),\ldots,\mathbf y_j(\mathbf x_{j,N_j^{\mathrm{fit}}})\}\), where \(N_j^{\mathrm{fit}}\) denotes the number of offline samples used for inverse-policy fitting.

To guarantee \(\hat{\mathbf u}(\mathbf 0;\hat{\mathbf W})=\mathbf 0\), define \(\mathcal N_0:=\{\hat{\mathbf W}\in\mathbb R^{L}:\mathcal G(\mathbf 0)\hat{\mathbf W}=\mathbf 0\}\). The Koopman-induced initial critic center is selected as \(\hat{\mathbf W}_{j}^{K}:=\arg\min_{\hat{\mathbf W}\in\mathcal N_0}\{\|\mathbf G_j\hat{\mathbf W}-\mathbf y_j^{\mathrm{fit}}\|^2+\kappa_W\|\hat{\mathbf W}\|^2\}\), where \(\kappa_W\geq0\) denotes the regularization coefficient.

Define the inverse-policy mismatch \(\mathbf e_{\mathrm{inv},j}(\mathbf x):=\mathcal G(\mathbf x)\hat{\mathbf W}_{j}^{K}-\mathbf y_j(\mathbf x)\), and suppose \(\|\mathbf e_{\mathrm{inv},j}(\mathbf x)\|\leq\bar e_{\mathrm{inv},j}\|\boldsymbol{\eta}_j(\mathbf x)\|\). Assume also that \(\|\mathcal G(\mathbf x)-\mathcal G(\mathbf 0)\|\leq\ell_{\mathcal G,j}\|\mathbf x\|\) on \(\Omega_j^{0}\), where \(\ell_{\mathcal G,j}>0\) denotes the corresponding local Lipschitz constant. Since the componentwise hyperbolic tangent is globally one-Lipschitz, every \(\hat{\mathbf W}\in\mathcal N_0\) satisfies \(\|\hat{\mathbf u}(\mathbf x;\hat{\mathbf W})-\mathbf u_{K,j}(\mathbf x)\|\leq[\varepsilon_{u,j}+L_{uW,j}\|\hat{\mathbf W}-\hat{\mathbf W}_{j}^{K}\|]\|\boldsymbol{\eta}_j(\mathbf x)\|\), where \(\varepsilon_{u,j}:=\frac{1}{2}\|R^{-1}\|\bar e_{\mathrm{inv},j}\) and \(L_{uW,j}:=\frac{\|R^{-1}\|\ell_{\mathcal G,j}}{2c_{\eta,j}^{-}}\).

\begin{proposition}[Offline certified critic initialization]
\label{prop:offline_critic_initialization}
For any \(\sigma_j\in(0,1)\), suppose \(\frac{(1-\sigma_j)c_{K,j}}{2\|P_jB_{K,j}\|}>\varepsilon_{u,j}\), and define \(r_j^{W}:=\frac{\frac{(1-\sigma_j)c_{K,j}}{2\|P_jB_{K,j}\|}-\varepsilon_{u,j}}{L_{uW,j}}\). Then every initial critic estimate \(\hat{\mathbf W}^{(0)}\in\mathcal W_j^{\mathrm{adm}}\), where \(\mathcal W_j^{\mathrm{adm}}:=\{\hat{\mathbf W}^{(0)}\in\mathcal N_0:\|\hat{\mathbf W}^{(0)}-\hat{\mathbf W}_{j}^{K}\|\leq r_j^{W}\}\), induces an admissible saturated policy on \(\Omega_j^{0}\).
\end{proposition}

\begin{proof}
Along the true lifted dynamics under \(\hat{\mathbf u}(\mathbf x;\hat{\mathbf W}^{(0)})\), add and subtract \(\mathbf u_{K,j}\). Then \(\dot V_{K,j}\leq-[c_{K,j}-2\|P_jB_{K,j}\|(\varepsilon_{u,j}+L_{uW,j}\|\hat{\mathbf W}^{(0)}-\hat{\mathbf W}_{j}^{K}\|)]\|\boldsymbol{\eta}_j\|^2\). For every \(\hat{\mathbf W}^{(0)}\in\mathcal W_j^{\mathrm{adm}}\), the definition of \(r_j^{W}\) gives \(\dot V_{K,j}\leq-\sigma_jc_{K,j}\|\boldsymbol{\eta}_j\|^2<0\). Moreover, \(\hat{\mathbf W}^{(0)}\in\mathcal N_0\) implies \(\hat{\mathbf u}(\mathbf 0;\hat{\mathbf W}^{(0)})=\mathbf 0\), while the hyperbolic-tangent parameterization guarantees \(|\hat u_{\ell}(\mathbf x;\hat{\mathbf W}^{(0)})|<\lambda\) for every component \(\ell=1,\ldots,m\). Thus the induced initial policy is admissible on \(\Omega_j^{0}\).
\end{proof}

Repeating the construction over the \(M\) offline data batches yields \(\{\hat{\mathbf W}_{j}^{K},r_j^{W},\Omega_j^{0}\}_{j=1}^{M}\). If the active operating region is known, one may initialize with any \(\hat{\mathbf W}(0)\in\mathcal W_j^{\mathrm{adm}}\). For a single initialization required to be valid across all certified offline regimes, define \(\Omega^{0}:=\bigcap_{j=1}^{M}\Omega_j^{0}\) and \(\mathcal W_{0}^{\mathrm{adm}}:=\bigcap_{j=1}^{M}\mathcal W_j^{\mathrm{adm}}\).

Since each \(\mathcal W_j^{\mathrm{adm}}\) is the intersection of a Euclidean ball and the linear subspace \(\mathcal N_0\), \(\mathcal W_{0}^{\mathrm{adm}}\) is convex whenever it is nonempty. Therefore, if \(\hat{\mathbf W}_{h}^{(0)}\in\mathcal W_{0}^{\mathrm{adm}}\), \(h=1,\ldots,N_W\), then \(\operatorname{co}\{\hat{\mathbf W}_{1}^{(0)},\ldots,\hat{\mathbf W}_{N_W}^{(0)}\}\subseteq\mathcal W_{0}^{\mathrm{adm}}\), where \(\operatorname{co}\{\cdot\}\) denotes the convex hull. Hence every initialization \(\hat{\mathbf W}(0)=\sum_{h=1}^{N_W}\alpha_h\hat{\mathbf W}_{h}^{(0)}\), with \(\alpha_h\geq0\) and \(\sum_{h=1}^{N_W}\alpha_h=1\), generates an admissible initial policy on \(\Omega^{0}\).

It is important that convexification is performed only after the common admissible region has been certified. In general, two independently stabilizing critic initializations \(\hat{\mathbf W}_{j_1}^{K}\) and \(\hat{\mathbf W}_{j_2}^{K}\), \(j_1\neq j_2\), do not imply that \(\alpha\hat{\mathbf W}_{j_1}^{K}+(1-\alpha)\hat{\mathbf W}_{j_2}^{K}\) is stabilizing for every \(\alpha\in[0,1]\), because the mapping from critic weights to the implemented policy is nonlinear through \(\tanh(\cdot)\).

A robust single initialization may, for example, be selected as the Chebyshev center \(\hat{\mathbf W}_{0}^{\star}\) of \(\mathcal W_{0}^{\mathrm{adm}}\), obtained from \(\max_{\hat{\mathbf W},r}r\) subject to \(\mathcal G(\mathbf 0)\hat{\mathbf W}=\mathbf 0\), \(\|\hat{\mathbf W}-\hat{\mathbf W}_{j}^{K}\|+r\leq r_j^{W}\) for all \(j=1,\ldots,M\), and \(r\geq0\), where \(r\) denotes the radius of the largest Euclidean ball centered at \(\hat{\mathbf W}\) and contained in the common admissible region. Thus \(\hat{\mathbf W}(0)=\hat{\mathbf W}_{0}^{\star}\) provides a data-informed alternative to zero or random initialization.

\begin{remark}[Relation to the replay learning mechanism]
\label{rem:koopman_replay_relation}
The proposed construction changes only the initialization of the critic estimate. After online learning is activated, the history stack remains exactly that of Section~III-C, \(\mathcal H=\{(\Delta\boldsymbol{\phi}_k,\hat{\mathcal R}_k)\}_{k=1}^{N}\), where \(\Delta\boldsymbol{\phi}_k=\boldsymbol{\phi}(\mathbf x(t_k))-\boldsymbol{\phi}(\mathbf x(t_k-\Delta T))\) and \(\hat{\mathcal R}_k=\int_{t_k-\Delta T}^{t_k}\ell(\mathbf x(\tau),\hat{\mathbf u}(\tau),\hat{\mathbf a}(\tau),\hat{\mathbf d}(\tau))\,d\tau\).

Accordingly, the replayed integral Bellman--Isaacs residual remains \(\xi_k(t)=\hat{\mathbf W}^{\top}(t)\Delta\boldsymbol{\phi}_k+\hat{\mathcal R}_k\), while \(m_k=1+\Delta\boldsymbol{\phi}_k^{\top}\Delta\boldsymbol{\phi}_k\), \(\boldsymbol{\psi}_k=\Delta\boldsymbol{\phi}_k/m_k\), and \(s_k(t)=\xi_k(t)/m_k\) remain unchanged.

Likewise, the ideal critic weight \(\mathbf W\) and the critic-weight estimation error \(\widetilde{\mathbf W}(t)=\hat{\mathbf W}(t)-\mathbf W\) retain exactly the definitions used in the main analysis. Thus, the Bellman--Isaacs equation, critic update, and finite-data informativity condition are unaffected by the Koopman-based initialization.

The Koopman rank requirement \(\operatorname{rank}(\mathbf Z_j)=n_j^{K}+m\) and the replay informativity condition \(\lambda_{\min}(\Psi)\geq\Lambda\), where \(\Psi:=\sum_{k=1}^{N}\boldsymbol{\psi}_k\boldsymbol{\psi}_k^{\top}\), serve different purposes: the former identifies a finite-dimensional lifted control model, whereas the latter supplies informative directions for critic learning. Neither condition implies the other.

If an offline trajectory was collected under the same policy and signal convention used to define \(\hat{\mathcal R}_k\), the associated \(\Delta\boldsymbol{\phi}_k\) and \(\hat{\mathcal R}_k\) may also be reused in the history stack \(\mathcal H\). Otherwise, the offline data are used only for the Koopman warm start and the Bellman--Isaacs replay stack is populated according to Section~\ref{subsec:experience_replay}.
\end{remark}

\begin{remark}[Scope of the certificate]
\label{rem:koopman_initialization_scope}
The construction certifies the critic-induced policy at online initialization, namely \(\hat{\mathbf W}(0)\in\mathcal W_{0}^{\mathrm{adm}}\). It does not by itself guarantee that the unconstrained online critic update preserves \(\hat{\mathbf W}(t)\in\mathcal W_{0}^{\mathrm{adm}}\) for all future time. If such transient safety is required, the learning vector field may be projected onto the tangent cone of \(\mathcal W_{0}^{\mathrm{adm}}\). A complete projected-learning and fixed-time analysis is beyond the scope of this appendix and is left for separate investigation.
\end{remark}

\begin{remark}[Interpretation]
\label{rem:koopman_initialization_interpretation}
The essential role of the Koopman model is therefore not to replace the unknown nonlinear plant in the online HJI learning problem. Instead, it converts finite pre-deployment trajectory information into a certified stabilizing policy, which is subsequently mapped through the inverse saturated-policy relation into a non-arbitrary admissible region of the critic-estimate space. The online critic then continues to learn through the original finite-window Bellman--Isaacs residual without requiring the identified Koopman model to remain exact.
\end{remark}

% ============================================================
\section{Simulation Studies}
\label{sec:simulation}
% ============================================================

To facilitate implementation, the algorithm presented in the paper is outlined below.

\begin{algorithm}[H]
\footnotesize
\caption{Koopman-Warm-Started Fixed-Time Critic-Only IRL}
\label{alg:koopman_ft_irl_compact}

\begin{algorithmic}[1]

\Require $\bx(0)$, pre-deployment data $\mathcal D^{\mathrm{off}}$, $\Delta T$, $N$, $0<q<1<r$.

\State Lift the available offline data through $\boldsymbol{\eta}=\mathcal K(\bx)$ and identify a finite-dimensional Koopman model $(A_K,B_K)$.

\State Construct a bounded stabilizing Koopman policy $\mathbf u_K(\bx)$ on the certified local operating region.

\State Define $\mathcal G(\bx)=\mathbf g^{\top}(\bx)(\nabla\boldsymbol{\phi}(\bx))^{\top}$ and recover the critic warm start from
$\mathcal G(\bx_\ell)\hat{\mathbf W}^{(0)}
\approx
-2\lambda R\tanh^{-1}\!\big(\lambda^{-1}\mathbf u_K(\bx_\ell)\big)$
over the stored offline states.

\State Set $\hat{\mathbf W}(0)=\hat{\mathbf W}^{(0)}$ and initialize
$\hat V(\bx)=\hat{\mathbf W}^{\top}\boldsymbol{\phi}(\bx)$.

\State Reuse compatible offline samples, or collect finite transient data, to form
$\mathcal H=\{(\Delta\boldsymbol{\phi}_k,\hat{\mathcal R}_k)\}_{k=1}^{N}$
until
$\sum_{k=1}^{N}\boldsymbol{\psi}_k\boldsymbol{\psi}_k^{\top}
\succeq\underline{\lambda}\mathbf I$.

\While{$t<t_f$ \Comment{\(t_f\): final simulation time}}

\State Generate $\hat{\mathbf u}$, $\hat{\mathbf a}$, and $\hat{\mathbf d}$ from
$(\nabla\boldsymbol{\phi})^{\top}\hat{\mathbf W}$ using
\eqref{eq:ideal_disturbance_NN}--\eqref{eq:learned_attack_policy}.

\State Acquire the trajectory over $[t-\Delta T,t]$ and compute
$\Delta\boldsymbol{\phi}(t)$ and $\hat{\mathcal R}(t)$.

\State Evaluate the normalized current and replay Bellman--Isaacs residuals
$s(t)$ and $\{s_k(t)\}_{k=1}^{N}$.

\State Update $\hat{\mathbf W}$ using the two-power replay law
\eqref{eq:fixed_time_loss_gradient}.

\EndWhile

\Ensure Data-driven critic warm start, bounded policy $|\hat u_j|<\lambda$, finite-data learning without persistent excitation, and practical fixed-time closed-loop convergence.

\end{algorithmic}
\end{algorithm}

\FloatBarrier

% \begin{algorithm}[H]
% \footnotesize
% \caption{Fixed-Time Critic-Only IRL}
% \label{alg:ft_irl_compact}
% \begin{algorithmic}[1]
% \Require $\bm{x}(0)$, $\widehat{\bm W}(0)$, $\Delta T$, $N$, $0<q<1<r$.
% \State Initialize $\hat V(\bm{x})=\widehat{\bm W}^{\top}\boldsymbol{\phi}(\bm{x})$.
% \State Collect a finite history stack
% $\mathcal H=\{(\Delta\boldsymbol{\phi}_k,\hat{\mathcal R}_k)\}_{k=1}^{N}$
% until
% $\sum_{k=1}^{N}\boldsymbol{\psi}_k\boldsymbol{\psi}_k^{\top}
% \succeq\underline{\lambda}\mathbf I$.
% \While{$t<t_f$ \Comment{$t_f$: final simulation time}}
% \State Generate $\hat{\mathbf u}$, $\hat{\mathbf a}$, and $\hat{\mathbf d}$ from
% $(\nabla\boldsymbol{\phi})^{\top}\widehat{\bm W}$ using
% \eqref{eq:ideal_disturbance_NN}--\eqref{eq:learned_attack_policy}.
% \State Acquire the trajectory over $[t-\Delta T,t]$ and compute
% $\Delta\boldsymbol{\phi}(t)$ and $\hat{\mathcal R}(t)$.
% \State Evaluate the normalized current and replay
% Bellman residuals $s(t)$ and $\{s_k\}_{k=1}^{N}$.
% \State Update $\widehat{\bm W}$ using the two-power replay law
% \eqref{eq:fixed_time_loss_gradient}.
% \EndWhile
% \Ensure Bounded policy $|\hat u_j|<\lambda$ and practical
% fixed-time closed-loop convergence.
% \end{algorithmic}
% \end{algorithm}
% \FloatBarrier

\begingroup
\setlength{\abovedisplayskip}{3pt}
\setlength{\belowdisplayskip}{3pt}
\setlength{\abovedisplayshortskip}{2pt}
\setlength{\belowdisplayshortskip}{2pt}

% ------------------------------------------------------------
\subsection{Simulation Setup}
\label{subsec:simulation_setup}
% ------------------------------------------------------------

The present simulation considers the regulation problem of stabilizing the manipulator at the origin. The two-link model is adopted only as a simple representative nonlinear benchmark; the proposed framework is not specific to this plant and can be extended to higher-dimensional control-affine robotic systems.

A standard two-link planar manipulator is considered to evaluate
the proposed method. Define
\begin{equation}
    \bx
    =
    \col
    \left\{
        q_1,q_2,\dot q_1,\dot q_2
    \right\}
    =
    \col
    \left\{
        x_1,x_2,x_3,x_4
    \right\}.
\label{eq:sim_state}
\end{equation}
The manipulator dynamics are
\begin{equation}
\begin{aligned}
    \mathbf M(\bm q)\ddot{\bm q}
    +
    \mathbf C(\bm q,\dot{\bm q})\dot{\bm q}
    +
    \mathbf G(\bm q)
    +
    \mathbf D\dot{\bm q}
    =
    \bu+\ba+\bm\omega,
\end{aligned}
\label{eq:sim_manipulator}
\end{equation}
where
\begin{equation}
\begin{aligned}
    \mathbf M(\bm q)
    =
    \begin{bmatrix}
        p_1+2p_3\cos q_2
        &
        p_2+p_3\cos q_2\\
        p_2+p_3\cos q_2
        &
        p_2
    \end{bmatrix},
\end{aligned}
\label{eq:sim_mass_matrix}
\end{equation}
\begin{equation}
\begin{aligned}
    \mathbf C(\bm q,\dot{\bm q})
    =
    \begin{bmatrix}
        -p_3\dot q_2\sin q_2
        &
        -p_3(\dot q_1+\dot q_2)\sin q_2\\
        p_3\dot q_1\sin q_2
        &
        0
    \end{bmatrix},
\end{aligned}
\label{eq:sim_coriolis_matrix}
\end{equation}
and
\begin{equation}
    \mathbf G(\bm q)
    =
    \begin{bmatrix}
        g_1\sin q_1
        +
        g_2\sin(q_1+q_2)\\
        g_2\sin(q_1+q_2)
    \end{bmatrix},
    \qquad
    \mathbf D
    =
    \diag(d_1,d_2).
\label{eq:sim_gravity_damping}
\end{equation}
The selected gravity convention ensures that
\(\bx=\bm 0\) is an equilibrium of the unforced system.

The control-affine representation corresponding to
\eqref{eq:plant} is
\begin{equation}
    f(\bx)
    =
    \begin{bmatrix}
        x_3\\
        x_4\\
        -\mathbf M^{-1}(\bm q)
        \left[
            \mathbf C(\bm q,\dot{\bm q})\dot{\bm q}
            +
            \mathbf G(\bm q)
            +
            \mathbf D\dot{\bm q}
        \right]
    \end{bmatrix},
\label{eq:sim_drift}
\end{equation}
\begin{equation}
    \mathbf g(\bx)
    =
    \begin{bmatrix}
        \bm 0_{2\times2}\\
        \mathbf M^{-1}(\bm q)
    \end{bmatrix},
    \qquad
    \bd(\bx,t)
    =
    \begin{bmatrix}
        \bm 0_2\\
        \mathbf M^{-1}(\bm q)\bm\omega(t)
    \end{bmatrix}.
\label{eq:sim_input_disturbance_matrices}
\end{equation}
With the parameters in Table~\ref{tab:sim_parameters},
\(\mathbf M(\bm q)\) is uniformly positive definite for every
\(\bm q\).

The matched FDI attack is selected as
\begin{equation}
    \ba(t)
    =
    \begin{cases}
    \begin{bmatrix}
        0.65+0.30\sin(1.8t)\\
        -0.50+0.25\cos(1.4t)
    \end{bmatrix},
    &4\leq t\leq10,\\[4mm]
    \bm 0_2,
    &\text{otherwise},
    \end{cases}
\label{eq:sim_attack}
\end{equation}
while the unknown external torque disturbance is
\begin{equation}
    \bm\omega(t)
    =
    \begin{bmatrix}
        0.10\sin(2.6t)+0.04\cos(5t)\\
        0.08\cos(2.1t)+0.03\sin(4.3t)
    \end{bmatrix}.
\label{eq:sim_external_disturbance}
\end{equation}

% The critic is implemented as
% a single-layer polynomial neural network with \(L=20\) fixed
% activation functions:
% \begin{equation}
% \begin{aligned}
%     \bphi(\bx)
%     =
%     \col\big\{&
%         x_ix_j,\quad 1\leq i\leq j\leq4;\\
%         &x_i^4,\quad 1\leq i\leq4;\\
%         &x_i^2x_j^2,\quad 1\leq i<j\leq4
%     \big\}.
% \end{aligned}
% \label{eq:sim_polynomial_basis}
% \end{equation}

% The critic is implemented as a single-layer RBF neural network with
% \(L=20\) fixed Gaussian basis functions,
% \begin{equation}
% \bphi(\bx)
% =
% \col\left\{
% \phi_1(\bx),\ldots,\phi_L(\bx)
% \right\},
% \qquad
% \phi_\ell(\bx)
% =
% \exp\left(
% -\frac{\|\bx-\boldsymbol{c}_\ell\|^2}{2\sigma_\ell^2}
% \right)
% -
% \exp\left(
% -\frac{\|\boldsymbol{c}_\ell\|^2}{2\sigma_\ell^2}
% \right),
% \end{equation}
% where \(\boldsymbol{c}_\ell\in\mathbb{R}^4\) and \(\sigma_\ell>0\)
% are the fixed centers and widths, respectively. The centers are
% distributed over the explored state region, and the constant offset
% ensures \(\bphi(\boldsymbol{0})=\boldsymbol{0}\).

The critic is implemented as a single-layer RBF neural network with
\(L=20\) fixed Gaussian basis functions,
\begin{equation}
\begin{aligned}
\bphi(\bx)
&=\col\{\phi_1(\bx),\ldots,\phi_L(\bx)\},\\
\phi_\ell(\bx)
&=\exp\!\left(
-\frac{\|\bx-\boldsymbol{\mu}_\ell\|^2}{2\sigma_\ell^2}
\right)\\
&\quad-\exp\!\left(
-\frac{\|\boldsymbol{\mu}_\ell\|^2}{2\sigma_\ell^2}
\right),
\quad \ell=1,\ldots,L .
\end{aligned}
\label{eq:sim_rbf_basis}
\end{equation}
Here, \(\boldsymbol{c}_\ell\in\mathbb{R}^4\) and \(\sigma_\ell>0\)
are fixed centers and widths selected over the explored state region.
The offset ensures \(\bphi(\boldsymbol{0})=\boldsymbol{0}\).

% The basis consists of ten quadratic neurons, four pure quartic
% neurons, and six mixed quartic neurons. It is continuously
% differentiable and satisfies
% \begin{equation}
%     \bphi(\bm 0)=\bm 0.
% \label{eq:sim_basis_origin}
% \end{equation}
% This polynomial critic structure is commonly used in
% continuous-time adaptive dynamic programming because its
% gradient is available analytically and no center or width
% selection is required.

The cost, constrained policy, virtual maximizing policies, and
critic update are implemented exactly as defined in
Sections~\ref{sec:prelim_problem} and
\ref{sec:fixed_time_irl}; they are not repeated here. The
numerical values are summarized in
Table~\ref{tab:sim_parameters}.

% \begin{table}[t]
% \caption{Plant, control, and learning parameters}
% \label{tab:sim_parameters}
% \centering
% \scriptsize
% \renewcommand{\arraystretch}{0.96}
% \setlength{\tabcolsep}{2.5pt}
% \resizebox{\columnwidth}{!}{
% \begin{tabular}{@{}llll@{}}
% \toprule
% Parameter & Value & Parameter & Value\\
% \midrule
% \(p_1\) & \(2.70\)
% &
% \(p_2\) & \(0.80\)\\
% \(p_3\) & \(0.35\)
% &
% \((g_1,g_2)\) & \((8.5,2.6)\)\\
% \((d_1,d_2)\) & \((0.12,0.08)\)
% &
% \(h\) & \(10^{-3}\,\mathrm{s}\)\\
% \(t_f\) & \(16\,\mathrm{s}\)
% &
% \(\Delta T\) & \(0.04\,\mathrm{s}\)\\
% \(\bx(0)\)
% &
% \(\col\{0.70,-0.55,0.20,-0.15\}\)
% &
% \(\lambda\) & \(8\)\\
% \(\mathbf Q_x\)
% &
% \(\diag(7,6,1.5,1.2)\)
% &
% \(\mathbf R\)
% &
% \(0.06\mathbf I_2\)\\
% \((\kappa_1,\kappa_2)\)
% &
% \((0.60,0.08)\)
% &
% \((\alpha,\beta)\)
% &
% \((0.70,1.50)\)\\
% \(\mathbf T\)
% &
% \(0.8\mathbf I_2\)
% &
% \(\mathbf S\)
% &
% \(\mathbf I_4\)\\
% \((\gamma_a,\gamma_d)\)
% &
% \((2.0,2.5)\)
% &
% \((q,r)\)
% &
% \((0.70,2.0)\)\\
% \(L\)
% &
% \(20\)
% &
% \(\mathbf\Gamma\)
% &
% \(3\mathbf I_{20}\)\\
% \(\sigma\)
% &
% \(10^{-3}\)
% &
% \(\hat{\mathbf W}(0)\)
% &
% \(\bm 0_{20}\)\\
% \(N\)
% &
% \(60\)
% &
% \(\underline\lambda\)
% &
% \(5\times10^{-2}\)\\
% \bottomrule
% \end{tabular}}
% \end{table}

\begin{table}[H]
\caption{Plant, control, and learning parameters}
\label{tab:sim_parameters}
\centering
\scriptsize
\renewcommand{\arraystretch}{0.96}
\setlength{\tabcolsep}{2.5pt}
\resizebox{\columnwidth}{!}{
\begin{tabular}{@{}llll@{}}
\toprule
Parameter & Value & Parameter & Value\\
\midrule
\(p_1\) & \(2.70\)
&
\(p_2\) & \(0.80\)\\
\(p_3\) & \(0.35\)
&
\((g_1,g_2)\) & \((8.5,2.6)\)\\
\((d_1,d_2)\) & \((0.12,0.08)\)
&
\(h\) & \(10^{-3}\,\mathrm{s}\)\\
\(t_f\) & \(16\,\mathrm{s}\)
&
\(\Delta T\) & \(0.04\,\mathrm{s}\)\\
\(\bx(0)\)
&
\(\col\{0.70,-0.55,0.20,-0.15\}\)
&
\(\lambda\) & \(8\)\\
\(\mathbf Q_x\)
&
\(\diag(7,6,1.5,1.2)\)
&
\(\mathbf R\)
&
\(0.06\mathbf I_2\)\\
\((\kappa_1,\kappa_2)\)
&
\((0.60,0.08)\)
&
\((\alpha,\beta)\)
&
\((0.70,1.50)\)\\
\(\mathbf T\)
&
\(0.8\mathbf I_2\)
&
\(\mathbf S\)
&
\(\mathbf I_4\)\\
\((\gamma_a,\gamma_d)\)
&
\((2.0,2.5)\)
&
\((q,r)\)
&
\((0.70,2.0)\)\\
\(L\)
&
\(20\)
&
\(\mathbf\Gamma\)
&
\(3\mathbf I_{20}\)\\
\(\sigma\)
&
\(10^{-3}\)
&
\(\hat{\mathbf W}(0)\)
&
\(\bm 0_{20}\)\\
\(N\)
&
\(60\)
&
\(\underline\lambda\)
&
\(5\times10^{-2}\)\\
\bottomrule
\end{tabular}}
\end{table}

% The differential equations are integrated using the fourth-order
% Runge--Kutta method. During the first \(2.5\,\mathrm{s}\), a
% small bounded two-channel multisine signal with peak magnitude
% \(0.05\) is used to populate the replay stack. The excitation is
% introduced before the hyperbolic-tangent map already defined in
% \eqref{eq:learned_control_policy}; therefore, the input constraint
% remains satisfied. The excitation is removed after the replay
% matrix reaches the prescribed finite-data level.

% A new finite-window datum is retained only when it increases
% \(\lambda_{\min}(\mathbf\Psi)\). Once
% \begin{equation}
%     \lambda_{\min}(\mathbf\Psi)
%     \geq
%     \underline\lambda
% \label{eq:sim_replay_requirement}
% \end{equation}
% is achieved, the history stack is frozen. All compared methods
% use the same initial condition, FDI signal, disturbance,
% integration step, and actuator bound.

% ------------------------------------------------------------
% The following subsections will be added after the simulation
% results have been generated.
% ------------------------------------------------------------

% \subsection{State and Control Responses}
% \label{subsec:state_control_results}

% \subsection{Critic Learning and Replay Informativity}
% \label{subsec:critic_results}

% \subsection{Comparative and Ablation Results}
% \label{subsec:ablation_results}

\endgroup

\subsection{Simulation Results}
\label{subsec:simulation_results}

The proposed architecture is evaluated in terms of resilient regulation,
actuator admissibility, finite-data critic learning, practical fixed-time
behavior, and sensitivity to critic initialization. The FDI attack is active
over \(4\leq t\leq10~\mathrm{s}\), while the external disturbance persists
throughout the simulation. Unless otherwise stated, the critic is initialized
by the proposed offline Koopman-based warm start.

% ================================================================
% CLOSED-LOOP REGULATION AND INPUT CONSTRAINT
% ================================================================

\begin{figure}[H]
    \centering
    \makebox[\columnwidth][c]{%
        \subfloat[Joint-position responses.]{
            \includegraphics[width=0.47\columnwidth]
            {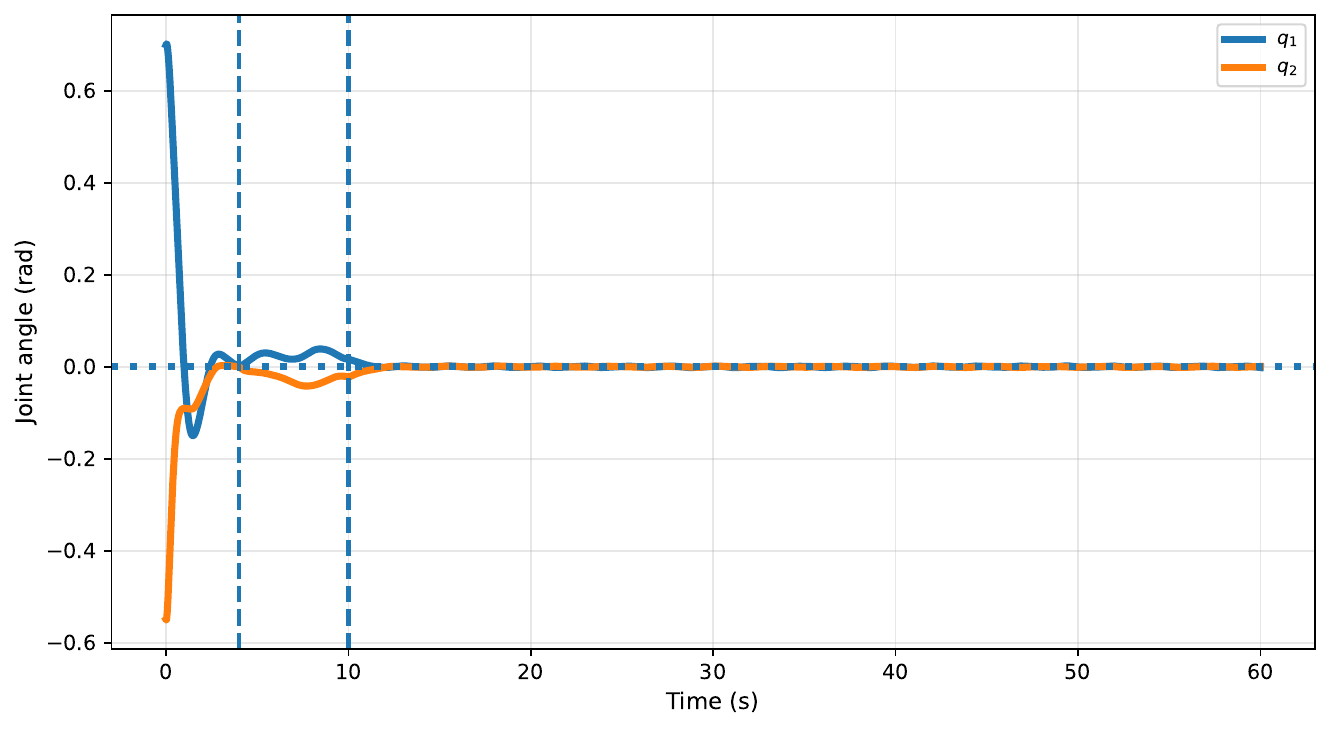}
            \label{fig:sim_joint_angles}
        }
        \hspace{0.015\columnwidth}
        \subfloat[Bounded control inputs.]{
            \includegraphics[width=0.47\columnwidth]
            {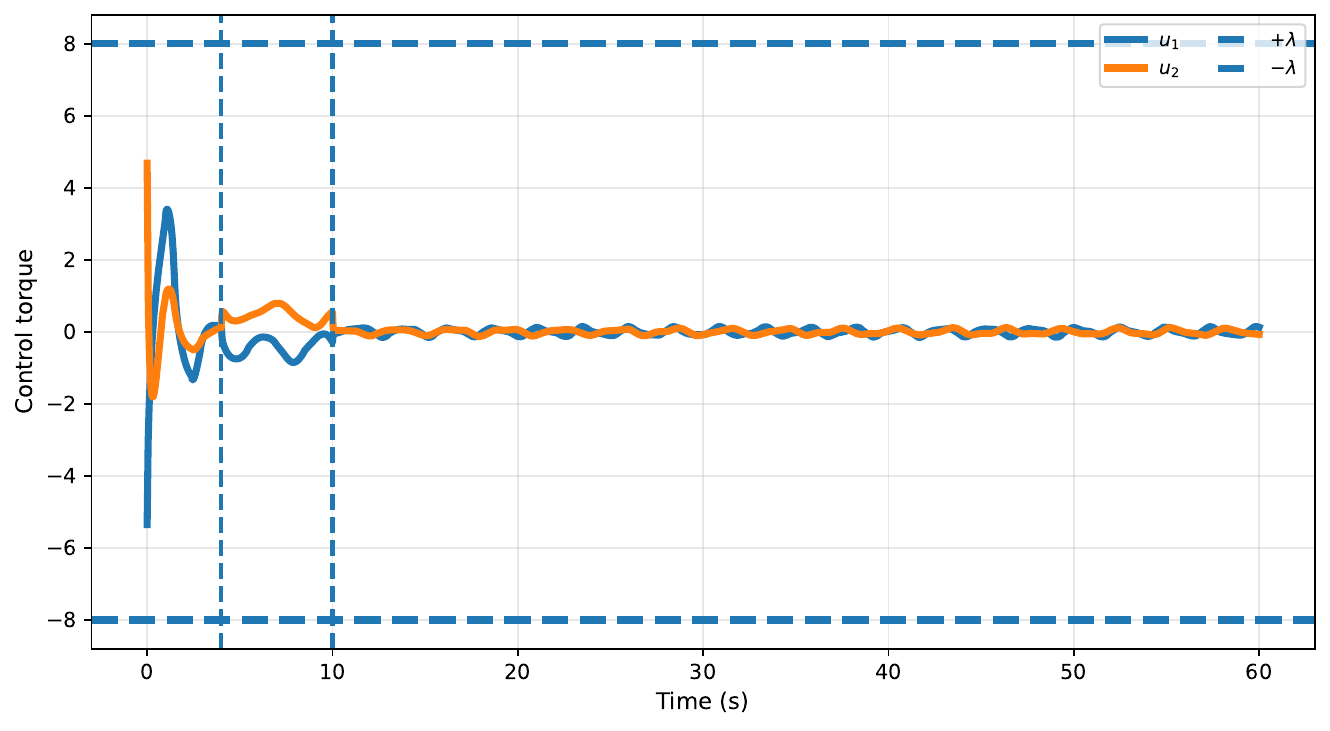}
            \label{fig:sim_control_inputs}
        }
    }
    \caption{Closed-loop regulation and actuator-constraint verification.}
    \label{fig:sim_closed_loop}
\end{figure}

Starting from
\(\bx(0)=\operatorname{col}\{0.70,-0.55,0.20,-0.15\}\),
Fig.~\ref{fig:sim_closed_loop}(a) shows rapid regulation of both joint
positions despite the subsequent attack and persistent disturbance.
Meanwhile, Fig.~\ref{fig:sim_closed_loop}(b) verifies
\(|\hat u_j(t)|<\lambda=8\) throughout the simulation, confirming that
the learned saturated policy preserves actuator admissibility by construction.

% ================================================================
% ATTACK AND DISTURBANCE
% ================================================================

\begin{figure}[H]
    \centering
    \makebox[\columnwidth][c]{%
        \subfloat[Matched FDI attack.]{
            \includegraphics[width=0.47\columnwidth]
            {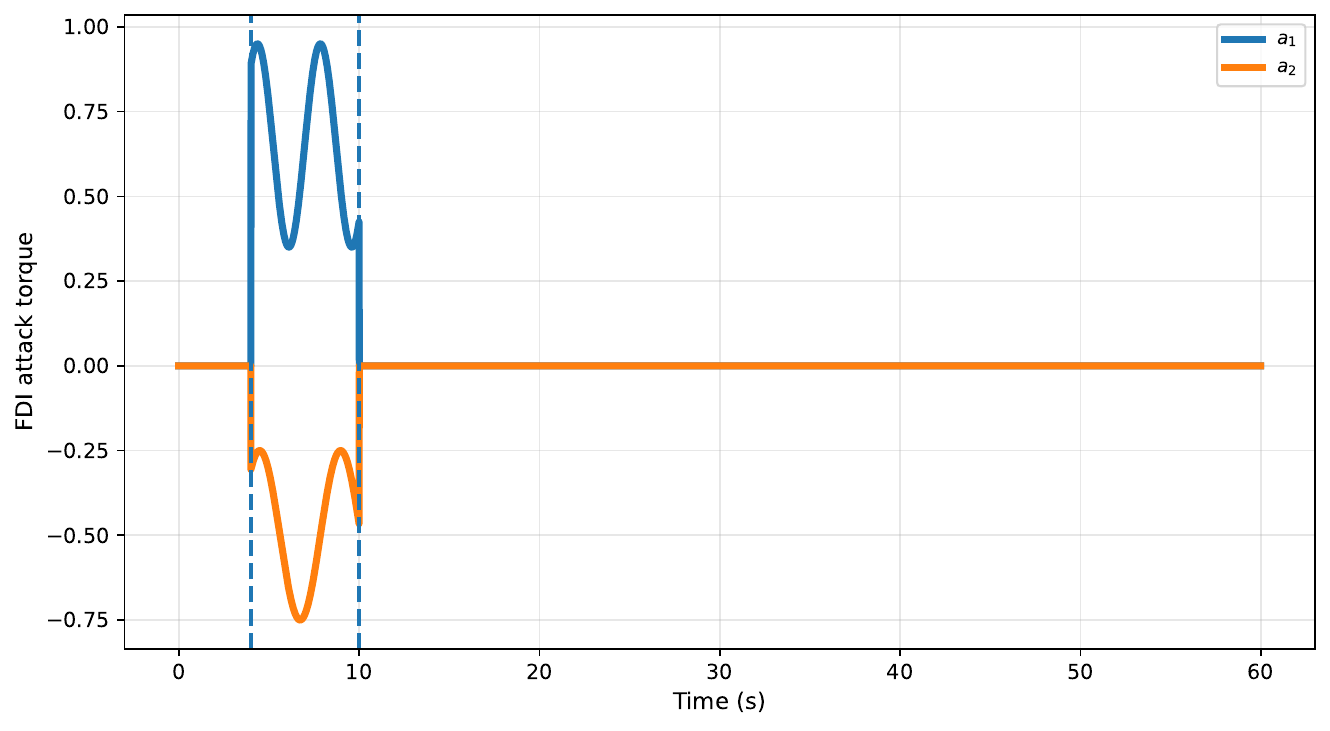}
            \label{fig:sim_fdi}
        }
        \hspace{0.015\columnwidth}
        \subfloat[External disturbance.]{
            \includegraphics[width=0.47\columnwidth]
            {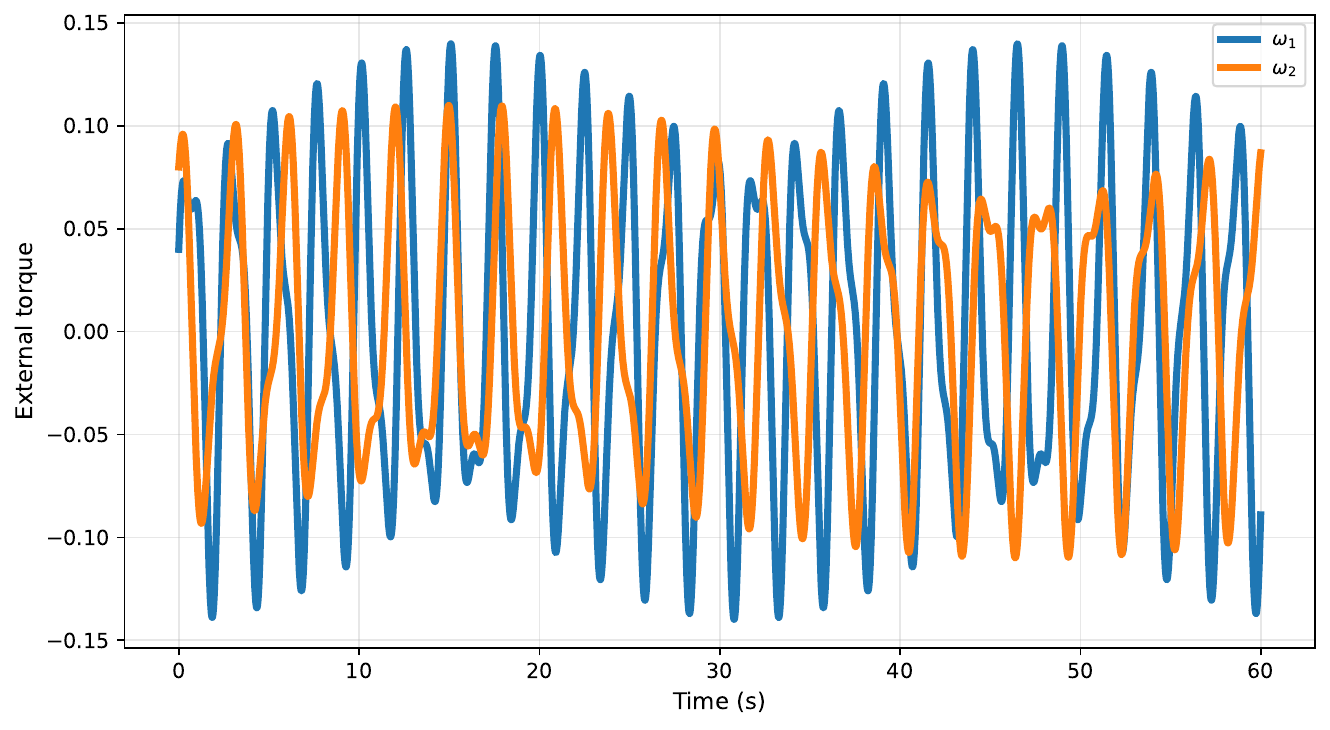}
            \label{fig:sim_disturbance}
        }
    }
    \caption{Adversarial and exogenous signals applied to the plant.}
    \label{fig:sim_adversarial_signals}
\end{figure}

As shown in Fig.~\ref{fig:sim_adversarial_signals}, the FDI channels reach
approximately \(0.95\) and \(0.75\) in magnitude, whereas the persistent
disturbance remains approximately within
\(|\omega_1|\leq0.14\) and \(|\omega_2|\leq0.11\).
Thus, the regulation in Fig.~\ref{fig:sim_closed_loop} is achieved under
simultaneous malicious actuation and nonvanishing exogenous disturbance
rather than under nominal operation.

% ================================================================
% CRITIC LEARNING AND BELLMAN RESIDUAL
% ================================================================

\begin{figure}[H]
    \centering
    \makebox[\columnwidth][c]{%
        \subfloat[Critic-weight estimates.]{
            \includegraphics[width=0.47\columnwidth]
            {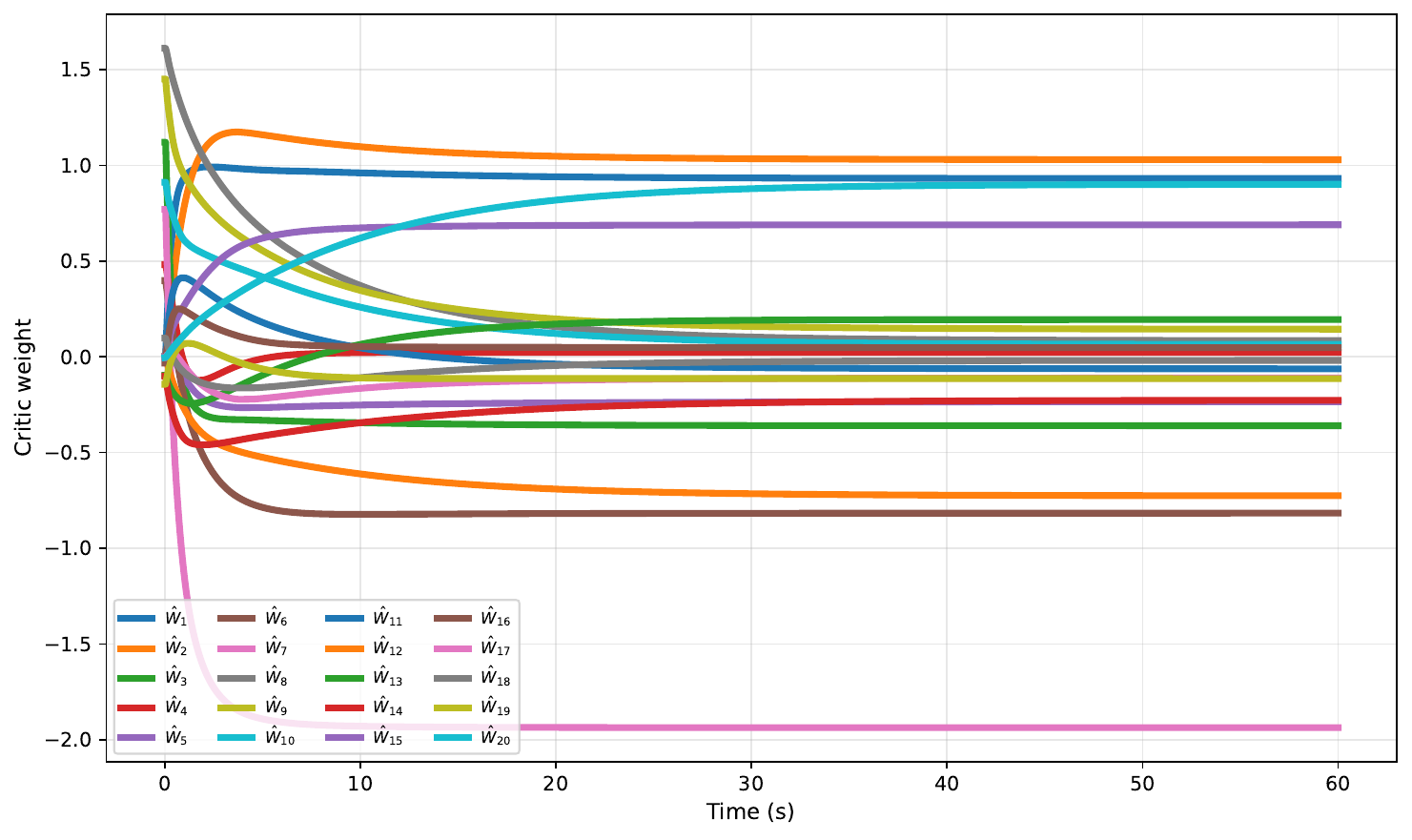}
            \label{fig:sim_critic_weights}
        }
        \hspace{0.015\columnwidth}
        \subfloat[Normalized Bellman--Isaacs residual.]{
            \includegraphics[width=0.47\columnwidth]
            {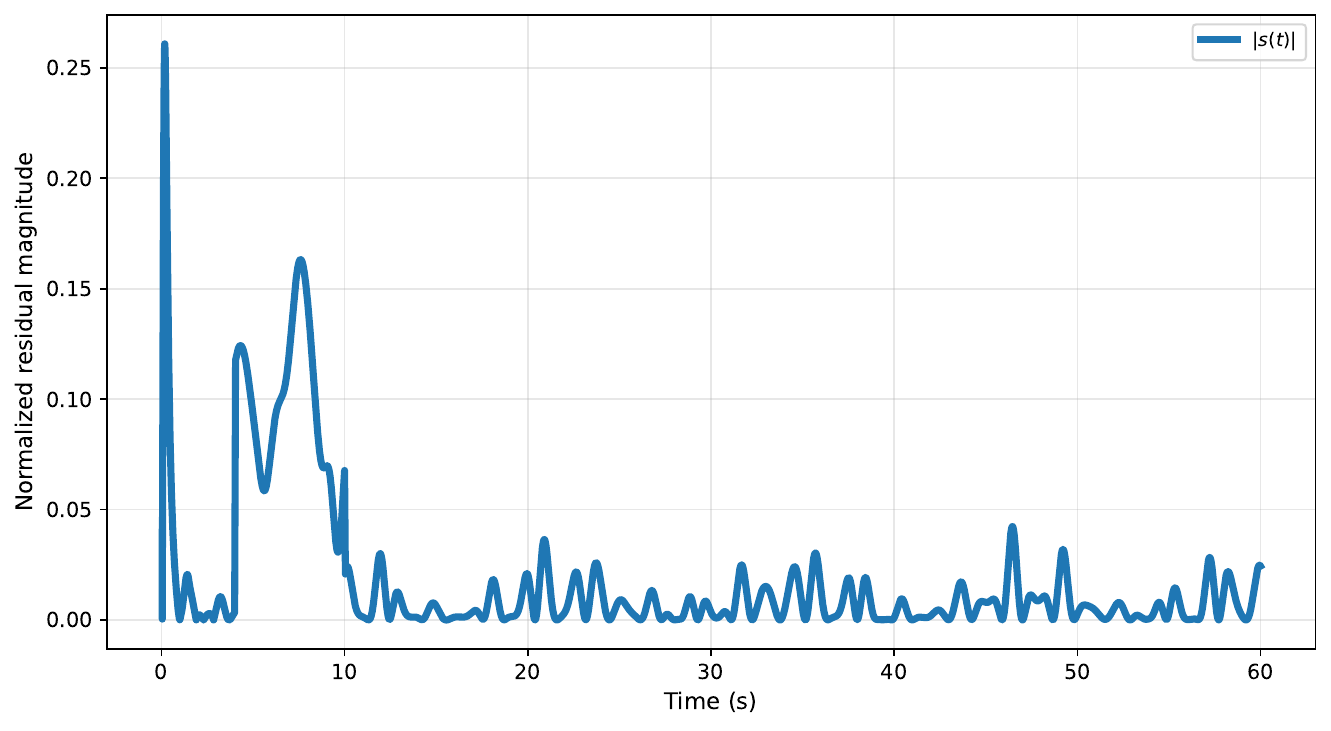}
            \label{fig:sim_bellman_residual}
        }
    }
    \caption{Critic-learning behavior under finite experience replay.}
    \label{fig:sim_learning}
\end{figure}

Figure~\ref{fig:sim_learning}(a) shows that all \(20\) critic weights remain
bounded and approach steady values. The Bellman--Isaacs residual in
Fig.~\ref{fig:sim_learning}(b) reaches about \(0.26\) during the initial
transient and subsequently remains in a small neighborhood of zero, including
during the FDI-active interval. These responses are consistent with the
practical critic-convergence result.

% ================================================================
% FINITE-DATA INFORMATIVITY AND FIXED-TIME VERIFICATION
% ================================================================

\begin{figure}[H]
    \centering
    \makebox[\columnwidth][c]{%
        \subfloat[Replay-Gramian eigenvalues.]{
            \includegraphics[width=0.47\columnwidth]
            {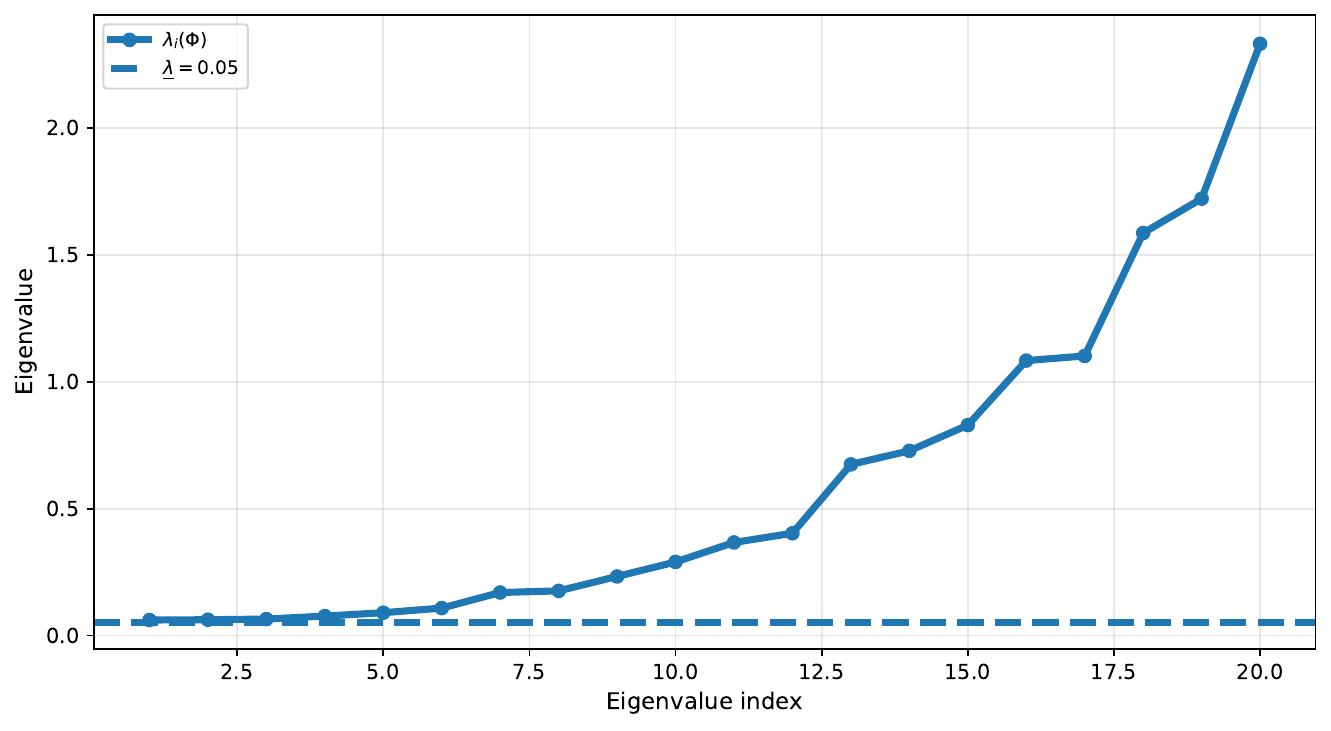}
            \label{fig:sim_replay_gramian}
        }
        \hspace{0.015\columnwidth}
        \subfloat[Multiple-initial-condition verification.]{
            \includegraphics[width=0.47\columnwidth]
            {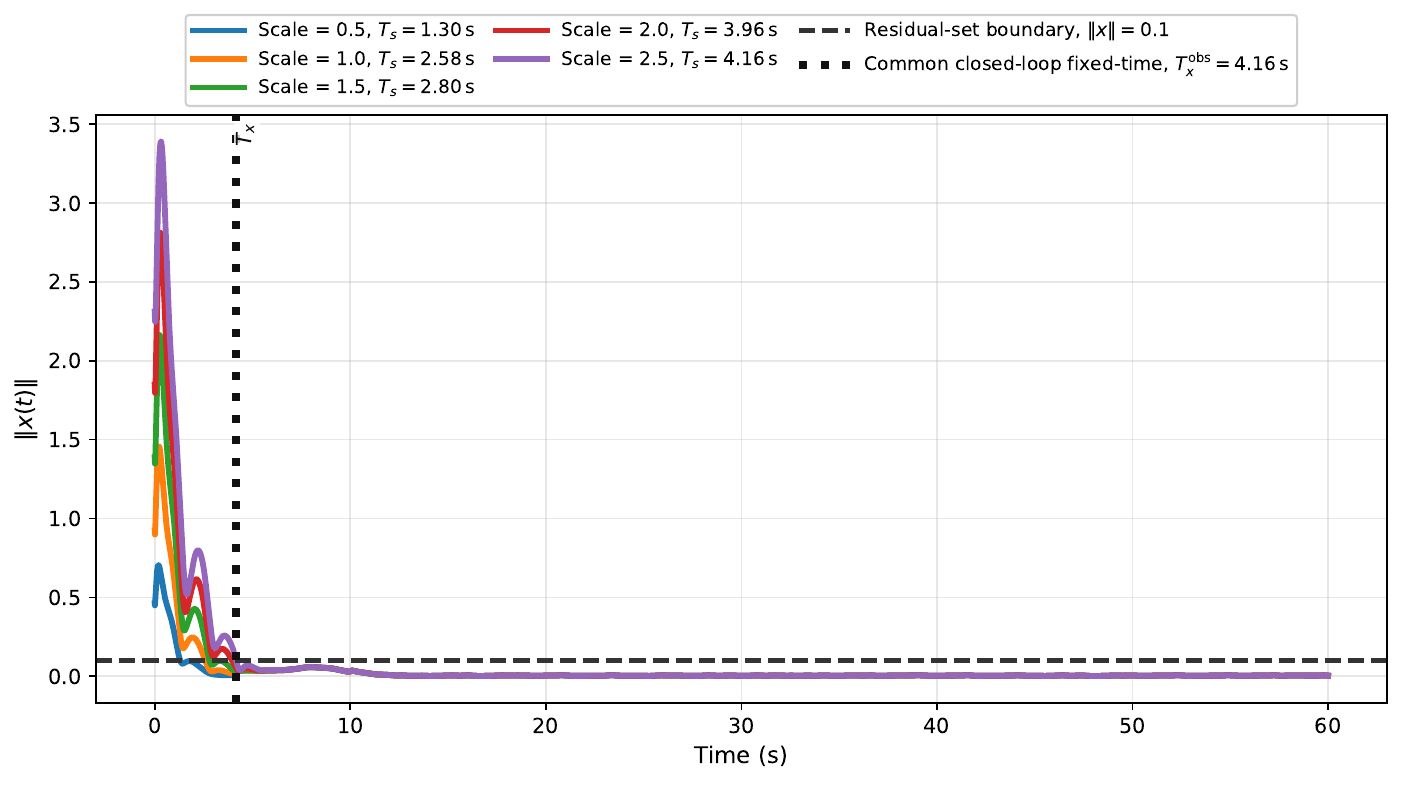}
            \label{fig:sim_multiple_initial_conditions}
        }
    }
    \caption{Finite-data informativity and practical fixed-time verification.}
    \label{fig:sim_theoretical_verification}
\end{figure}

The replay Gramian satisfies
\(\lambda_{\min}(\mathbf{\Phi})=0.0614>
\underline{\lambda}=0.05\), verifying finite-data informativity without
persistent excitation. For initial-state scales
\(0.5,\ 1.0,\ 1.5,\ 2.0,\ 2.5\), the corresponding residual-set entry
times are \(1.30,\ 2.58,\ 2.80,\ 3.96,\) and \(4.16~\mathrm{s}\).
Hence,
\[
T_x^{\mathrm{obs}}
=
\max_i T_{s,i}
=
4.16~\mathrm{s},
\]
numerically corroborating a common practical fixed-time bound over the
tested initial conditions.

% ================================================================
% CRITIC-INITIALIZATION SENSITIVITY
% ================================================================

\begin{figure}[H]
    \centering
    \makebox[\columnwidth][c]{%
        \subfloat[Closed-loop state norm.]{
            \includegraphics[width=0.47\columnwidth]
            {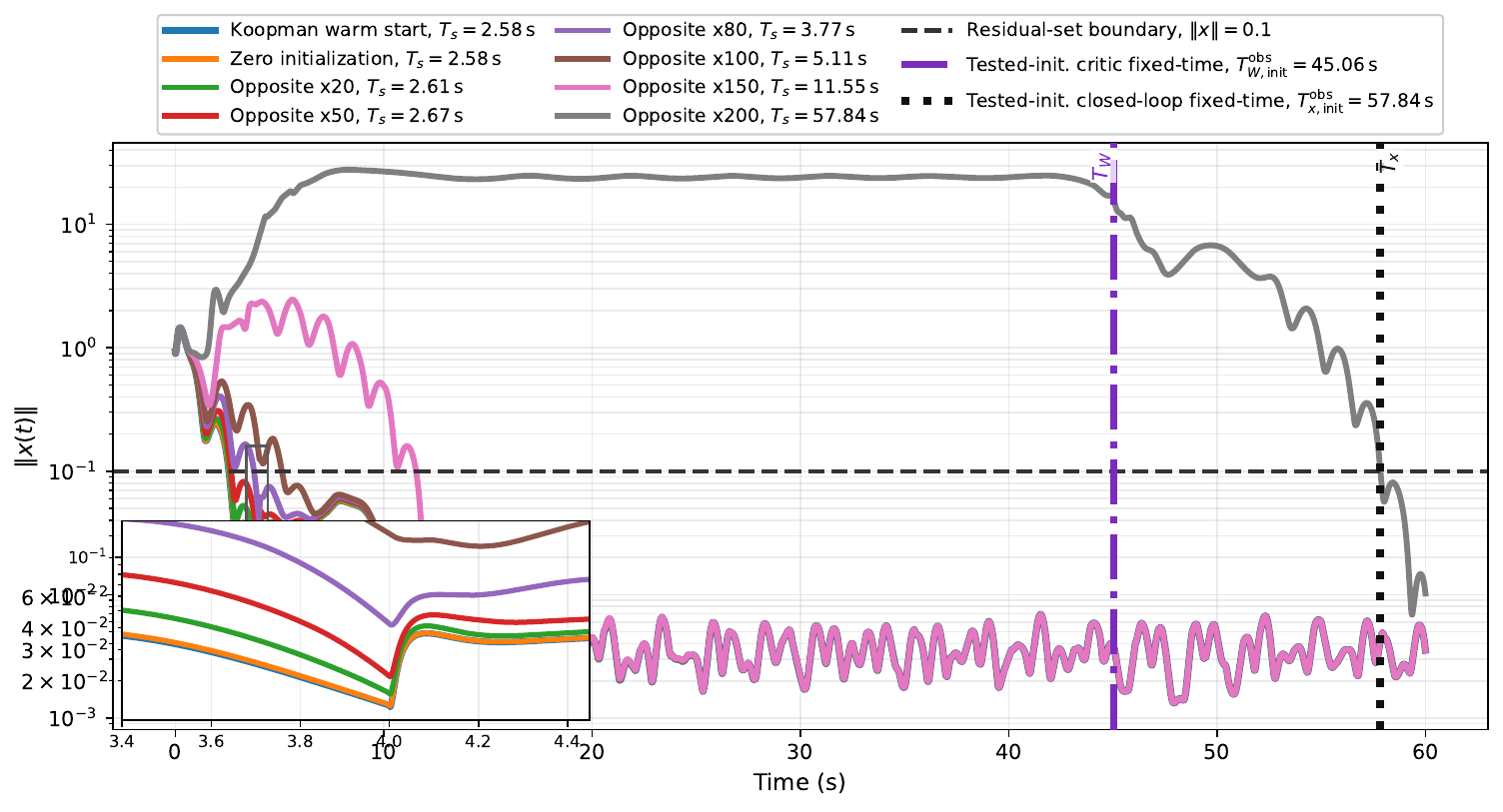}
            \label{fig:sim_adverse_state}
        }
        \hspace{0.015\columnwidth}
        \subfloat[Corresponding control effort.]{
            \includegraphics[width=0.47\columnwidth]
            {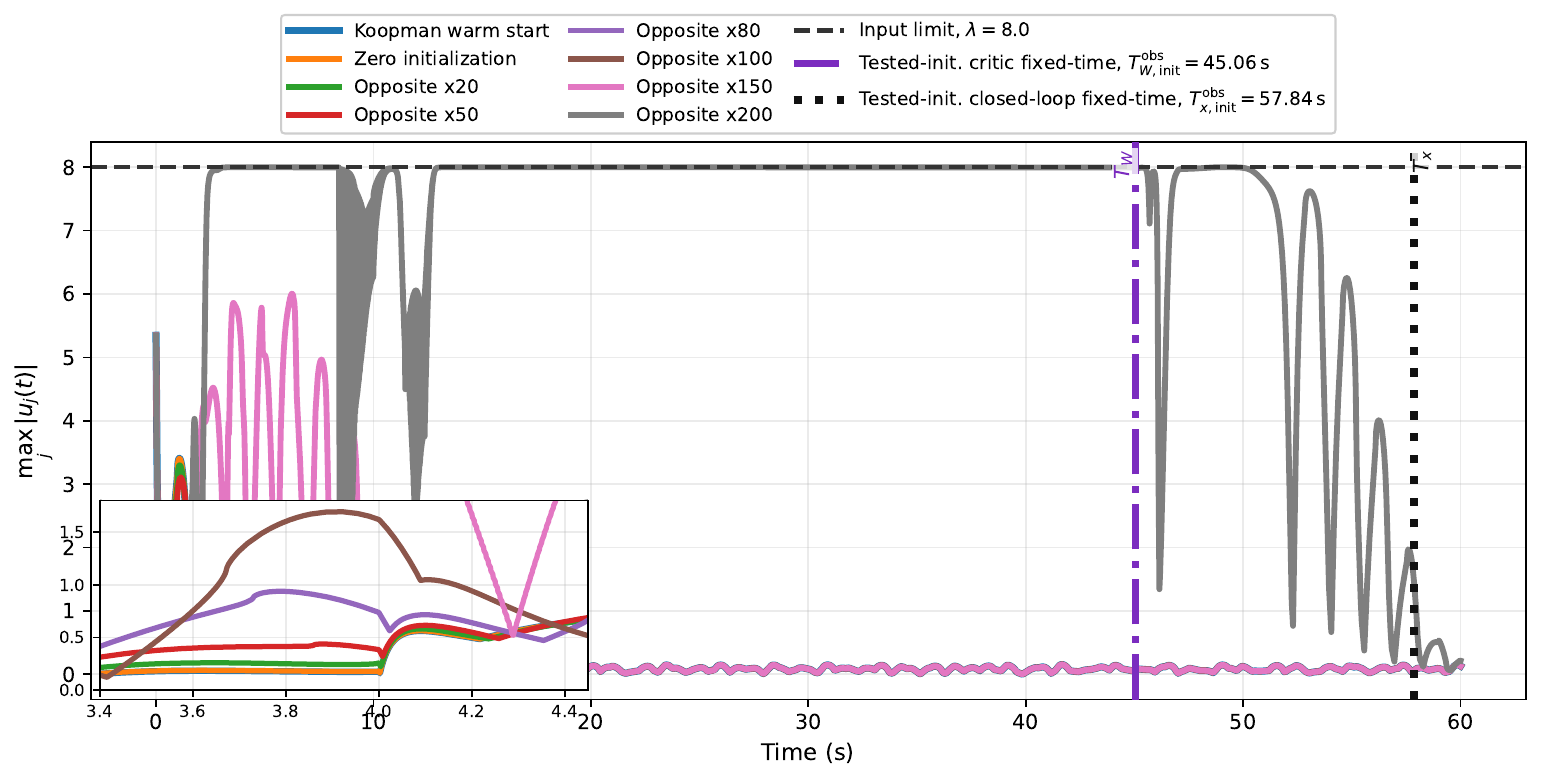}
            \label{fig:sim_adverse_control}
        }
    }
    \caption{Closed-loop sensitivity and recovery under adverse critic
    initialization.}
    \label{fig:sim_initialization}
\end{figure}

To quantify initialization sensitivity, the Koopman warm start is compared
with zero initialization and
\(\hat{\mathbf W}_{\mathrm{adv}}(0)
=-\sigma\hat{\mathbf W}_{K}(0)\),
where \(\sigma\in\{20,50,80,100,150,200\}\).
The state-set entry time increases from \(2.58~\mathrm{s}\) for the Koopman
warm start to \(57.84~\mathrm{s}\) for the extreme \(\sigma=200\)
initialization, which also drives the control close to \(\lambda=8\) for
extended intervals. Nevertheless, all tested cases ultimately recover, with
the observed bounds
\(T_{W,\mathrm{init}}^{\mathrm{obs}}=45.06~\mathrm{s}\) and
\(T_{x,\mathrm{init}}^{\mathrm{obs}}=57.84~\mathrm{s}\).
Thus, critic initialization strongly shapes the physical learning transient,
while the data-driven warm start avoids the large excursions induced by
poorly oriented initial weights. These observations also motivate the
development of certified admissible regions in critic-weight space and
initialization-dependent recovery bounds as natural extensions of the
present framework.

\section{Conclusion}

This paper developed a critic-only fixed-time integral reinforcement
learning framework for input-constrained unknown nonlinear systems
subject to matched FDI attacks and external disturbances. A saturated
HJI formulation, finite-data experience replay, and a two-power critic
update were integrated to enable model-free secure policy learning
without requiring persistent excitation along the entire trajectory.
The analysis established practical fixed-time boundedness of both the
critic weight error and the resulting nonlinear closed-loop system,
while preserving the prescribed actuator limits. Simulation results
on a nonlinear two-link manipulator corroborated the theoretical
properties under persistent disturbances and intermittent FDI attacks.
Future work will focus on relaxing the value-function regularity
requirements and developing data-informed critic initialization.

\section*{Acknowledgment}
The authors would like to thank Ho Chi Minh City University of Technology (HCMUT) and Vietnam National University Ho Chi Minh City (VNU-HCM) for supporting this research.

\bibliographystyle{IEEEtran}
\bibliography{references}

@article{AbuKhalafLewis2005,
  author  = {Abu-Khalaf, Murad and Lewis, Frank L.},
  title   = {Nearly optimal control laws for nonlinear systems with saturating actuators using a neural network {HJB} approach},
  journal = {Automatica},
  volume  = {41},
  number  = {5},
  pages   = {779--791},
  year    = {2005}
}

@article{VrabieLewis2009,
  author  = {Vrabie, Draguna and Pastravanu, Octavian and Abu-Khalaf, Murad and Lewis, Frank L.},
  title   = {Adaptive optimal control for continuous-time linear systems based on reinforcement learning},
  journal = {Automatica},
  volume  = {45},
  number  = {2},
  pages   = {477--484},
  year    = {2009}
}

@article{Modares2014,
  author  = {Modares, Hamidreza and Lewis, Frank L. and Naghibi-Sistani, Mohammad-Bagher},
  title   = {Integral reinforcement learning and experience replay for adaptive optimal control of partially-unknown constrained-input continuous-time systems},
  journal = {Automatica},
  volume  = {50},
  number  = {1},
  pages   = {193--202},
  year    = {2014}
}

@book{BasarBernhard1995,
  author    = {Ba{\c{s}}ar, Tamer and Bernhard, Pierre},
  title     = {$H_\infty$-Optimal Control and Related Minimax Design Problems},
  publisher = {Birkh{\"a}user},
  year      = {1995}
}

@article{Polyakov2012,
  author  = {Polyakov, Andrey},
  title   = {Nonlinear feedback design for fixed-time stabilization of linear control systems},
  journal = {IEEE Transactions on Automatic Control},
  volume  = {57},
  number  = {8},
  pages   = {2106--2110},
  year    = {2012}
}

@article{Gong2025,
  author  = {Gong, Zhenyu and Yang, Feisheng and Yuan, Yuan and Ma, Qian and Zheng, Wei Xing},
  title   = {Secure Formation Control of Multiagent System Against {FDI} Attack Using Fixed-Time Convergent Reinforcement Learning},
  journal = {IEEE Transactions on Control of Network Systems},
  volume  = {12},
  number  = {2},
  pages   = {1203--1213},
  year    = {2025}
}

@article{LewisVrabie2009,
  author  = {Lewis, Frank L. and Vrabie, Draguna},
  title   = {Reinforcement Learning and Adaptive Dynamic Programming for Feedback Control},
  journal = {IEEE Circuits and Systems Magazine},
  volume  = {9},
  number  = {3},
  pages   = {32--50},
  year    = {2009}
}

@article{VamvoudakisLewis2010,
  author  = {Vamvoudakis, Kyriakos G. and Lewis, Frank L.},
  title   = {Online Actor--Critic Algorithm to Solve the Continuous-Time Infinite Horizon Optimal Control Problem},
  journal = {Automatica},
  volume  = {46},
  number  = {5},
  pages   = {878--888},
  year    = {2010}
}

@article{Bhasin2013,
  author  = {Bhasin, Shubhendu and Sharma, Nitin and Yang, H. and Dixon, Warren E.},
  title   = {A Novel Actor--Critic--Identifier Architecture for Approximate Optimal Control of Uncertain Nonlinear Systems},
  journal = {Automatica},
  volume  = {49},
  number  = {1},
  pages   = {82--92},
  year    = {2013}
}

@article{AbuKhalafLewisHuang2006,
  author  = {Abu-Khalaf, Murad and Lewis, Frank L. and Huang, Jie},
  title   = {Policy Iterations on the Hamilton--Jacobi--Isaacs Equation for {$H_\infty$} State Feedback Control With Input Saturation},
  journal = {IEEE Transactions on Automatic Control},
  volume  = {51},
  number  = {12},
  pages   = {1989--1995},
  year    = {2006}
}

@article{ModaresLewisJiang2015,
  author  = {Modares, Hamidreza and Lewis, Frank L. and Jiang, Zhong-Ping},
  title   = {{$H_\infty$} Tracking Control of Completely Unknown Continuous-Time Systems via Off-Policy Reinforcement Learning},
  journal = {IEEE Transactions on Neural Networks and Learning Systems},
  volume  = {26},
  number  = {10},
  pages   = {2550--2562},
  year    = {2015}
}

@article{LuoWuHuang2015,
  author  = {Luo, Biao and Wu, Huai-Ning and Huang, Tingwen},
  title   = {Off-Policy Reinforcement Learning for {$H_\infty$} Control Design},
  journal = {IEEE Transactions on Cybernetics},
  volume  = {45},
  number  = {1},
  pages   = {65--76},
  year    = {2015}
}

@article{BianJiang2014,
  author  = {Bian, Tao and Jiang, Zhong-Ping},
  title   = {Adaptive Dynamic Programming and Optimal Control of Nonlinear Nonaffine Systems},
  journal = {Automatica},
  volume  = {50},
  number  = {10},
  pages   = {2624--2632},
  year    = {2014}
}

@article{Lin1992,
  author  = {Lin, Long-Ji},
  title   = {Self-Improving Reactive Agents Based on Reinforcement Learning, Planning and Teaching},
  journal = {Machine Learning},
  volume  = {8},
  pages   = {293--321},
  year    = {1992}
}

@article{Mnih2015,
  author  = {Mnih, Volodymyr and Kavukcuoglu, Koray and Silver, David and Rusu, Andrei A. and Veness, Joel and Bellemare, Marc G. and Graves, Alex and Riedmiller, Martin and Fidjeland, Andreas K. and Ostrovski, Georg and Petersen, Stig and Beattie, Charles and Sadik, Amir and Antonoglou, Ioannis and King, Helen and Kumaran, Dharshan and Wierstra, Daan and Legg, Shane and Hassabis, Demis},
  title   = {Human-Level Control Through Deep Reinforcement Learning},
  journal = {Nature},
  volume  = {518},
  number  = {7540},
  pages   = {529--533},
  year    = {2015}
}

@article{BhatBernstein2000,
  author  = {Bhat, Sanjay P. and Bernstein, Dennis S.},
  title   = {Finite-Time Stability of Continuous Autonomous Systems},
  journal = {SIAM Journal on Control and Optimization},
  volume  = {38},
  number  = {3},
  pages   = {751--766},
  year    = {2000}
}

@article{MoulayPerruquetti2006,
  author  = {Moulay, Emmanuel and Perruquetti, Wilfrid},
  title   = {Finite Time Stability and Stabilization of a Class of Continuous Systems},
  journal = {Journal of Mathematical Analysis and Applications},
  volume  = {323},
  number  = {2},
  pages   = {1430--1443},
  year    = {2006}
}

@article{Pasqualetti2013,
  author  = {Pasqualetti, Fabio and D{\"o}rfler, Florian and Bullo, Francesco},
  title   = {Attack Detection and Identification in Cyber-Physical Systems},
  journal = {IEEE Transactions on Automatic Control},
  volume  = {58},
  number  = {11},
  pages   = {2715--2729},
  year    = {2013}
}

@article{Fawzi2014,
  author  = {Fawzi, Hamza and Tabuada, Paulo and Diggavi, Suhas},
  title   = {Secure Estimation and Control for Cyber-Physical Systems Under Adversarial Attacks},
  journal = {IEEE Transactions on Automatic Control},
  volume  = {59},
  number  = {6},
  pages   = {1454--1467},
  year    = {2014}
}

@article{Teixeira2015,
  author  = {Teixeira, Andr{\'e} M. H. and Shames, Iman and Sandberg, Henrik and Johansson, Karl H.},
  title   = {A Secure Control Framework for Resource-Limited Adversaries},
  journal = {Automatica},
  volume  = {51},
  pages   = {135--148},
  year    = {2015}
}

@article{Dibaji2019,
  author  = {Dibaji, Seyed Mehran and Pirani, Mohammad and Flamholz, Daniel B. and Annaswamy, Anuradha M. and Johansson, Karl H. and Chakrabortty, Aranya},
  title   = {A Systems and Control Perspective of CPS Security},
  journal = {Annual Reviews in Control},
  volume  = {47},
  pages   = {394--411},
  year    = {2019}
}

@article{SanchezTorres2018PredefinedTime,
  author  = {S{\'a}nchez-Torres, Juan Diego and G{\'o}mez-Guti{\'e}rrez, David and L{\'o}pez, Esteban and Loukianov, Alexander G.},
  title   = {A Class of Predefined-Time Stable Dynamical Systems},
  journal = {IMA Journal of Mathematical Control and Information},
  volume  = {35},
  number  = {suppl\_1},
  pages   = {i1--i29},
  year    = {2018},
  doi     = {10.1093/imamci/dnx004}
}

@article{Cao2022PrescribedTimeTracking,
  author  = {Cao, Ye and Cao, Jianfu and Song, Yongduan},
  title   = {Practical Prescribed Time Tracking Control over Infinite Time Interval Involving Mismatched Uncertainties and Non-Vanishing Disturbances},
  journal = {Automatica},
  volume  = {136},
  pages   = {110050},
  year    = {2022},
  doi     = {10.1016/j.automatica.2021.110050}
}

@article{WangLai2020FixedTimeControl,
  author  = {Wang, Fang and Lai, Guanyu},
  title   = {Fixed-Time Control Design for Nonlinear Uncertain Systems via Adaptive Method},
  journal = {Systems \& Control Letters},
  volume  = {140},
  pages   = {104704},
  year    = {2020},
  doi     = {10.1016/j.sysconle.2020.104704}
}

@article{Heydari2018StabilizingInitialPolicy,
  author  = {Ali Heydari},
  title   = {Stability Analysis of Optimal Adaptive Control Under Value Iteration Using a Stabilizing Initial Policy},
  journal = {IEEE Transactions on Neural Networks and Learning Systems},
  volume  = {29},
  number  = {9},
  pages   = {4522--4527},
  year    = {2018},
  doi     = {10.1109/TNNLS.2017.2755501}
}

@article{Lee2015InvariantExploration,
  author  = {Jae Young Lee and Jin Bae Park and Yoon Ho Choi},
  title   = {Integral Reinforcement Learning for Continuous-Time Input-Affine Nonlinear Systems With Simultaneous Invariant Explorations},
  journal = {IEEE Transactions on Neural Networks and Learning Systems},
  volume  = {26},
  number  = {5},
  pages   = {916--932},
  year    = {2015},
  doi     = {10.1109/TNNLS.2014.2328590}
}

@article{AbuKhalaf2005SaturatingActuatorsHJB,
  author  = {Murad Abu-Khalaf and Frank L. Lewis},
  title   = {Nearly Optimal Control Laws for Nonlinear Systems With Saturating Actuators Using a Neural Network {HJB} Approach},
  journal = {Automatica},
  volume  = {41},
  number  = {5},
  pages   = {779--791},
  year    = {2005},
  doi     = {10.1016/j.automatica.2005.01.034}
}

@article{Vamvoudakis2010OnlineActorCritic,
  author  = {Kyriakos G. Vamvoudakis and Frank L. Lewis},
  title   = {Online Actor--Critic Algorithm to Solve the Continuous-Time Infinite Horizon Optimal Control Problem},
  journal = {Automatica},
  volume  = {46},
  number  = {5},
  pages   = {878--888},
  year    = {2010},
  doi     = {10.1016/j.automatica.2009.11.011}
}

@unpublished{VuEtAl2026FixedTimeMASIRL,
  author = {Tien Dat Vu and Minh Doan},
  title  = {Fixed-Time Integral Reinforcement Learning for Saturated Nonlinear Multi-Agent Systems Under FDI Attacks},
  note   = {Manuscript submitted to IEEE Transactions on Control Systems Technology, Manuscript No. 26-1065.2},
  year   = {2026}
}

@article{WilliamsKevrekidisRowley2015,
  author  = {Williams, Matthew O. and Kevrekidis, Ioannis G. and Rowley, Clarence W.},
  title   = {A Data-Driven Approximation of the Koopman Operator: Extending Dynamic Mode Decomposition},
  journal = {Journal of Nonlinear Science},
  volume  = {25},
  number  = {6},
  pages   = {1307--1346},
  year    = {2015}
}

@article{BruntonEtAl2016,
  author  = {Brunton, Steven L. and Brunton, Bingni W. and Proctor, Joshua L. and Kutz, J. Nathan},
  title   = {Koopman Invariant Subspaces and Finite Linear Representations of Nonlinear Dynamical Systems for Control},
  journal = {PLoS ONE},
  volume  = {11},
  number  = {2},
  pages   = {e0150171},
  year    = {2016}
}

@article{KordaMezic2018,
  author  = {Korda, Milan and Mezi{\'c}, Igor},
  title   = {Linear Predictors for Nonlinear Dynamical Systems: Koopman Operator Meets Model Predictive Control},
  journal = {Automatica},
  volume  = {93},
  pages   = {149--160},
  year    = {2018}
}

@inproceedings{KrolickiSutavaniVaidya2022,
  author    = {Krolicki, A. and Sutavani, S. and Vaidya, U.},
  title     = {Koopman-Based Policy Iteration for Robust Optimal Control},
  booktitle = {Proceedings of the American Control Conference},
  address   = {Atlanta, GA, USA},
  pages     = {1317--1322},
  year      = {2022}
}

@article{ChenLewisXie2024,
  author  = {Chen, C. and Lewis, Frank L. and Xie, K. and Xie, S.},
  title   = {Adaptive Optimal Control of Unknown Nonlinear Systems via Homotopy-Based Policy Iteration},
  journal = {IEEE Transactions on Automatic Control},
  volume  = {69},
  number  = {5},
  pages   = {3396--3403},
  year    = {2024}
}

@article{PanYu2016CompositeLearning,
  author  = {Yongping Pan and Haoyong Yu},
  title   = {Composite Learning From Adaptive Dynamic Surface Control},
  journal = {IEEE Transactions on Automatic Control},
  year    = {2016},
  volume  = {61},
  number  = {9},
  pages   = {2603--2609},
  doi     = {10.1109/TAC.2015.2495232}
}

@article{ParikhKamalapurkarDixon2019ICL,
  author  = {Anup Parikh and Rushikesh Kamalapurkar and Warren E. Dixon},
  title   = {Integral Concurrent Learning: Adaptive Control with Parameter Convergence Using Finite Excitation},
  journal = {International Journal of Adaptive Control and Signal Processing},
  year    = {2019},
  volume  = {33},
  number  = {12},
  pages   = {1775--1787},
  doi     = {10.1002/acs.2945}
}

\raggedbottom

% ------------------------------------------------------------
% Compact IEEE biography spacing
% IEEEtran default is 4\baselineskip, which may force the
% next biography to a new page unnecessarily.
% ------------------------------------------------------------
\makeatletter
\def\@IEEEBIOskipN{0.5\baselineskip}
\makeatother

\vspace{0.3em}
\section*{Author Biographies}
\vspace{-0.6em}

% ------------------------------------------------------------
% Tien Dat Vu
% ------------------------------------------------------------

\begin{IEEEbiography}
[{\includegraphics[
    width=1in,
    height=1.25in,
    clip,
    keepaspectratio
]{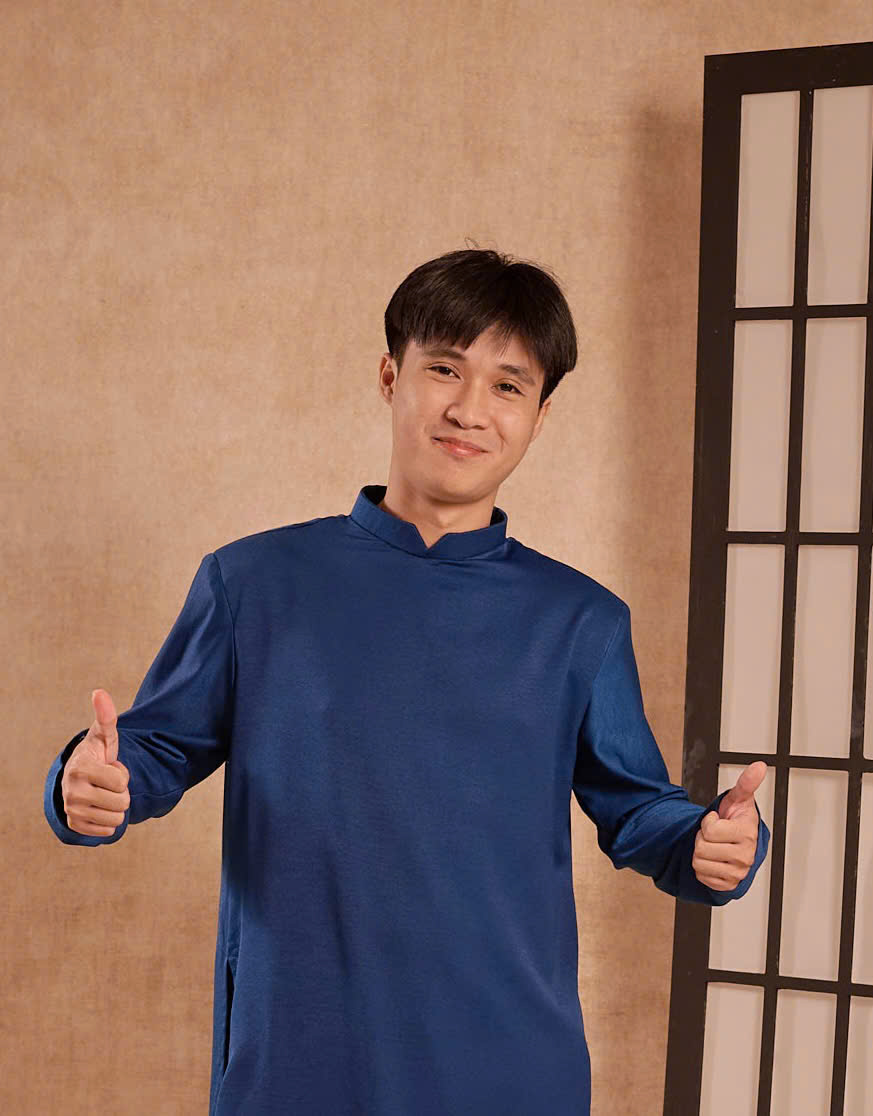}}]
{Tien Dat Vu}
is currently a senior undergraduate student in the
Vietnamese--French Program in Mechatronics Engineering at
Ho Chi Minh City University of Technology (HCMUT),
Vietnam National University Ho Chi Minh City (VNU-HCM),
Ho Chi Minh City, Vietnam. His research interests include
learning-based and nonlinear control, adaptive dynamic
programming, reinforcement learning, data-driven control,
secure and resilient control, optimal control of uncertain
nonlinear systems, multi-agent systems, and artificial
intelligence for dynamical systems.
\end{IEEEbiography}

% ------------------------------------------------------------
% Nhat Minh Doan
% ------------------------------------------------------------

\begin{IEEEbiography}
[{\includegraphics[
    width=1in,
    height=1.25in,
    clip,
    keepaspectratio
]{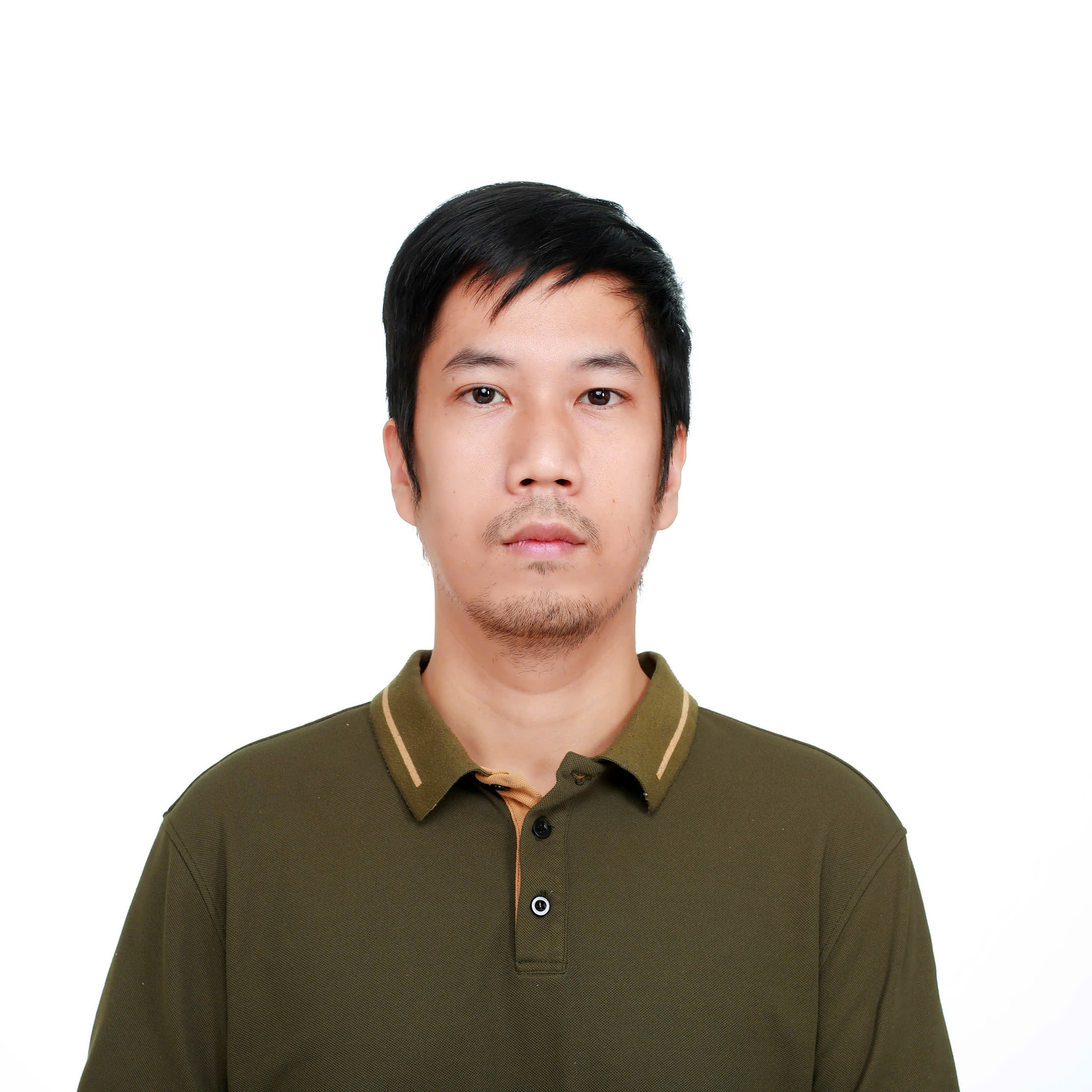}}]
{Nhat Minh Doan}
received the B.S. degree in Mechanical Engineering from
Bucknell University, Lewisburg, PA, USA, in 2015, and the
Ph.D. degree from Keio University, Tokyo, Japan, in 2021.
He is currently a Lecturer with the Faculty of Mechanical
Engineering, Ho Chi Minh City University of Technology
(HCMUT), Vietnam National University Ho Chi Minh City
(VNU-HCM), Ho Chi Minh City, Vietnam. His research interests
include unmanned aerial vehicle design and control, wind
turbine systems, multi-agent systems, networked and distributed
control, and the modeling, analysis, and control of complex
mechanical and networked dynamical systems.
\end{IEEEbiography}

\end{document}